\documentclass[12pt]{article}

\usepackage{color} 

\usepackage{amsfonts}
\usepackage{amsmath} 
\usepackage{amssymb}
\usepackage{amsthm}
\usepackage{graphicx}

\usepackage{hyperref}

\newtheorem{theorem}{Theorem}[section]                                          
                          
\newtheorem{lemma}[theorem]{Lemma}
\newtheorem{corollary}[theorem]{Corollary}

\newtheorem{definition}{Definition}[section]
\newtheorem{remark}{Remark}[section]

\newcommand{\ketbra}[2]{\left| #1\right\rangle\left\langle #2\right|}
\newcommand{\Tr}{\hbox{\rm tr}}

\begin{document}
%%%%%%%%%%%%%%%%%%%%%%%%%%%%%%%%%%%%%%%%

\title{Supercritical Poincar\'e-Andronov-Hopf bifurcation in a mean field quantum laser equation}

\author{ F. Fagnola  
\thanks{Dipartimento di Matematica, Politecnico di Milano, 
I-20133, Milano, Italy.  e-mail: franco.fagnola@polimi.it } 
\and 
C.M. Mora 
\thanks{Departamento de Ingenier\'{\i}a Matem\'{a}tica, Universidad de Concepci\'on, 
4089100, Casilla 160-C, Concepci\'on, Chile.
e-mail: cmora@ing-mat.udec.cl}
\thanks{
Supported by project VRID-Enlace 218.013.043-1.0 of Universidad de Concepci\'on.}
}

\maketitle

\begin{abstract}
We deal with the dynamical system properties of a Gorini-Kossakowski-Sudarshan-Lindblad (GKSL) equation 
with mean-field Hamiltonian that models a simple laser by applying a mean field approximation to 
a quantum system describing a single-mode optical cavity and a set of two level atoms,
each coupled to a reservoir.
We prove that the mean field quantum master equation
has a unique regular stationary solution.
In case a relevant parameter $C_\mathfrak{b} $, 
i.e., the cavity cooperative parameter, 
is less than $1$, 
we prove that any regular solution converges exponentially fast to the equilibrium,
and so the regular stationary state is a globally asymptotically stable equilibrium solution.
We obtain that a locally exponential stable limit cycle 
is born at the regular stationary state as $C_\mathfrak{b} $
passes through the critical value $1$.
Then,
the  mean-field laser equation  has a Poincar\'e-Andronov-Hopf bifurcation at $C_\mathfrak{b} =1 $ of supercritical-like type.
Namely,  
we derive  rigorously, at the level of density matrices --for the first time--, 
the transition from a global attractor quantum state, where the light is not emitted,
to a locally stable set of coherent quantum states producing coherent light.
Moreover,
we establish the local exponential stability of the limit cycle
in case a relevant parameter is between the first and second laser thresholds appearing in the semiclassical laser theory.
Thus,
we get that the coherent laser light  persists over time under this condition. 
In order to prove the exponential convergence of the quantum state, as the time goes to $+ \infty$,
we develop a new technique for proving the exponential convergence in open quantum systems
that is based in a new variation of constant formula,
which is obtained by combining probabilistic techniques with classical arguments from the semigroup theory.  
Furthermore, 
applying our main results we find the long-time behavior of the von Neumann entropy, the photon-number statistics,
and the quantum variance of the quadratures.
\end{abstract}

\vspace{2pc}

\vspace{1pc}
\noindent{\it AMS}: 
37L10, 37L05, 37L15, 37A60, 47A55, 60H30, 81S22, 82C10.

%\pacs{42.65.Sf, 42.55.Ah, 03.65.Yz, 05.45.-a, 02.30.Oz, 02.30.Tb, 02.50.Fz} 
%

%%%%%%%%%%%%%%%%%%%%%%%%%%%%%%%%%%%%%%%%

\section{Introduction}
\label{sec:Introduction}

%%%%%%%%%%%%%%%%%%%%%%%%%%%%%%%%%%%%%%%%%

In this paper we prove rigorously the occurrence of a supercritical Poincar\'e-Andronov-Hopf bifurcation  (Hopf bifurcation for short),
at the level of density operators, in a full quantum laser model.
Thus,
we develop the understanding of the dynamical systems properties of 
the infinite-dimensional open quantum systems.

% Modelo 
We study a laser composed of  many identical  two-level atoms with transition frequency $\omega$, as a gain medium,
that interact with an electromagnetic field, with resonance frequency $\omega$, propagating in one direction
(see, e.g., \cite{BreuerPetruccione,Ohtsubo2013,TartwijkAgrawal1998}).
The atoms make spontaneously  downward and upward transitions at rates 
$\kappa_-$ and $\kappa_+$, respectively.
The photons  leave the resonant mode of the radiation field at rate $2 \kappa$
due to the light output, together with losses in the resonator.
Under the mean field approximation, 
as the number  of non-interacting two-level atoms goes to $\infty$,
the laser evolution is described by 
the following effective Gorini-Kossakowski-Sudarshan-Lindblad  (GKSL for short) equation 
\begin{align}
 \label{eq:Laser1}
 \frac{d }{dt} \rho_t
& =
 -\mathrm{i}  \left[ \mathcal{H} \left( \rho_t \right) ,  \rho_t \right] 
+  
\kappa \left( 2 \,  a\,  \rho_t a^\dagger - 
a^\dagger a  \rho_t -  \rho_t a^\dagger a\right) 
\\
\nonumber
&   \quad
+  \frac{\kappa_-}{2} \left(  2 \, \sigma^{-}  \rho_t \,\sigma^{+} 
- \sigma^{+} \sigma^{-}  \rho_t 
-  \rho_t \,\sigma^+\sigma^{-}\right) 
+ \frac{ \kappa_+ }{ 2 } \left( 2 \,  \sigma^+  \rho_t \,\sigma^{-} 
- \sigma^{-}\sigma^+  \rho_t
- \rho_t \,\sigma^{-}\sigma^+\right)
\end{align}
having the  mean-field Hamiltonian
\[
\mathcal{H} \left( \varrho \right) 
 = 
\omega \left( a^\dagger a + \frac{1}{2} \sigma^{3} \right) 
+ \mathrm{i} \, g 
 \left(
 \Bigl(  \Tr \left( \sigma^{-}  \varrho \right) a^{\dagger}  -  \Tr \left( \sigma^{+}  \varrho  \right) a \Bigr)
  + 
\Bigl(  \Tr\left( a^{\dagger}  \varrho  \right) \sigma^{-}  -  \Tr\left( a \,  \varrho  \right) \sigma^{+} \Bigr)   
  \right) 
\]
(see, e.g., \cite{BreuerPetruccione,Spohn1980,Mori2013}),
where the constant $g \in \mathbb{R} \setminus \left\{ 0 \right\}$ characterizes the coupling between atoms and the field mode.
Here,
the unknown 
$\rho_t$ is a non-negative trace-class operator on $\ell^2 \left(\mathbb{Z}_+ \right) \otimes \mathbb{C}^2 $, 
$\omega$ is a real number, $\kappa,\kappa_+,\kappa_-  \in \left] 0, + \infty \right[$,
$
\sigma^{+} =
\begin{pmatrix}
 0 & 1 \\ 0 & 0
\end{pmatrix} 
$,
$
\sigma^{-} =
\begin{pmatrix}
 0 & 0 \\ 1 & 0
\end{pmatrix}
$,
$
\sigma^{3} =
\begin{pmatrix}
 1 & 0 \\ 0 & -1
\end{pmatrix} 
$,
and the  closed operators $a^{\dagger}$, $a$ on $\ell^2 \left(\mathbb{Z}_+ \right)$
are defined by
$a^{\dagger}e_{n} = \sqrt{ n+1} \, e_{n+1}$ for all $n \in   \mathbb{Z}_{+}$
and 
$
ae_{n} 
=
\begin{cases}  
\sqrt{n} \, e_{n-1}  &  \text{ if }   n \in  \mathbb{N} 
 \\ 0  &  \text{ if } n = 0  
\end{cases} 
$,
where
$(e_n)_{n\ge 0}$ denotes the standard basis of $\ell^2(\mathbb{Z_+})$.
The non-linear quantum master equation \eqref{eq:Laser1}
reproduces the Dicke-Haken-Lax model of the laser,
and can be formally obtained from the dissipative Tavis-Cummings 
model  governing the unidirectional ring-cavity laser with  $n_{atom}$  atoms
by taking the limit as $n_{atom} \rightarrow \infty$ of the partial trace with respect to  $n_{atom}-1$ atoms of the full density operator
(see, e.g., Section 3.7.3 of \cite{BreuerPetruccione}, Section V.E of \cite{Spohn1980}, and \cite{Mori2013}).

% Maxwell-Bloch equations - dynamics 
We are interested in investigating the long-term stable behavior of the solution to \eqref{eq:Laser1}.
Lasers can show stable or unstable behaviors according to the operating conditions.
Numerous investigations on the qualitative properties of  the laser dynamics
have been devoted essentially to the application of the  linear stability analysis to
complex ordinary differential equations describing  the expectation values of some quantum observables
like rate equations and semiclassical laser models
(see, e.g., \cite{AlliSewell1995,erneux_glorieux_2010,Haken1985,HeppLieb1973,Khanin2006,NingHaken1990,Ohtsubo2013,TartwijkAgrawal1998}).
In the physical situation under consideration,
using \eqref{eq:Laser1} we obtain that 
$ A \left( t \right) = \Tr\left(  \rho_{t}  \, a  \right) $,
$ S \left( t  \right)  = \Tr\left(  \rho_{t} \, \sigma^{-}  \right) $
and
$ D \left( t  \right) = \Tr\left(  \rho_{t} \,  \sigma^{3} \right) $
satisfy
\begin{equation}
 \label{eq:Lorenz}
 \left\{ 
 \begin{aligned} 
  \frac{d}{dt} A \left( t \right)
 & = 
 - \left( \kappa + \mathrm{i} \, \omega \right) A \left( t \right)  + g \ S \left( t \right) 
 \\
  \frac{d}{dt} S \left( t  \right)
&  =  
 - \left( \gamma + \mathrm{i} \, \omega \right)   S \left( t  \right)
 + g \ A  \left( t \right)   D \left( t  \right)
 \\
 \frac{d}{dt} D \left( t  \right)
&  =  
-  4 g \ \Re \left(
 \overline{ A  \left( t \right)}   S \left( t  \right) 
 \right)
 - 2 \gamma \left(  D \left( t  \right) - d \right) 
\end{aligned}
  \right. ,
\end{equation}
where 
$
d = \left( \kappa_+ - \kappa_- \right) / \left(  \kappa_- + \kappa_+ \right) 
$
and
$\gamma = \left(  \kappa_-  + \kappa_+  \right)/2$
 (see, e.g., \cite{FagMora2019}).
 In \cite{AlliSewell1995,HeppLieb1973},
versions of (\ref{eq:Lorenz})  are derived by taking limit in many body linear quantum master equations
(see, e.g., \cite{Bagarello2002} for a study of the relation between the models considered in \cite{AlliSewell1995} and \cite{HeppLieb1973}).
 In the semiclassical laser theory,
\eqref{eq:Lorenz} describes the dynamics of 
the field, polarization and population inversion 
(i.e., $\Tr\left( \rho_{t}   \, a \right)$, $\Tr\left( \rho_{t} \, \sigma^{-}  \right)$ and $ \Tr\left(  \rho_{t} \, \sigma^{3}  \right)$,
respectively)
of ring lasers such as  far-infrared $NH_3$ lasers  (see, e.g., \cite{Haken1985,Ohtsubo2013,TartwijkAgrawal1998}). 
The Maxwell-Bloch equations \eqref{eq:Lorenz} develop a stable set of periodic solutions 
from the stable fixed point $\left( 0, 0, d \right)$ as  the cavity cooperative parameter 
\[
C_\mathfrak{b} 
:= 
 \frac{g^2 \, d}{\kappa \, \gamma}
 =
 \frac{ 2 g^2 \left( \kappa_+ - \kappa_- \right) }{ \kappa \left( \kappa_-  + \kappa_+  \right)^2 }
\]
crosses $1$
(see, e.g., \cite{AlliSewell1995,BreuerPetruccione,Fowler1982,HeppLieb1973,NingHaken1990}),
and so 
\eqref{eq:Lorenz} undergoes a supercritical Hopf bifurcation at $C_\mathfrak{b} =1 $.

% Justificaci\'on de estudio de full quantum laser model and trabajos previos
Unlike semiclassical models,
quantum master equations in GKSL form, and their mean-field approximations,
describe the quantum mechanical properties, not only mean values, of both the atoms and the light fields,
and hence they capture very well quantum effects like 
coherence, correlations, spontaneous emissions and photon-number statistics (see, e.g., \cite{Haken1985}).
This motivates the study of the dynamical properties of the evolution of density operators
representing laser states.
In this direction,
numerical studies of  the bifurcation structure of the steady state 
of quantum master equations in GKSL form
have been carried out by, e.g., \cite{ArmenMabuchi2006,Ivanchenko2017,Meaney2010,YusipovIvanchenko2019}.
In a different physical context like mirrorless lasers,
the superradiance phase transitions has been studied in depth 
(see, e.g., \cite{BreuerPetruccione,HarocheRaimond2006,HeppLiebAP1973}).

% Resultados del art\'iculo, F\'isica

In this paper, we establish rigorously
the qualitative changes in the dynamics of 
the solution to the mean field laser equation \eqref{eq:Laser1},
at the level of density matrices,
as the parameter 
$
C_\mathfrak{b} 
$
passes through the critical value $1$.
We present the first mathematical proof --to the best of our knowledge--
of  a supercritical Hopf bifurcation in infinite dimensional GKSL-like equations,
and at the same time
we get
the transition from a global attractor state where the light is not emitted
to a locally stable set of coherent states producing coherent light.
Moreover, 
we prove that the difference between $\rho_t$ and a certain periodic function of  coherent states
converges exponentially fast  to  $0$ as $t \rightarrow + \infty$ 
whenever
$ \kappa^2 +  5 \kappa \gamma >  \gamma \left( \kappa - 3 \gamma \right) C_\mathfrak{b}$
and 
$\rho_0$ is in a neighborhood of certain coherent states.
Thus,
the coherent laser light  persists over time if $C_\mathfrak{b}$ is between the first and second laser thresholds.
The above two physical phenomena are explained in the semiclassical laser theory 
by the Maxwell-Bloch equations \eqref{eq:Lorenz},
but a full quantum foundation was not yet given.
From the dynamical systems viewpoint,
\eqref{eq:Laser1} is a  model problem for understanding the behavior  of the mean-field GKSL master equations,
which generate  non-linear quantum dynamical semigroups (see, e.g., \cite{AlickiLendi2007,Kolokoltsov2010,MerkliBerman2012}).
Indeed,
\eqref{eq:Laser1} could play a role in open quantum systems
similar to the one played by the Lorenz equations in finite-dimensional dynamical systems.

In our analysis, first we show that
\begin{equation}
 \label{eq:I10}
 \varrho_{\infty} :=
 \ketbra{e_0}{e_0}
\otimes 
\left( \frac{d+1}{2}  \ketbra{e_+}{e_+}  + \frac{1-d}{2}   \ketbra{e_-}{e_-}  \right)
\end{equation}
is the unique $N$-regular stationary state for (\ref{eq:Laser1}) with $\omega \neq 0$, 
the physical situation we are interested in. 
That is, 
if $\omega \neq 0$, then 
(\ref{eq:I10}) is  the unique density operator $\varrho$  such that 
$\rho_t \equiv \varrho$ satisfies (\ref{eq:Laser1}) and $\varrho$  is $N$-regular, 
which means, roughly speaking, that
the trace of  $ a^\dagger a \, \varrho $ is well defined
(see Section \ref{subsec:not} for the definition of $N$-regular density operator).
This invariant solution yields 
the unique stationary solution of (\ref{eq:Lorenz}).
In case $ C_\mathfrak{b} < 1$ and $\omega \neq 0$,
we obtain that $\rho_t$ converges in the trace norm exponentially fast to $\varrho_{\infty}$ as $t \rightarrow + \infty$,
and hence
$\varrho_{\infty}$ is the global attractor for \eqref{eq:Laser1}.
In the state $ \varrho_{\infty}$ the light is not emitted,
and hence we  quickly perceive a faint light output 
% the laser beam fades away quickly 
when the normalized pump parameter  $d$  is below $ \kappa \, \gamma / g^2 $.

Second,
we  consider the free interaction solutions to \eqref{eq:Laser1} with $\omega \neq 0$,
that is,
the solutions of \eqref{eq:Laser1} 
that also satisfy the Liouville-Von Neumann equation 
\begin{equation*}
% \label{eq:I2}
\frac{d }{dt} \rho_t
 = - \mathrm{i}  \, \omega
\left[  a^\dagger a +\sigma^{3} / 2, \rho_t  \right] ,
\end{equation*}
which describes 
the evolution of the physical system
in absence of interactions between the laser mode, atoms and the bath.
If $ C_\mathfrak{b} \leq 1$,
then we deduce that $\varrho_{\infty}$ is the unique $N$-regular free interaction solution to (\ref{eq:Laser1}).
In case 
the cavity cooperative parameter $C_\mathfrak{b}$ is greater than $1$ and $\omega \neq 0$,
we obtain that all the non-constant $N$-regular free interaction solutions are of the form
$ t \mapsto \varrho_c \left( \omega t - \theta \right)$ for  any $\theta \in \left[ 0, 2 \pi \right[$,
where for each $\vartheta \in \mathbb{R}$ we set
\begin{align}
\label{eq:Def_orbita}
& 
\varrho_c \left( \vartheta \right)
=
\\ \nonumber
&
\ketbra{  \mathcal{E} \left(  \frac{ \gamma \sqrt{ C_\mathfrak{b} -1 } }{ \sqrt{2} \left| g \right|} 
\hbox{\rm e}^{- \mathrm{i} \vartheta}   \right)} 
{  \mathcal{E} \left(\frac{ \gamma \sqrt{ C_\mathfrak{b} -1 }}{ \sqrt{2} \left| g \right| }  
\hbox{\rm e}^{-\mathrm{ i } \vartheta}   \right)}
\otimes 
\begin{pmatrix}
 \frac{1}{2} \left( 1 + \frac{ d }{ C_\mathfrak{b} }  \right)
 &
 \hbox{\rm e}^{- \mathrm{ i } \vartheta}  \frac{ \kappa \gamma }{ \sqrt{2} g \left| g \right|} \sqrt{ C_\mathfrak{b} -1 }
 \\
 \hbox{\rm e}^{ \mathrm{i} \vartheta }  \frac{ \kappa \gamma }{ \sqrt{2} g \left| g \right|} \sqrt{ C_\mathfrak{b} -1 }
 &
 \frac{1}{2} \left( 1 - \frac{d }{ C_\mathfrak{b} }  \right) 
\end{pmatrix} ,
\end{align}
the coherent vector $ \mathcal{E} \left( \zeta \right) $ associated with $\zeta \in \mathbb{C}$ is defined by 
\begin{equation}
\label{eq:Coherent-Vector}
 \mathcal{E} \left( \zeta \right) 
=
\exp \left(  - \left\vert \zeta \right\vert^2 / 2 \right)
\sum_{n = 0}^{+ \infty} \zeta^n e_n / \sqrt{ n !} ,
\end{equation}
and $(e_n)_{n\ge 0}$ stands for the canonical orthonormal basis of $\ell^2(\mathbb{Z_+})$.
Thus,
the laser emits coherent light when the normalized pump parameter  $d$  exceeds the barrier  $ \kappa \, \gamma / g^2 $,
and the periodic solutions  $ t \mapsto \varrho_c \left( \omega t - \theta \right)$ yield  periodic solutions of \eqref{eq:Lorenz}.
In the dynamical system language,
the phase path of all non-constant $N$-regular free interaction solutions of  (\ref{eq:Laser1}) 
gives the closed orbit $\left\{ \varrho_c \left( \vartheta \right): \vartheta \in \left[ 0, 2 \pi \right] \right\}$.

%\begin{equation}
%\label{eq:Def_CL}
% \begin{aligned}
%\rho_t^{\theta}
%& := 
%\ketbra{  \mathcal{E} \left(  \frac{ \gamma \sqrt{ C_\mathfrak{b} -1 } }{ \sqrt{2} \left| g \right|} 
%\hbox{\rm e}^{- \mathrm{i} \left( \omega t - \theta  \right)}   \right)} 
%{  \mathcal{E} \left(\frac{ \gamma \sqrt{ C_\mathfrak{b} -1 }}{ \sqrt{2} \left| g \right| }  
%\hbox{\rm e}^{-\mathrm{ i } \left( \omega t - \theta  \right)}   \right)}
%\otimes 
%\\
%& \hspace{3cm}
%\begin{pmatrix}
% \frac{1}{2} \left( 1 + \frac{ d }{ C_\mathfrak{b} }  \right)
% &
% \hbox{\rm e}^{- \mathrm{ i } \left( \omega t - \theta  \right)}  \frac{ \kappa \gamma }{ \sqrt{2} g \left| g \right|} \sqrt{ C_\mathfrak{b} -1 }
% \\
% \hbox{\rm e}^{ \mathrm{i} \left( \omega t - \theta  \right) }  \frac{ \kappa \gamma }{ \sqrt{2} g \left| g \right|} \sqrt{ C_\mathfrak{b} -1 }
% &
% \frac{1}{2} \left( 1 - \frac{d }{ C_\mathfrak{b} }  \right) 
%\end{pmatrix} 
%\end{aligned}
%\end{equation}
%for all $t \geq 0$,
%where $\theta$ is any real number belonging to $ \left[ 0, 2 \pi \right[$,

We prove that the  cycle $\left\{ \varrho_c \left( \vartheta \right): \vartheta \in \left[ 0, 2 \pi \right] \right\}$
is locally exponential stable  whenever $\kappa \leq 3 \gamma $  (the cavity is not too lossy),
or
$\kappa > 3 \gamma $ with  
$
\left(
\kappa^2 +  5 \kappa \gamma
\right)
/
\left(
\gamma \left( \kappa - 3 \gamma \right) 
\right)
>  
C_\mathfrak{b} 
$,
and so 
$\left\{ \varrho_c \left( \vartheta \right): \vartheta \in \left[ 0, 2 \pi \right] \right\}$
is an attractive limit cycle in the phase space if $C_\mathfrak{b} >1 $ is close to $1$.
Hence,
\eqref{eq:Laser1} has a Hopf bifurcation at $C_\mathfrak{b} =1 $ of supercritical-like type.
As far as we know, 
this is the first time that Hopf bifurcation
is rigorously established at the level of (infinite dimensional) density matrices 
in the study of nonlinear evolutions of open quantum systems.
The bad-cavity condition $\kappa > 3 \gamma $,
which is paraphrased as the relaxation time of the atoms  is greater than three times the relaxation time of the field,
takes place in lasers of type C (see, e.g., \cite{Khanin2006,Ohtsubo2013,TartwijkAgrawal1998}).
In this case,
we have proved that the laser beam is stable when  the normalized pump parameter  is in the interval
\[
\frac{\kappa \, \gamma}{g^2 } <
d
<  
\frac{ \kappa^3 +  5 \kappa^2 \gamma}{  g^2 \left( \kappa - 3 \gamma \right)} .
\]
If $C_\mathfrak{b}$ is beyond the second threshold
$
\left(
\kappa^2 +  5 \kappa \gamma
\right)
/
\left(
\gamma \left( \kappa - 3 \gamma \right) 
\right)
$,
then
the set of known periodic solutions of \eqref{eq:Lorenz} 
loses its stability.

Third,
the mean values and quantum fluctuations of unbounded observables like quantum quadratures 
provide important information about the laser behavior.
We study the long time behavior of the unbounded operators $A$ 
that are relatively bounded with respect to the number operator $a^\dagger a$.
If $ C_\mathfrak{b} < 1$,
then we get the exponential convergence of the mean value of $A$
to the trace of $A \, \varrho_{\infty}$ as the time goes to $+ \infty$.
In case  $ C_\mathfrak{b} > 1$ and $ \kappa^2 +  5 \kappa \gamma >  \gamma \left( \kappa - 3 \gamma \right) C_\mathfrak{b}$
we prove that
$
\Tr \left(  \rho_t  \, A \right)
 -
 \Tr \left(   \varrho_c \left( \omega t - \theta \right) A \right)
$
converges exponentially fast to $0$ as $t \rightarrow + \infty$,
for certain $\theta \in \left[ 0, 2 \pi \right[$,
whenever 
$ \rho_0$ is close enough to the limit cycle  $\left\{ \varrho_c \left( \vartheta \right): \vartheta \in \left[ 0, 2 \pi \right] \right\}$.
Thus,
we determine how the full quantum dynamics described by \eqref{eq:Laser1}
leads to the occurrence of the supercritical-like Hopf bifurcation  in \eqref{eq:Lorenz} at $C_\mathfrak{b} =1$.
In addition,
we characterize, for instance,  the long time behavior of
the photon-number statistics, the quantum variance of the quadratures, and the von Neumann entropy.

% Resultados del art\'iculo, Matem\'atica
In \cite{FagMora2019}  we prove that the mean field quantum laser equation (\ref{eq:Laser1}) 
has a unique $N$-regular solution
--in a weak sense--, and we obtain (\ref{eq:Lorenz}) from (\ref{eq:Laser1}).
To this end, in \cite{FagMora2019} we get the existence and uniqueness of the $N$-regular solution to
the non-homogeneous GKSL equation 
\begin{equation}
\label{eq:AuxiliarGKSL}
\begin{aligned}
 \frac{d }{dt} \rho_t
& =
 -\mathrm{i}  \left[ \omega \left( a^\dagger a + \frac{1}{2} \sigma^{3} \right)  ,  \rho_t \right] 
+
\left[ 
\alpha \left( t \right) a^\dagger -  \overline{\alpha \left( t \right)} a 
+ \overline{\beta  \left( t \right) } \sigma^{-}  - \beta  \left( t \right) \sigma^+  ,
\rho_t \right] 
\\
&   \quad
+  
\kappa \left( 2 \,  a\,  \rho_t a^\dagger - 
a^\dagger a  \rho_t -  \rho_t a^\dagger a\right) 
+  \frac{\kappa_-}{2} \left(  2 \, \sigma^{-}  \rho_t \,\sigma^{+} 
- \sigma^{+} \sigma^{-}  \rho_t 
-  \rho_t \,\sigma^+\sigma^{-}\right) 
\\
&   \quad
+ \frac{ \kappa_+ }{ 2 } \left( 2 \,  \sigma^+  \rho_t \,\sigma^{-} 
- \sigma^{-}\sigma^+  \rho_t
- \rho_t \,\sigma^{-}\sigma^+\right) ,
\end{aligned}
\end{equation}
where 
$\alpha, \beta : \left[ 0 , \infty \right[ \rightarrow \mathbb{C}$
are continuous
and
$ \rho_t \in \mathfrak{L}_{1}^{+}\left( \ell^2(\mathbb{Z}_+)\otimes \mathbb{C}^2 \right)$,
as well as  we derive the equation of motions of the mean values of $a$, $\sigma^{-}$ and $ \sigma^{3}$
with respect to the $N$-regular solution to  (\ref{eq:AuxiliarGKSL}).
This study is based on the stochastic Schr\"odinger equations 
(see, e.g., \cite{Barchielli,Benoist2020,BreuerPetruccione,MoraReAAP}),
which provide probabilistic representations --unravelings-- of (\ref{eq:AuxiliarGKSL}).
Thus,
in \cite{FagMora2019} we deduce that the $N$-regular solution of  (\ref{eq:Laser1}) 
coincides with  the $N$-regular solution of  (\ref{eq:AuxiliarGKSL})  with
$\alpha \left( t \right) = g \,  S \left( t  \right)$ and $\beta \left( t \right) = g \,  A\left( t  \right)$.
In the  current paper,
we obtain  dynamical systems properties of  (\ref{eq:Laser1})
by treating the long-time behavior of  (\ref{eq:AuxiliarGKSL}) coupled to (\ref{eq:Lorenz})
via $\alpha \left( t \right) = g \,  S\left( t  \right)$ and $\beta \left( t \right) = g \,  A\left( t  \right)$.
For this purpose,
we develop a new variation of constant formula for (\ref{eq:AuxiliarGKSL}),
which is proved by combining classical arguments from the semigroup theory
with an analysis involving the linear stochastic Schr\"odinger equation (\ref{eq:SSE}) given below. 
Moreover,
we deduce the exponential convergence of the solution of  (\ref{eq:AuxiliarGKSL}) to its equilibrium state
in case $\alpha \left( t \right)$ and $\beta \left( t \right)$ are  constant functions.
To do this, we estimate, loosely speaking,
the rate of decoupling of the atoms and the electromagnetic field,
as well as we obtain the exponential convergence of the atoms and the field to their equilibrium states
when we neglect the interaction between them. 
Then,
we prove the exponential convergence of the solution of  (\ref{eq:Laser1}) to its invariant sets 
by means of perturbation techniques 
applied to (\ref{eq:AuxiliarGKSL}) coupled to (\ref{eq:Lorenz}),
which is a new way  to handle the long-time behavior of open quantum systems.
In this analysis we use  a unitary transformation of  (\ref{eq:Laser1}) to treat the limit cycle of (\ref{eq:Lorenz}),
which leads to study  the asymptotic behavior of (\ref{eq:AuxiliarGKSL}) with $\omega = 0$.

We organize the article in three main sections.
Section \ref{sec:QuantumBifurcation} states the main results of this paper.
In Section \ref{sec:LinearQMEs} we address (\ref{eq:AuxiliarGKSL}).
Section \ref{sec:Proofs} presents the proofs of all theorems,
where 
we use the results given in Section \ref{sec:LinearQMEs} to prove 
the theorems stated in Section \ref{sec:QuantumBifurcation}.

\subsection{Notation}
\label{subsec:not}

As far as possible,
we use the same notation as in \cite{FagMora2019}.
Thus,
we consider a separable complex Hilbert space $\left(\mathfrak{h},\left\langle \cdot,\cdot\right\rangle \right) $,
whose scalar product $\left\langle \cdot,\cdot \right\rangle $ is anti-linear in the first variable and linear in the second one.
The canonical orthonormal basis of $\ell^2(\mathbb{Z_+})$ is denoted by $(e_n)_{n\ge 0}$,
as well as 
$
e_+ =
\begin{pmatrix}
 1 \\ 0
\end{pmatrix}
$
and 
$
e_- =
\begin{pmatrix}
 0 \\ 1
\end{pmatrix} 
$
is the standard basis of $ \mathbb{C}^2 $.
We write $\mathcal{D}\left(A\right)$ for the domain of $A$, whenever $A$ is a linear operator in $\mathfrak{h}$. 
As usual,
we set $ \left[ A,  B \right] = AB - BA$ in case $A,B$  are linear operators   in  $\mathfrak{h}$,
and $N= a^\dagger a $.
We write $\mathfrak{L}\left( \mathfrak{X},\mathfrak{Z}\right) $ for the space of 
all bounded operators from $\mathfrak{X}$ to $\mathfrak{Y}$,
where $\mathfrak{X}$ and $\mathfrak{Y}$ are normed spaces.
By $\mathfrak{L}\left( \mathfrak{X}\right) $ we mean $ \mathfrak{L}\left( \mathfrak{X},\mathfrak{X}\right) $. 
The space  of all trace-class operators on $\mathfrak{h}$,
with the trace norm,
is denoted by $\mathfrak{L}_{1}\left( \mathfrak{h}\right)$.

Suppose that the operator $C : \mathcal{D}\left(C\right) \subset \mathfrak{h} \rightarrow \mathfrak{h} $ 
is positive and self-adjoint. 
We recall that 
$\varrho \in \mathfrak{L}_1\left( \mathfrak{h} \right)$ is a density operator iff  $\varrho$ is a non-negative operator with unit trace.
A non-negative operator  $\varrho \in \mathfrak{L}\left( \mathfrak{h} \right)$ is called $C$-regular 
iff 
there exists $\lambda_{n} \geq 0$ and $ v_{n} \in \mathcal{D}\left( C\right) $, together with a countable set $\mathfrak{I}$,
such that
$
 \varrho=\sum_{n\in\mathfrak{I}}\lambda_{n}\left\vert v_{n}\rangle\langle v_{n}\right\vert
$,
$
 \sum_{n\in \mathfrak{I}} \left( \lambda_{n} \right) ^{2}<\infty
$,
and
$
 \sum_{n\in \mathfrak{I}}\lambda_{n}\left\Vert Cu_{n}\right\Vert ^{2}<\infty
$
(see, e.g., \cite{ChebGarQue98,FagMora2019,MoraAP}).
We write $\mathfrak{L}_{1,C}^{+} \left( \mathfrak{h}\right) $  for the family 
of all  density operators in $\mathfrak{h}$ that are $C$-regular.

Moreover,
for  any $x,y\in \mathcal{D}\left( C\right) $
we define the graph scalar product 
$\left\langle x,y\right\rangle_{C}=\left\langle x,y\right\rangle +\left\langle Cx,Cy\right\rangle $
 and the graph norm 
 $ \left\Vert x\right\Vert _{C}=
\sqrt{\left\langle x,x\right\rangle _{C}}$.
We use the symbol $L^{2}\left( \mathbb{P},\mathfrak{h}\right) $ to denote 
the space of all square integrable functions 
$ X : \left( \Omega ,\mathfrak{F},\mathbb{P}\right) \rightarrow \left( \mathfrak{h},\mathfrak{B}\left( \mathfrak{h}\right) \right)$,
where  $ \mathfrak{B} \left( \mathfrak{h} \right)$ is formed by all Borel set on $\mathfrak{h}$.
Moreover,  $L_{C}^{2}\left( \mathbb{P},\mathfrak{h}\right) $  stands for the set of all $\xi \in L^{2}\left( \mathbb{P},\mathfrak{h}\right) $ such that $\xi \in \mathcal{D}\left( C\right) $ a.s. and $\mathbb{E} \left( \left\Vert \xi \right\Vert _{C}^{2} \right) <\infty $. 
For any $x\in \mathcal{D}\left( C\right) $  we define $\pi_C(x)=x$,
together with $\pi_C(x)=0$ whenever $ x \in \mathfrak{h} \setminus  \mathcal{D}\left( C\right)$.

Recall that 
$\omega \in\mathbb{R}$, $ g \in \mathbb{R} \smallsetminus \left\{ 0 \right\}$,
and 
$\kappa,\kappa_+,\kappa_- > 0$.
Moreover,
in Section \ref{sec:Introduction} we take 
$\gamma = \left( \kappa_+ + \kappa_- \right)/2$,
$
d = \left( \kappa_+ - \kappa_- \right) / \left( \kappa_+ + \kappa_- \right) 
$,
and 
$
C_\mathfrak{b} = g^2 \, d / \left( \kappa \, \gamma \right) 
$.
Then $\kappa_- = \gamma \left(1-d \right)$, and  $\kappa_+ = \gamma \left(1+d \right)$.
Using $\kappa_-, \kappa_+ > 0$ we deduce that $\gamma > 0$ and $d\in \left]-1,1 \right[$.
In what follows, the letters $K \geq 0$ and $ \lambda > 0$ denote generic constants.  
We will write $K \left( \cdot \right)$ for different  non-decreasing 
non-negative functions on the interval $\left[ 0, \infty \right[$ when 
no confusion is possible.

%%%%%%%%%%%%%%%%%%%%%%%%%%%%%%%%%%%%%%%%%%%%%%%%%%

\section{Quantum Hopf bifurcation}
\label{sec:QuantumBifurcation}

%%%%%%%%%%%%%%%%%%%%%%%%%%%%%%%%%%%%%%%%%%%%%%%%%%

\subsection{Invariant sets}

We  begin by determining the stationary solutions to (\ref{eq:Laser1}).
We recall that a $C$-weak solution to (\ref{eq:Laser1}) is a collection of  $C$-regular density operators $\left( \rho_t  \right)_{t \geq 0}$ in 
$\ell^2 \left(\mathbb{Z}_+ \right) \otimes \mathbb{C}^2 $
such that $t \mapsto \Tr\left( a \rho_{t}  \right) $ is continuous
and
\[
\frac{d}{dt}\Tr\left( A \rho_{t}  \right) 
 = 
\Tr\left( A \left(  
\mathcal{L}_{\star}^h \, \rho_t 
+   g \left[ 
\Tr\left( \sigma^{-}  \rho_t  \right) a^\dagger  -  \Tr\left( \sigma^{+} \rho_t  \right) a 
+
 \Tr\left( a^\dagger  \rho_t  \right) \sigma^{-} -  \Tr\left( a \,  \rho_t \right) \sigma^+ 
, \rho_t \right] 
 \right)  \right) 
\]
for all $t \geq 0$
and 
$A \in\mathfrak{L}\left( \ell^2 \left(\mathbb{Z}_+ \right) \otimes \mathbb{C}^2 \right) $,
where
\begin{align}
 \label{eq:3.21}
& \mathcal{L}_{\star}^h \, \varrho
  =
 - \mathrm{i} \, \omega 
\left[ 
\left(  a^\dagger a  +\sigma^{3} / 2 \right)  , \varrho \right] 
+   \kappa \left( 2 \,  a\,\varrho a^\dagger -  a^\dagger a \varrho -  \varrho a^\dagger a\right) 
\\
\nonumber
& \quad      
 + \frac{ \gamma \left(1-d \right)  }{ 2} \left( 2 \, \sigma^{-} \varrho \,\sigma^+ 
- \sigma^+\sigma^{-}\varrho 
- \varrho \,\sigma^+\sigma^{-}\right) 
+  \frac{ \gamma \left(1+d \right) }{ 2 } \left( 2 \,  \sigma^+\varrho \,\sigma^{-} 
- \sigma^{-}\sigma^+\varrho 
- \varrho \,\sigma^{-}\sigma^+\right) .
\end{align}
According to \cite{FagMora2019} we have that (\ref{eq:Laser1})  has a unique $N^p$-weak solution,
as well as that 
the Maxwe  ll-Bloch equations (\ref{eq:Lorenz}) hold
whenever 
$\rho_0  \in \mathfrak{L}_{1,N^p}^{+} \left( \ell^2 \left(\mathbb{Z}_+ \right) \otimes \mathbb{C}^2 \right) $
with $p \in \mathbb{N}$.
Next,
we show that  (\ref{eq:Laser1}) has a unique $N$-regular invariant state whenever $\omega \neq 0$,
which yields the stationary solution of (\ref{eq:Lorenz}),
which is 
$  \Tr\left( a  \rho_t \right) = \Tr\left( \sigma^{-}  \rho_t  \right) = 0 $ and  $\Tr\left( \sigma^{3} \rho_t  \right) = d $
for all $t \geq 0$.

\begin{definition}
Consider a $C$-regular density operator $\varrho$.
We say that $\varrho$ is a stationary state for (\ref{eq:Laser1})
iff $ \rho_t \equiv  \varrho $ is a constant $C$-weak solution to (\ref{eq:Laser1}).
\end{definition}

\begin{theorem}
\label{th:StatState-LaserE}
Let the density operator $\varrho_{\infty}$ be defined by (\ref{eq:I10}).
Then $\varrho_{\infty}$ is a stationary state for  (\ref{eq:Laser1}).
Moreover, in case $\omega \neq 0$,
$\varrho_{\infty}$ is the unique $N$-regular density operator which is a stationary state for  (\ref{eq:Laser1}).
\end{theorem}

\begin{proof}
 Deferred to Section  \ref{sec:Proof:StatState-LaserE}.
\end{proof}

We turn our attention to the regular solutions of  (\ref{eq:Laser1}) 
that are also unitary evolutions generated by the Hamiltonian $  \omega  \left( N +  \sigma^3 / 2 \right)$,
which arises from neglecting the interactions between the laser mode, atoms and the bath.

\begin{definition}
Assume that $\left( \rho_t \right)_{t \geq 0}$ is a $C$-weak solution to (\ref{eq:Laser1}).
We call $\left( \rho_t \right)_{t \geq 0}$ free interaction solution to  (\ref{eq:Laser1}) if and only if 
\begin{equation*}
 \rho_t =
\exp \left( - \mathrm{i}  \omega  \left( N +  \sigma^3 / 2 \right) t \right)
\varrho_0
\exp \left(  \mathrm{i}  \omega  \left( N +  \sigma^3 / 2 \right) t \right)
\hspace{1cm}
\forall t \geq 0 .
\end{equation*}
\end{definition}

\begin{remark}
If  $\left( \rho_t \right)_{t \geq 0}$ is a $N$-regular free interaction solution to (\ref{eq:Laser1}),
then $\left( \rho_t \right)_{t \geq 0}$ also satisfies the  quantum master equation
$
 \frac{d }{dt} \rho_t
 = - \mathrm{i}  \frac{ \omega}{2}
\left[ 2 a^\dagger a +\sigma^{3}, \rho_t  \right] 
$.
\end{remark}

Consider (\ref{eq:Laser1}) with  $\omega \neq 0$.
Now, we find all non-constant free interaction solutions that are born at the regular stationary state
as $C_\mathfrak{b}$ passes through the bifurcation value $1$.
In case $C_\mathfrak{b} > 1$, i.e., $g^2 d  >   \kappa \gamma $,
these free interaction solutions lead to
the periodic solutions of (\ref{eq:Lorenz}),
which are given by 
$ \Tr\left(  \rho_{t} \, \sigma^{3} \right)  =  \gamma \kappa / \left( g^2 \right) $,
$ \Tr\left(  \rho_{t} \, a \right) 
=  z \, \exp \left( - i \omega t \right) \gamma  \sqrt{ C_\mathfrak{b} -1 } / \left(  \sqrt{2} \left| g \right| \right)$
and
$ \Tr\left( \rho_{t} \,  \sigma^{-}  \right) =  
z \, \exp \left( - i \omega t \right) \gamma  \kappa \sqrt{ C_\mathfrak{b} -1 } / \left(  \sqrt{2} \left| g \right| g \right)$
for any $z \in \mathbb{C}$ satisfying $\left\vert z \right\vert = 1 $.

\begin{theorem}
\label{th:FreeS-LaserE}
Let $\omega \neq 0$.
Take 
$
C_\mathfrak{b} =  d g^2 / \left( \gamma \kappa \right) 
$,
and consider  $\varrho_c $ is given by (\ref{eq:Def_orbita}).
If 
$
C_\mathfrak{b} \leq 1
$,
then (\ref{eq:Laser1})  does not have any non-constant $N$-regular free interaction solution.
In case 
$
C_\mathfrak{b} > 1
$,
the family
$
\left\{ 
t \mapsto \varrho_c \left( \omega t - \theta \right) : \theta \in \left[ 0, 2 \pi \right[
\right\}
$
is composed by all non-constant $N$-regular free  interaction solutions to (\ref{eq:Laser1}). 
\end{theorem}

\begin{proof}
 Deferred to Section  \ref{sec:Proof:FreeS-LaserE}.
\end{proof}

\begin{remark}
Suppose that $\omega \neq 0$.
According to the proof of Theorem \ref{th:FreeS-LaserE} we have that
$\varrho_{\infty}$,  described by (\ref{eq:I10}, 
is the unique constant $N$-regular free interaction solution to (\ref{eq:Laser1}).
 \end{remark}
 
 \begin{remark}
 The function $ t \mapsto \varrho_c \left( \omega t - \theta \right) $ is a periodic $N$-weak
 solution to (\ref{eq:Laser1}).
 By Theorem \ref{th:FreeS-LaserE},
 in the phase space all non-constant $N$-regular free interaction solutions of (\ref{eq:Laser1})
have the same closed (or periodic) orbit, which is 
$\left\{ \varrho_c \left( \vartheta \right): \vartheta \in \left[ 0, 2 \pi \right] \right\}$.
 \end{remark}
 
 %%%%%%%%%%%%%%%%
 
\subsection{Long-time behavior}

 Suppose that $\rho_0  \in \mathfrak{L}_{1,N}^{+} \left( \ell^2 \left(\mathbb{Z}_+ \right) \otimes \mathbb{C}^2 \right) $
 satisfies  $\Tr\left(  \rho_{0} \, a   \right) = \Tr\left( \rho_{0} \, \sigma^{-}  \right) = 0$.
 From (\ref{eq:Lorenz}) it follows that 
 $\Tr\left(  \rho_{t} \, a  \right) = \Tr\left(  \rho_{t} \, \sigma^{-}  \right) = 0$ for all $t \geq 0$,
 and so 
 $ 
\frac{d}{dt} \left( \Tr\left(   \rho_{t} \,  \sigma^{3} \right) - d \right) 
 =  
 - 2 \gamma \left(  \Tr\left(    \rho_{t} \, \sigma^{3} \right) - d \right)
$.
Therefore, 
$\Tr\left(   \rho_{t} \, \sigma^{3}  \right)$ converges exponentially fast to $d$ as $t \rightarrow +\infty$.
Theorem \ref{th:LongTimeD} below provides a full quantum  explanation for this long time behavior.

\begin{theorem}
\label{th:LongTimeD}
Let $\varrho$ be a $N$-regular density operator in $\ell^2 \left(\mathbb{Z}_+ \right) \otimes \mathbb{C}^2$
such that $\Tr\left( \varrho \,  a \right) = \Tr\left(  \varrho  \, \sigma^{-} \right) = 0$.
Suppose that $\left( \rho_t  \right)_{t \geq 0}$ is the $N$-weak solution to (\ref{eq:Laser1}) with initial state $\varrho$.
Then $\Tr\left( \rho_{t} \, a  \right) = \Tr\left( \rho_{t} \, \sigma^{-}  \right) = 0$ for all $t \geq 0$, and
\begin{equation}
 \label{eq:8.22}
\Tr \left( \left\vert \rho_t - \varrho_{\infty}  \right\vert  \right) 
\leq
12 \, \exp \left( - \gamma  t \right)    \left( 1 +  \left\vert d \right\vert \right)
+
4 \, \exp \left( - \kappa t \right)  \sqrt{ \Tr \left( \varrho \,  N  \right) } 
\quad \quad \quad \forall t \geq 0,
\end{equation}
with $\varrho_{\infty}$ defined by (\ref{eq:I10}).
\end{theorem}

\begin{proof}
 Deferred to Section  \ref{sec:Proof:LongTimeD}.
\end{proof}

Let $C_\mathfrak{b} < 1$.
Then,
the equilibrium solution $\left( 0, 0, d \right)$ of the the Maxwell-Bloch equations (\ref{eq:Lorenz})
is asymptotically stable  (see, e.g., \cite{FagMora2019}).
Hence,
$
\lim_{t \rightarrow + \infty}   \Tr \left(   \rho_{t} \, a \right) 
=
\lim_{t \rightarrow + \infty} \Tr \left( \rho_{t} \, \sigma^{-}  \right) = 0$,
and
$\lim_{t \rightarrow + \infty}  \Tr \left(   \rho_{t} \, \sigma^{3} \right)  = d$.
Now,
we show that $\rho_t$ converges in the trace norm to the stationary state (\ref{eq:I10})
with exponential rate,
as well as 
we get the limiting behavior of the mean values of $N$-bounded  operators like
$  N \otimes  \left(  \sigma^{+} + \sigma^{-} \right)$.

\begin{theorem}
\label{th:LongTime}
Let $\left( \rho_t  \right)_{t \geq 0}$ be the $N$-weak solution to (\ref{eq:Laser1}) 
with $\omega \neq 0$ and initial datum
$\rho_0 \in \mathfrak{L}_{1,N}^{+} \left( \ell^2 \left(\mathbb{Z}_+ \right) \otimes \mathbb{C}^2 \right) $.
Suppose that 
$
C_\mathfrak{b} < 1
$.
Then
\begin{equation}
\label{eq:LTB.1}
 \Tr \left( \left\vert \rho_t - \varrho_{\infty}   \right\vert  \right)  
 \leq
 K_{sys} \left( \left\vert g \right\vert , \Tr\left( \rho_{0} \, N  \right) \right) \exp \left( - \delta_{sys} \, t \right) \hspace{2cm} \forall t \geq 0 ,
\end{equation} 
where
$\varrho_{\infty}$ is given by (\ref{eq:I10}),
\begin{equation}
\label{eq:RateCb<1}
 \delta_{sys}
=
\begin{cases}  
 \min \left\{  \kappa  ,  \gamma \right\} / 2 & \text{ if } d < 0
\\
  \left( 1- C_\mathfrak{b} \right) \min \left\{ \kappa ,   \gamma \right\} / 3 & \text{ if } d \geq 0 
\end{cases}   ,
\end{equation}
and
$ K_{sys} \left( \cdot , \cdot \right) $ is a non-decreasing non-negative function of two variables 
that depends  on  the parameters $d$, $\kappa$ and $\gamma$.
Fix $\widetilde{K} > 0$.
Then, 
for all $t \geq 0$ we have 
\begin{equation}
 \label{eq:8.35}
\begin{aligned}
& 
\left\vert 
\Tr \left( \rho_t \, A  \right)
- 
\frac{d+1}{2} \langle e_0 \otimes e_+, A \, e_0 \otimes e_+ \rangle
-
\frac{1-d}{2} \langle e_0 \otimes e_-, A \, e_0 \otimes e_- \rangle 
\right\vert
\\
& 
\leq
\widetilde{K} \hat{K}_{sys} \left( \left\vert g \right\vert  ,  \Tr\left( \rho_{0} \, N  \right) \right) 
\exp \left( - \delta_{sys} \, t  \right) ,
 \end{aligned}
\end{equation}
for any 
$A :  \ell^2 \left(\mathbb{Z}_+ \right) \otimes \mathbb{C}^2 \rightarrow  \ell^2 \left(\mathbb{Z}_+ \right) \otimes \mathbb{C}^2$
satisfying 

\begin{equation}
\label{eq:LTB.c}
\max \left\{\left\Vert  A  \, x  \right\Vert, \left\Vert  A^{\star}  \, x  \right\Vert \right\}
\leq 
\widetilde{K}   \left\Vert  x  \right\Vert_{N}
\qquad \qquad
\forall x \in \mathcal{D}\left( N \right) .
\end{equation}
Here, 
$\hat{K}_{sys} \left( \cdot , \cdot \right) $
is a non-decreasing non-negative function depending on  $d$, $\kappa$ and $\gamma$.
\end{theorem}

\begin{proof}
 Deferred to Section  \ref{sec:Proof:LongTime}.
\end{proof}

In case $ C_\mathfrak{b}  < 1 $,
from Theorem \ref{th:LongTime} we conclude that the radiation field converges very fast to 
his ground state  $\ketbra{e_0}{e_0}$.
Furthermore,
using Theorem \ref{th:LongTime} one can obtain at once the long time behavior of, for instance, 
the photon-number statistics,  
the quantum variance of the quadratures,
and the quantum linear entropy,
which are relevant physical quantities that are  not given by the Maxwell-Bloch equations \eqref{eq:Lorenz}.

\begin{corollary}
Suppose that 
$ C_\mathfrak{b}  < 1 $ and $\omega \neq 0$.
Let $\delta_{sys}$ be defined by \eqref{eq:RateCb<1}.
Consider the $N$-weak solution $\left( \rho_t  \right)_{t \geq 0}$  to (\ref{eq:Laser1}).
Then:

\begin{itemize}
 
 \item There exists $K > 0$ such that  all $t \geq 0$,
 $
\left\vert 
\Tr \left( \rho_t \, \ketbra{e_0}{e_0}  \right)
- 
1
\right\vert
\leq K \exp \left( - \delta_{sys} \, t \right)
$
and 
$
\left\vert 
\Tr \left( \rho_t \, \ketbra{e_n}{e_n}  \right)
\right\vert
\leq K \exp \left( - \delta_{sys} \, t \right)
$
for any $n \in \mathbb{N}$.

\item 
For all $t \geq 0$,
$
\left\vert 
\Tr \left( \rho_t \, Q^2  \right) - \Tr \left( \rho_t \, Q  \right)^2
- 
1/2
\right\vert
\leq K \exp \left( - \delta_{sys} \, t \right)
$
and
\[
\left\vert 
\Tr \left( \rho_t \, P^2  \right) - \Tr \left( \rho_t \, P  \right)^2
- 
1/2
\right\vert 
\leq K \exp \left( - \delta_{sys} \, t \right) ,
\]
where 
$Q = \left(  a^\dagger  + a  \right) / \sqrt{2}
$ 
and
$
P = \mathrm{i} \left(  a^\dagger  - a  \right) / \sqrt{2}
$.
% are the position and momentum operators

\item For all $t \geq 0$,
$
\left\vert 
\left( 1 - \Tr \left( \rho_t ^2  \right)  \right) 
- 
 \left( 1- d^2 \right) / 2 \right\vert
\leq
K \exp \left( - \delta_{sys} \, t \right) 
$.

\end{itemize}

\end{corollary}

We equip the phase space  
$ \left( \mathfrak{L}_{1,N}^{+} \left( \ell^2 \left(\mathbb{Z}_+ \right) \otimes \mathbb{C}^2 \right) ,   \mathfrak{d}_{N} \right)$
with the distance $ \mathfrak{d}_{N} $ defined below.
Applying  Theorem \ref{th:LongTime} we obtain that
 $ \mathfrak{d}_{N} \left( \rho_t  , \varrho_{\infty} \right) $ converges exponentially fast to $0$ as $t \rightarrow + \infty$
whenever $\rho_0 \in \mathfrak{L}_{1,N}^{+} \left( \ell^2 \left(\mathbb{Z}_+ \right) \otimes \mathbb{C}^2 \right) $
and $ C_\mathfrak{b} < 1 $.
Therefore,
$\varrho_{\infty}$ is the globally stable equilibrium point  of the dynamical system on 
$ \left( \mathfrak{L}_{1,N}^{+} \left( \ell^2 \left(\mathbb{Z}_+ \right) \otimes \mathbb{C}^2 \right) ,   \mathfrak{d}_{N} \right)$
given by   (\ref{eq:Laser1}) with $\omega \neq 0$ and $ C_\mathfrak{b} < 1 $.

\begin{definition}
\label{def_distancia}
For any $  \varrho, \widetilde{ \varrho } \in  \mathfrak{L}_{1,C}^{+} \left( \mathfrak{h}\right) $
we define $ \mathfrak{d}_{C} \left( \varrho , \widetilde{ \varrho } \right) $ to be the supremum over
all $ \left\vert  \Tr  \left(  A  \left(    \varrho  - \widetilde{ \varrho } \right) \right) \right\vert $
with $ A : \mathfrak{h} \rightarrow \mathfrak{h} $ linear operator  satisfying 
\[
\max \left\{ \left\Vert   A \, x \right\Vert^2,  \left\Vert   A^* \, x \right\Vert^2 \right\} 
\leq 
\left( \left\Vert   x \right\Vert^2 + \left\Vert   C \, x \right\Vert^2  \right) / 2
\qquad \qquad 
\forall  x \in \mathcal{D}\left(C\right) .
\]
\end{definition}

\begin{remark}
The space $  \mathfrak{L}_{1,C}^{+} \left( \mathfrak{h}\right) $ equipped with 
$ \mathfrak{d}_{C}  $, given by Definition \ref{def_distancia}, is a metric space. 
Moreover,
$
\Tr  \left(   \left\vert   \varrho  - \widetilde{ \varrho } \right\vert \right)
 \leq
 \sqrt{2} \, \mathfrak{d}_{C} \left( \varrho , \widetilde{ \varrho } \right) 
$
for all $  \varrho, \widetilde{ \varrho } \in  \mathfrak{L}_{1,C}^{+} \left( \mathfrak{h}\right) $,
and 
$ \mathfrak{d}_{I} \left( \varrho , \widetilde{ \varrho } \right) 
=  
\Tr  \left(   \left\vert   \varrho  - \widetilde{ \varrho } \right\vert \right)  
$,
where $I$ is the identity operator in $ \mathfrak{h}$.
\end{remark}

We turn to the case $C_\mathfrak{b} > 1$.
Next, we deal with the long-time convergence of the solution of
 (\ref{eq:Laser1}) with $\omega \neq 0$ and  
$ \kappa^2 +  5 \kappa \gamma >  \gamma \left( \kappa - 3 \gamma \right) C_\mathfrak{b}$.
In this case,
the cavity is good (i.e., $\kappa \leq 3 \gamma $)
or $C_\mathfrak{b} $ is less than the second threshold 
$ 
\left( \kappa^2 +  5 \kappa \gamma \right) / \left(  \gamma \left( \kappa - 3 \gamma \right)  \right)
$.
For these parameter values,
one can obtain the local stability of the  periodic solutions of the Maxwell-Bloch equations  \eqref{eq:Lorenz}
by using linear stability analysis (see, e.g., Lemma  \ref{lem:LorenzBifurcation} and  \cite{NingHaken1990}).

\begin{theorem}
\label{th:LimitCycle}
Let  $C_\mathfrak{b} > 1$ and $\omega \neq 0$.
Assume that $\kappa \leq 3 \gamma $ 
or that $\kappa > 3 \gamma $ and $ \kappa^2 +  5 \kappa \gamma >  \gamma \left( \kappa - 3 \gamma \right) C_\mathfrak{b}$. 
Then,
there exist constants $\epsilon, \lambda > 0$ such that for any
$N$-weak solution $\left( \rho_t  \right)_{t \geq 0}$ to (\ref{eq:Laser1}) with $\omega \neq 0$
we have 
\begin{equation}
 \label{eq:LC4}
 \Tr \left( \left\vert  
 \rho_t 
 - 
 \varrho_c \left( \omega t - \theta_{\infty} \right)
 \right\vert  \right)
\leq
  K \left( \Tr\left( \rho_{0} \, N  \right)  \right) \exp \left(  - \lambda t \right) 
 \hspace{3cm} \forall t \geq 0
\end{equation}
provided that the initial datum $\rho_0 \in \mathfrak{L}_{1,N}^{+} \left( \ell^2 \left(\mathbb{Z}_+ \right) \otimes \mathbb{C}^2 \right) $ satisfies 
\begin{equation}
\label{eq:CondStabEpsilon}
\left\{
\begin{aligned}
 &
  \left\vert   \left\vert  \Tr\left(  \rho_{0}  \, a \right)  \right\vert  -  \frac{\gamma \sqrt{ C_\mathfrak{b} - 1} }{\sqrt{2} \left\vert g \right\vert}  \right\vert 
 < \epsilon ,
 \quad
  \left\vert  \Tr\left( \rho_{0} \, \sigma^{-}  \right)   -   \frac{ \kappa }{ g } \frac{\gamma \sqrt{ C_\mathfrak{b} - 1} }{\sqrt{2} \left\vert g \right\vert}  
 \frac{  \Tr\left(  \rho_{0} \, a \right) }{ \left\vert  \Tr\left(  \rho_{0}  \, a \right)  \right\vert }\right\vert
 < \epsilon ,
 \\
&  
\text{and }
\left\vert \Tr\left( \rho_{0} \, \sigma^{3} \right) -  \frac{d}{C_\mathfrak{b}} \right\vert 
< \epsilon .
\end{aligned}
\right.
\end{equation}
Here,
$\varrho_c $ is given by (\ref{eq:Def_orbita}),
and $\theta_{\infty} \in  \left[ 0, 2 \pi \right[$ is the argument of the unit complex number 
\begin{equation*}
\frac{  \Tr\left(  \rho_{0} \, a \right) }{  \left\vert   \Tr\left(  \rho_{0} \, a \right) \right\vert } 
\exp \left(  \mathrm{i}  g \int_0^{+ \infty} \Im \left( \frac{S  \left( s \right) }{  A \left( s \right) } \right) ds \right) ,
\end{equation*}
where
$\left( A \left( t  \right), S \left( t  \right), D \left( t  \right)  \right)$ is the solution of  \eqref{eq:Lorenz}
with  $\omega = 0$,
$ A \left( 0 \right) = \Tr\left( \rho_0 \, a \right) $,
$ S \left( 0  \right)  = \Tr\left( \rho_0 \, \sigma^{-}  \right) $
and
$ D \left( 0  \right) = \Tr\left(  \rho_0 \,  \sigma^{3} \right) $.
Moreover, 
under the condition (\ref{eq:CondStabEpsilon}) we have
\begin{equation}
 \label{eq:LC5}
  \left\vert  
  \Tr \left(  \rho_t  \, A \right)
  -
   \Tr \left(    \varrho_c \left( \omega t - \theta_{\infty} \right) A \right)
\right\vert  
\leq
 \widetilde{K}  \, K\left( \Tr\left( \rho_{0} \, N  \right)  \right) \exp \left(  - \lambda t \right) 
 \hspace{3cm} \forall t \geq 0
\end{equation}
for any linear operator $A$ in 
$\ell^2 \left(\mathbb{Z}_+ \right) \otimes \mathbb{C}^2 $ satisfying 
\begin{equation}
\label{eq:LC6}
 \max \left\{\left\Vert  A  \, x  \right\Vert, \left\Vert  A^{\star}  \, x  \right\Vert \right\}
\leq 
\widetilde{K}  \left\Vert  x  \right\Vert_{N}
\qquad \forall x \in \mathcal{D}\left( N \right) ,
\end{equation}
where $\widetilde{K}  > 0$
and 
the non-decreasing non-negative function $ K\left( \cdot \right) $ does not depend on $A$.
\end{theorem}

\begin{proof}
 Deferred to Section  \ref{sec:Proof:LimitCycle}.
In order to facilitate the reading of the proof,
we advice to read first the proofs of  Theorems \ref{th:FreeS-LaserE}  and \ref{th:LongTime},
although it is not necessary. 
\end{proof}

According to Theorem \ref{th:LimitCycle} and Lemma \ref{lem:VecindadCL}, given below,
we have that there exist a neighborhood of the limit cycle 
$\left\{ \varrho_c \left( \vartheta \right): \vartheta \in \left[ 0, 2 \pi \right] \right\}$
in 
$ \left( \mathfrak{L}_{1,N}^{+} \left( \ell^2 \left(\mathbb{Z}_+ \right) \otimes \mathbb{C}^2 \right) ,   \mathfrak{d}_{N} \right)$
such that 
$ \rho_t $  approaches exponentially fast to $\left\{ \varrho_c \left( \vartheta \right): \vartheta \in \left[ 0, 2 \pi \right] \right\}$
--in the metric $ \mathfrak{d}_{N}$--
whenever $ \rho_0 $  is in this neighborhood.
Hence,
$\left\{ \varrho_c \left( \vartheta \right): \vartheta \in \left[ 0, 2 \pi \right] \right\}$
is a  stable limit cycle of the dynamical system on 
$ \left( \mathfrak{L}_{1,N}^{+} \left( \ell^2 \left(\mathbb{Z}_+ \right) \otimes \mathbb{C}^2 \right) ,   \mathfrak{d}_{N} \right)$
given by  (\ref{eq:Laser1}) with $\omega \neq 0$ and  
$ \kappa^2 +  5 \kappa \gamma >  \gamma \left( \kappa - 3 \gamma \right) C_\mathfrak{b}$.
By Theorem \ref{th:LongTimeD},
the basis of attraction of 
$\left\{ \varrho_c \left( \vartheta \right): \vartheta \in \left[ 0, 2 \pi \right] \right\}$
is not composed by all  $N$-regular density operator different from $ \varrho_{\infty} $.

\begin{lemma}
\label{lem:VecindadCL}
 Consider $\varrho \in \mathfrak{L}_{1,N}^{+} \left( \ell^2 \left(\mathbb{Z}_+ \right) \otimes \mathbb{C}^2 \right) $
 and  $C_\mathfrak{b} > 1$.
 If 
 $ \Tr\left( \varrho \, a \right) \neq 0$ and 
\begin{equation}
\label{eq:DistLC}
\inf \left\{   \mathfrak{d}_{N} \left( \varrho , \varrho_c \left( \vartheta \right) \right) 
 :  \vartheta \in \left[ 0, 2 \pi \right] \right\}
 < 
 \varepsilon ,
\end{equation}
then  the inequalities (\ref{eq:CondStabEpsilon}) hold with 
 $\rho_{0} = \varrho$ and $\epsilon = \sqrt{2} \, \varepsilon  \left( 1 + 2 \, \kappa / g  \right)$.
Here,  
$\varrho_c$ and  $ \mathfrak{d}_{C} $ are described by (\ref{eq:Def_orbita}) and Definition  \ref{def_distancia},
respectively.
 Moreover, 
 $ \Tr\left( \varrho \, a \right) \neq 0$ 
 whenever 
 $
 \inf \left\{   \mathfrak{d}_{N} \left( \varrho , \varrho_c \left( \vartheta \right) \right) 
 :  \vartheta \in \left[ 0, 2 \pi \right] \right\}
 <
  \gamma \sqrt{ C_\mathfrak{b} -1 } / \left(  2 \left| g \right| \right) 
 $.
\end{lemma}

\begin{proof}
 Deferred to Section  \ref{sec:Proof:VecindadCL}.
\end{proof}

Under the hypotheses of Theorem \ref{th:LimitCycle},
the laser operates stably.  
Applying Theorem \ref{th:LimitCycle} we deduce that 
if the initial density operator is in a small enough neighborhood of 
the orbit $\left\{ \varrho_c \left( \vartheta \right): \vartheta \in \left[ 0, 2 \pi \right] \right\}$,
then, for instance, 
the probability distribution of finding $n$ photons are pulled toward a Poissonian statistics,
the product of the standard deviations of the position and momentum operators converges to the 
lower bound of the Heisenberg's uncertainty principle,
and
the quantum linear entropy converges exponentially fast to 
$
1/2 + d^2 / \left( 2 \, C_{\mathfrak{b}}^2 \right) - d^2/ C_{\mathfrak{b}}
$.

\begin{corollary}
Let  $C_\mathfrak{b} > 1$ and $\omega \neq 0$.
Suppose that  $\kappa \leq 3 \gamma $ 
or that $\kappa > 3 \gamma $ and $ \kappa^2 +  5 \kappa \gamma >  \gamma \left( \kappa - 3 \gamma \right) C_\mathfrak{b}$. 
Then,
there exist constants $\epsilon, \lambda > 0$ such that the fulfillment of \eqref{eq:CondStabEpsilon} implies that:

\begin{itemize}
 
 \item There exists $K > 0$ such that  all $t \geq 0$ and $n \in \mathbb{N}$,
 \[
\left\vert 
\Tr \left( \rho_t \, \ketbra{e_n}{e_n}  \right)
-
 \hbox{\rm e}^{  -  \frac{ \gamma^2 \left( C_\mathfrak{b} -1 \right) }{ 2 g^2 }  }
\left(  \frac{ \gamma^2 \left( C_\mathfrak{b} -1 \right) }{ 2 g^2 } \right)^n / n!
\right\vert
\leq K \exp \left( - \lambda \, t \right) .
\]

\item
For all $t \geq 0$,
$
\left\vert 
\Tr \left( \rho_t \, Q^2  \right) - \Tr \left( \rho_t \, Q  \right)^2
- 
1/2
\right\vert
\leq K \exp \left( - \lambda \, t \right)
$
and
\[
\left\vert 
\Tr \left( \rho_t \, P^2  \right) - \Tr \left( \rho_t \, P  \right)^2
- 
1/2
\right\vert 
\leq K \exp \left( - \lambda \, t \right) .
\]

\item For all $t \geq 0$,
$
\left\vert 
\left( 1 - \Tr \left( \rho_t ^2  \right)  \right) 
- 
\left( \frac{1}{2} + \frac{1}{2} \frac{d^2}{C_{\mathfrak{b}} ^2} - \frac{d^2}{C_{\mathfrak{b}}} \right)
\right\vert
\leq
K \exp \left( -  \lambda \, t \right) 
$.

\end{itemize}

\end{corollary}

From the  proof of Theorems \ref{th:LongTime} and \ref{th:LimitCycle} we have 
$
\sup_{t \geq 0}  \Tr\left(  \rho_{t} \, N  \right) 
<
+ \infty
$,
and so combining Lemma 18 of \cite{Winter2016} (see also \cite{Shirokov2010}) 
with  Theorems \ref{th:LongTime} and \ref{th:LimitCycle} 
we obtain the long-time limit of the von Neumann entropy.

\begin{corollary}
Consider the $N$-weak solution $\left( \rho_t  \right)_{t \geq 0}$ 
to (\ref{eq:Laser1}) with $\omega \neq 0$.

\begin{itemize}
 
\item 
If $C_\mathfrak{b} < 1$, then
$
 \lim_{t \rightarrow + \infty} - \Tr \left( \rho_t   \log \left( \rho_t   \right) \right)
 =
 -\frac{1}{2} \log \left( \frac{1}{4} - \frac{d^2}{4} \right)
 - \frac{d}{2} \log \left( \frac{1 + d }{ 1 - d} \right)
 $.

\item If $C_\mathfrak{b} > 1$ and 
$\gamma \left( \kappa - 3 \gamma \right) C_\mathfrak{b} <  \kappa^2 +  5 \kappa \gamma $,
then
there exist $\epsilon > 0$ such that
$ - \Tr \left( \rho_t   \log \left( \rho_t   \right) \right) $ converges to
\[
-
 \frac{1}{2} \log \left( \frac{1}{4} -   \frac{d^2}{4 C_{\mathfrak{b}} ^2} \left( 2 C_{\mathfrak{b}} - 1 \right)
 \right)
 -
 \frac{d \sqrt{2 C_{\mathfrak{b}} - 1} }{2 C_{\mathfrak{b}}} 
 \log \left(
 \frac{ C_{\mathfrak{b}} + d \sqrt{ 2 C_{\mathfrak{b}} - 1} }
 { C_{\mathfrak{b}} - d \sqrt{ 2 C_{\mathfrak{b}} - 1} } 
 \right)
 \]
 as $t \rightarrow + \infty $ in case  \eqref{eq:CondStabEpsilon} holds.
 
\end{itemize}
 
\end{corollary}

\begin{remark}
If we change the phase of the electromagnetic field of the laser modeled by  \eqref{eq:Laser1},
then 
the evolution of the density operator $ \hat{\rho}_t$  describing the  laser,
under the mean field approximation, 
is governed by the GKSL equation \eqref{eq:Laser1} with  the mean-field Hamiltonian $\mathcal{H}$ replaced by 
\begin{align*}
 \mathcal{H} \left( \varrho \right) 
& =
\omega \left(  a^\dagger a + \sigma^{3} / 2 \right) 
+ 
g  \left( \mathrm{e}^{ \mathrm{i} \phi}  \Tr \left( \sigma^{-}  \varrho \right) a^{\dagger}  +  \mathrm{e}^{ - \mathrm{i} \phi} \Tr \left( \sigma^{+}  \varrho  \right) a 
+  \mathrm{e}^{ \mathrm{i} \phi} \Tr\left( a^{\dagger}  \varrho  \right) \sigma^{-}  + \mathrm{e}^{ - \mathrm{i} \phi} \Tr\left( a \,  \varrho  \right) \sigma^{+}    
\right) .
\end{align*}
Since 
$
\rho_t
 = 
 \exp \left(  \mathrm{i} \left(  \pi / 2 - \phi \right) N \right) \,
 \hat{\rho}_t  \,
 \exp \left( -  \mathrm{i} \left(  \pi / 2 - \phi \right) N \right)
$
satisfies  (\ref{eq:Laser1}),
the long time behavior of $\hat{\rho}_t$ is characterized by that of  $\rho_t$.
\end{remark}

\section{Linear quantum master equation}
\label{sec:LinearQMEs}

This section is devoted to the non-homogeneous  linear evolution equation (\ref{eq:AuxiliarGKSL}).
Suppose for a moment that $\rho_t$ is the  $N$-weak solution of  the  mean-field laser equation (\ref{eq:Laser1})
with $\rho_0 \in \mathfrak{L}_{1,N}^{+} \left( \ell^2 \left(\mathbb{Z}_+ \right) \otimes \mathbb{C}^2 \right) $.
According to  \cite{FagMora2019} we have that
$t \mapsto \left( \Tr\left( a \, \rho_{t} \right),  \Tr\left( \sigma^{-} \rho_{t}  \right) ,  \Tr\left(  \sigma^{3}  \rho_{t} \right) \right) $
is the solution to the  Maxwell-Bloch equations (\ref{eq:Lorenz}) with $\left( \Tr\left( a \, \rho_{0} \right),  \Tr\left( \sigma^{-} \rho_{0}  \right) ,  \Tr\left(  \sigma^{3}  \rho_{0} \right) \right) $ as initial condition.
Then, as in \cite{FagMora2019} we replace in  (\ref{eq:Laser1}) the functions 
$  \Tr\left( \sigma^{-} \rho_{t}  \right) $ and  $  \Tr\left( a \, \rho_{t} \right) $ 
by $S \left( t  \right)$ and $ A \left( t \right)$, respectively, 
to obtain that $\rho_t$ satisfies  (\ref{eq:AuxiliarGKSL}) with initial datum 
$\rho_0$ and coefficients  $\alpha \left( t \right) = g \,  S\left( t  \right)$ and $\beta \left( t \right) = g \,  A\left( t  \right)$.
We use this relation 
% between (\ref{eq:Laser1}) and  (\ref{eq:AuxiliarGKSL}) coupled to (\ref{eq:Lorenz})
to study the equilibrium point of the  mean-field laser equation.
Moreover,
in order to study the cycle of  (\ref{eq:Laser1}) with  $C_\mathfrak{b} > 1$, 
applying the unitary transformation 
$
\rho_t
\mapsto 
\exp \left(  \mathrm{i}  \omega  \left( N +  \sigma^3 / 2 \right) t \right)
\rho_t
\exp \left( - \mathrm{i}  \omega  \left( N +  \sigma^3 / 2 \right) t \right)
$
we transform  (\ref{eq:Laser1}) into  (\ref{eq:Laser1}) with $\omega = 0$.
This leads to treat the fixed points of the system formed by  (\ref{eq:Lorenz}) 
and  (\ref{eq:AuxiliarGKSL}) 
with $\alpha \left( t \right) = g \,  S\left( t  \right)$ and $\beta \left( t \right) = g \,  A\left( t  \right)$,
in case $\omega = 0$.

If the functions $\alpha \left( t \right)$ and $\beta \left( t \right)$ are constant,
i.e.,  $\alpha \left( t \right) \equiv  \alpha_0 \in \mathbb{C}$ and $ \beta \left( t \right) \equiv  \beta_0 \in \mathbb{C}$,
then 
(\ref{eq:AuxiliarGKSL}) becomes the autonomous quantum master equation 
\begin{equation}
\label{eq:3.2g}
\left\{
 \begin{aligned}
  \frac{d}{dt}  \rho^R_t \left( \varrho \right)
& =
\mathcal{L}_{\star}^h \, \rho^R_t
+
\left[ 
\alpha_0 \, a^\dagger -  \overline{\alpha_0 } \, a 
+ \overline{\beta_0 } \, \sigma^{-}  - \beta_0  \,\sigma^+  ,
\rho^R_t \right]
\hspace{1cm} \forall t \geq 0 
\\
\rho^R_0 \left( \varrho \right) 
& = \varrho  
\end{aligned}
\right. ,
\end{equation}
where 
$ \mathcal{L}_{\star}^h $ is defined by (\ref{eq:3.21});
we  recall that  
$d\in \left]-1,1 \right[$, $\omega \in\mathbb{R}$ and  $\kappa,\gamma>0$.
For any $\varrho \in \mathfrak{L}_{1, N^p}^{+} \left( \ell^2(\mathbb{Z}_+)\otimes \mathbb{C}^2 \right) $
with $p \in \mathbb{N}$,
(\ref{eq:3.2g}) has a unique $N^p$-weak solution (see  \cite{FagMora2019}).
Using Theorems 4.1 and 4.3 of \cite{MoraAP} we get that
the family of bounded linear operators 
\[
\left(
\rho^R_t :   \mathfrak{L}_{1, N}^{+}  \left( \ell^2(\mathbb{Z}_+)  \otimes   \mathbb{C}^2 \right)  
\rightarrow 
\mathfrak{L}_{1} \left( \ell^2(\mathbb{Z}_+)   \otimes  \mathbb{C}^2 \right) 
\right)_{t \geq 0}
\]
can be extended uniquely to a  one-parameter semigroup of contractions 
$\left( \rho^R_t \left( \cdot \right) \right)_{t \geq 0}$ on 
$\mathfrak{L}_{1}\left( \ell^2(\mathbb{Z}_+)\otimes \mathbb{C}^2 \right) $,
which indeed is a $C_0$-semigroup as the next theorem shows.

\begin{theorem}
\label{th:ConvEqHE}
The family  $\left( \hspace{-0.3pt} \rho^R_t  \right)_{t \geq 0}$ 
is a strongly continuous semigroup on bounded linear operators on
$\mathfrak{L}_{1}   \left( \ell^2(\mathbb{Z}_+)   \otimes  \mathbb{C}^2 \right) $.
\end{theorem}

\begin{proof}
 Deferred to Section  \ref{sec:Proof:th:ConvEqHE}.
\end{proof}

We rewrite (\ref{eq:AuxiliarGKSL}) as 
\begin{equation}
 \label{eq:AuxiliarGKSLr}
 \begin{aligned}
\frac{d}{dt}  \rho_t
& =
\mathcal{L}_{\star}^h \, \rho_t
+
\left[ 
\alpha_0 \, a^\dagger -  \overline{\alpha_0 } \, a 
+ \overline{\beta_0 } \, \sigma^{-}  - \beta_0  \,\sigma^+  ,
\rho_t \right]
\\
& \quad
+
\left[ 
\left( \alpha \left( t \right)  - \alpha_0 \right)  a^\dagger -  \overline{\left( \alpha \left( t \right)  - \alpha_0 \right)} a 
+ \overline{ \left( \beta  \left( t \right) - \beta_0 \right)} \sigma^{-}  
-  \left( \beta  \left( t \right) - \beta_0 \right) \sigma^+  , \rho_t \right] ,
\end{aligned}
\end{equation}
where $ \mathcal{L}_{\star}^h $ is given by (\ref{eq:3.21}).
Thus,
we see (\ref{eq:AuxiliarGKSL}) as a perturbation of (\ref{eq:3.2g}) in case 
the functions $\alpha  \left( t \right) $ and $\beta  \left( t \right) $
converge to the points $\alpha_0 $ and $ \beta_0 $ as $ t \rightarrow + \infty$.
Since the current  perturbation theory does not apply to \eqref{eq:AuxiliarGKSLr}  --to the best of our knowledge--,
we next develop mathematical  perturbation methods for (\ref{eq:AuxiliarGKSLr}),
and so for (\ref{eq:Laser1}).
First, we establish  a variation of constant formula for (\ref{eq:AuxiliarGKSLr})
by  using techniques from functional analysis and stochastic processes.

\begin{theorem}
\label{th:VariationConstants}
Let $\left( \rho_t \right)_{t \geq 0}$ be the $N$-weak solution of (\ref{eq:AuxiliarGKSL})
with $\alpha, \beta : \left[ 0 , \infty \right[ \rightarrow \mathbb{C}$ continuous functions.
Then, for all $t \geq s \geq 0$ we have 
\begin{equation}
\label{eq:AuxiliarGKSL5}
 \rho_t 
 =
 \rho^R_{t-s} \left( \rho_s  \right) 
 + \int_s^t  \rho^R_{t-u} \left( 
 \left[  \left( \alpha_R \left( u \right) a^\dagger -  \overline{\alpha_R \left( u \right)} a \right)
+\left( \overline{\beta_R  \left( u \right) } \sigma^{-}  - \beta_R  \left( u \right) \sigma^+ \right)   ,
\rho_u \right]  \right) du ,
\end{equation}
where $ \rho^R_t \left( \cdot \right)$ is as in Theorem \ref{th:ConvEqHE},
$ \alpha_R \left( u \right) = \alpha \left( u \right) - \alpha_0 $
and 
$ \beta_R  \left( u \right) = \beta \left( u \right) - \beta_0 $.
\end{theorem}

\begin{proof}
 Deferred to Section  \ref{sec:Proof:th:VariationConstants}.
\end{proof}

If  $\alpha  \left( t \right) $ and $\beta  \left( t \right) $
converge fast enough to  $\alpha_0 $ and $ \beta_0 $ as $ t \rightarrow + \infty$,
then from  (\ref{eq:AuxiliarGKSL5}) we infer that (\ref{eq:AuxiliarGKSL})  and (\ref{eq:3.2g}) have similar long-time behavior.
In more detail, 
using Theorem \ref{th:VariationConstants} we get the following estimate 
of the trace distance between the $N$-weak solution to (\ref{eq:AuxiliarGKSL})
and 
the equilibrium state  of (\ref{eq:3.2g}).

\begin{corollary}
 \label{cor:Ineq-Conv}
Let 
$\left( \rho_t \right)_{t \geq 0}$ be the $N$-weak solution of (\ref{eq:AuxiliarGKSL})
with $\alpha, \beta : \left[ 0 , \infty \right[ \rightarrow \mathbb{C}$ continuous functions.
Suppose that 
$\rho^R_t \left( \cdot \right)$ is the one-parameter semigroup of contractions 
described by the $N$-weak solutions of  (\ref{eq:3.2g}),
$ \alpha_R \left( u \right) = \alpha \left( u \right) - \alpha_0 $,
and that
$ \beta_R  \left( u \right) = \beta \left( u \right) - \beta_0 $.
Then, for all $t  \geq s \geq 0$ we have
\begin{equation}
\label{eq:30n}
 \begin{aligned}
 \Tr \left( \left\vert  \rho_t  - \varrho^f_{\infty} \otimes \varrho^a_{\infty} \right\vert  \right)
& \leq 
\Tr \left( \left\vert   \rho^R_{t-s} \left( \rho_s  \right)   -  \varrho^f_{\infty} \otimes \varrho^a_{\infty} \right\vert  \right)
+
4  \int_s^t \left\vert  \alpha_R \left( u \right) \right\vert \sqrt{ \Tr \left( \rho_u \, N  \right) +1 } \, du
\\
& \quad
+ 
2   \left( \left\Vert  \sigma^{-} \right\Vert +   \left\Vert  \sigma^{+} \right\Vert \right)
 \int_s^t \left\vert \beta_R  \left( u \right) \right\vert du ,
 \end{aligned}
\end{equation}
where 
$
\varrho^f_{\infty} 
 =
\ketbra{ \mathcal{E} \left( \frac{ \alpha_0 }{ \kappa  + \mathrm{i} \omega } \right)} 
{ \mathcal{E} \left( \frac{ \alpha_0 }{ \kappa  + \mathrm{i} \omega } \right)}
$
with
$ \mathcal{E} \left( \cdot \right) $ the coherent vector given by \eqref{eq:Coherent-Vector},
and
\[
\varrho^a_{\infty}
=
\begin{pmatrix}
 \frac{1}{2} +  \frac{d \left( \gamma^2 + \omega^2 \right) }{ 2  \left( \gamma^2 + \omega^2 + 2  \left\vert \beta_0 \right\vert ^2 \right) }
 &
 \frac{d  \beta_0 \left( \gamma - \mathrm{i} \omega \right) }{\gamma^2 + \omega^2 + 2  \left\vert \beta_0 \right\vert ^2 }
 \\
 \frac{d  \bar{\beta_0} \left( \gamma + \mathrm{i} \omega \right) }{\gamma^2 + \omega^2 + 2  \left\vert \beta_0 \right\vert ^2 }
 &
  \frac{1}{2} -  \frac{d \left( \gamma^2 + \omega^2 \right) }{ 2 \left( \gamma^2 + \omega^2 + 2  \left\vert \beta_0 \right\vert ^2 \right)}
\end{pmatrix} .
\]

\end{corollary}
 
\begin{proof}
 Deferred to Section  \ref{sec:Proof:Ineq-Conv}.
\end{proof}

In view of (\ref{eq:30n}),
we now obtain the rate of convergence of the solution of (\ref{eq:3.2g}) 
to its equilibrium state $\varrho^f_{\infty} \otimes \varrho^a_{\infty} $.
For this purpose, 
we make use of  specific features of (\ref{eq:3.2g}),
together with mathematical techniques for proving 
the exponential convergence to the equilibrium state of  a quantum Markov semigroup
(see, e.g.,  \cite{AccardiFagnolaHachicha2006,CarboneFagnola2000,Carbone2008}).
We also employ Theorem \ref{th:ConvAuxiliarQME} and Corollary \ref{cor:ConvEMCSimple}
to assure the uniqueness of the equilibrium point and the limit cycle of the non-linear laser equation (\ref{eq:Laser1}),
in the proofs of Theorems \ref{th:StatState-LaserE} and  \ref{th:FreeS-LaserE}.

\begin{theorem}
 \label{th:ConvAuxiliarQME} 
Let  $\rho^R_t \left( \varrho \right)$ be the $N$-weak solution of (\ref{eq:3.2g})
with
$  \varrho  \in \mathfrak{L}_{1, N}^{+} \left( \ell^2(\mathbb{Z}_+)\otimes \mathbb{C}^2 \right) $.
Then 
\begin{equation}
\label{eq:5.1}
\Tr \left( \left\vert \rho^R_t \left( \varrho \right) - \varrho^f_{\infty} \otimes \varrho^a_{\infty} \right| \right)
\leq
12 \, \hbox{\rm e}^{- \gamma  t}    \left( 1 +  \left\vert d \right\vert \right)
+
\hbox{\rm e}^{- \kappa t}  \left( 
\frac{2 \left\vert  \alpha_0  \right\vert}{ \sqrt{ \kappa^2   +  \omega^2 }}
+
4  \sqrt{ \Tr \left( \varrho  N  \right) } \right)
\end{equation}
for all $ t \geq 0$,
where
$ \varrho^f_{\infty}$ and $\varrho^a_{\infty}$ are defined as in Corollary \ref{cor:Ineq-Conv}.
\end{theorem}

\begin{proof}
 Deferred to Section  \ref{sec:Proofth:th:ConvAuxiliarQME}.
\end{proof}

Applying Theorem \ref{th:ConvAuxiliarQME} with $\alpha_0 = \beta_0 = 0$ we get:

\begin{corollary}
\label{cor:ConvEMCSimple}
 Suppose that $\left( \rho^h_t \left( \varrho \right) \right)_{t \geq 0}$ is the $N$-weak solution of 
\begin{equation}
\label{eq:3.2}
\frac{d}{dt}  \rho^h_t \left( \varrho \right)
= \mathcal{L}_{\star}^h \, \rho^h_t \left( \varrho \right)
\hspace{1cm}
\forall t \geq 0,
\hspace{1cm}
\rho^h_0 \left( \varrho \right) = \varrho  \in \mathfrak{L}_{1, N}^{+} \left( \ell^2(\mathbb{Z}_+)\otimes \mathbb{C}^2 \right) ,
\end{equation}
where
$ \mathcal{L}_{\star}^h $ is as in (\ref{eq:3.21}).
Let $\varrho_{\infty}$ be given by (\ref{eq:I10}).
Then, for all $t \geq 0$ we have
\begin{equation*}
\Tr \left( \left\vert \rho^h_t \left( \varrho \right) - \varrho_{\infty}  \right\vert  \right) 
\leq
12 \, \hbox{\rm e}^{- \gamma  t}    \left( 1 +  \left\vert d \right\vert \right)
+
4 \, \hbox{\rm e}^{- \kappa t}  \sqrt{ \Tr \left( \varrho \,  N  \right) } .
\end{equation*}
\end{corollary}

The proofs of Theorems \ref{th:LongTime} and \ref{th:LimitCycle}
are mainly based on a perturbation method.
Applying Corollary \ref{cor:Ineq-Conv} we approximate the solution of (\ref{eq:Laser1}), 
resp. an unitary transformation of it,
by the quantum evolution $\rho^R_t$ corresponding to $\alpha_0 = \beta_0 = 0$,
resp. $\omega = 0$  and suitable parameters $\alpha_0$, $\beta_0$.
Then,
 using  Theorem \ref{th:ConvAuxiliarQME}, or Corollary \ref{cor:ConvEMCSimple}, 
 we obtain the stability of the equilibrium point or the limit cycle of (\ref{eq:Laser1}),
as appropriate. 
Other perturbation techniques for  GKSL quantum master equations have been developed to treat, for instance,
the  Markov property of quantum Markov semigroups \cite{KumarSinhaSrivastava2020},
the adiabatic elimination (see, e.g., \cite{Forni2018AdiabaticEF} and references therein), 
and the estimation of  the steady-state density matrix  (see, e.g., \cite{Li2014} and references therein).

\section{Proofs}
\label{sec:Proofs}

\subsection{Proofs of theorems from  Section \ref{sec:LinearQMEs}}

\subsubsection{Proof of Theorem \ref{th:ConvEqHE}}
\label{sec:Proof:th:ConvEqHE}

\begin{proof}[Proof of Theorem \ref{th:ConvEqHE}]
Let  $\varrho$ be a non-negative trace-class operator on $ \ell^2(\mathbb{Z}_+)   \otimes  \mathbb{C}^2 $.
According to Lemma 7.10 of \cite{MoraAP} we have that there exists a sequence 
of $N$-regular non-negative operators $\varrho_n$ such as 
$ \lim_{n \rightarrow + \infty} \Tr \left( \left\vert  \varrho - \varrho_n  \right\vert  \right) = 0$,
where  $N$ is the number operator. 
Now,
\begin{align*}
 & \Tr \left( \left\vert   \rho^R_t \left( \varrho \right) - \rho^R_s \left( \varrho \right) \right\vert  \right)
 \\
 & \leq
\Tr \left( \left\vert   \rho^R_t \left( \varrho \right) - \rho^R_t \left( \varrho_n \right) \right\vert  \right)
+
\Tr \left( \left\vert   \rho^R_t \left( \varrho_n \right) - \rho^R_s \left( \varrho_n \right) \right\vert  \right)
+
\Tr \left( \left\vert   \rho^R_s \left( \varrho_n \right) - \rho^R_s \left( \varrho \right) \right\vert  \right)
\\
 & \leq
2 \, \Tr \left( \left\vert  \varrho - \varrho_n  \right\vert  \right)
+
\Tr \left( \left\vert   \rho^R_t \left( \varrho_n \right) - \rho^R_s \left( \varrho_n \right) \right\vert  \right) .
\end{align*}
From, e.g.,  \cite{FagMora2019} it follows that the hypothesis of  Theorem  4.3 of \cite{MoraAP} holds,
and hence  
$
\lim_{s \rightarrow t}  \Tr \left( \left\vert   \rho^R_t \left( \varrho_n \right) - \rho^R_s \left( \varrho_n \right) \right\vert  \right)
= 0
$.
This leads to  
\begin{equation}
\label{eq:8.26}
 \lim_{s \rightarrow t}  \Tr \left( \left\vert   \rho^R_t \left( \varrho \right) - \rho^R_s \left( \varrho \right) \right\vert  \right)
= 0.
\end{equation}
Decomposing the real and imaginary parts of an element of $ \mathfrak{L}_{1}\left( \ell^2(\mathbb{Z}_+)   \otimes  \mathbb{C}^2  \right)$ 
into positive and negative parts (see, e.g., proof of Theorem 4.1 of \cite{MoraAP})
we find that 
(\ref{eq:8.26}) holds for any $ \varrho \in \mathfrak{L}_{1}\left( \ell^2(\mathbb{Z}_+)   \otimes  \mathbb{C}^2  \right) $.
\end{proof}

\subsubsection{Proof of Theorem \ref{th:VariationConstants}}
\label{sec:Proof:th:VariationConstants}

Let $\mathcal{L}_{\star}^{full} \left( t \right)$, $\mathcal{L}_{\star}^{R} \left( t \right)$ 
be the linear operators in  $\mathfrak{L}_{1}\left(  \ell^2 \left(\mathbb{Z}_+ \right) \otimes \mathbb{C}^2  \right) $
defined by 
\[
\mathcal{L}_{\star}^{full} \left( t \right) \, \varrho
=
 \mathcal{L}_{\star}^h \, \varrho
+
\left[ 
 \alpha \left( t \right) a^\dagger -  \overline{\alpha \left( t \right)} a 
+  \overline{\beta  \left( t \right) } \sigma^{-}  - \beta  \left( t \right) \sigma^+  ,
\varrho \right] 
\]
and 
$
\mathcal{L}_{\star}^{R} \left( t \right) \, \varrho
=
\mathcal{L}_{\star}^h \, \varrho
+
\left[ 
\alpha_0 \, a^\dagger -  \overline{\alpha_0 } \, a 
+ \overline{\beta_0 } \, \sigma^{-}  - \beta_0  \,\sigma^+  ,
\varrho  \right]
$
for any $N$-regular density operator 
$ \varrho \in \mathfrak{L}_{1}\left(  \ell^2 \left(\mathbb{Z}_+ \right) \otimes \mathbb{C}^2  \right) $;
we recall that $ \mathcal{L}_{\star}^h $ is given by (\ref{eq:3.21}).
Suppose that $\rho_t$ is the $N^p$-weak solution of (\ref{eq:AuxiliarGKSL})
(see, e.g., \cite{FagMora2019}). 
Therefore, $\rho_t$ satisfies 
\[
\frac{d}{dt}  \rho_t 
=  \mathcal{L}_{\star}^{full} \left( t \right) \, \rho_t 
\]
in the  $N^p$-weak sense.
According to \cite{FagMora2019} we have that
$t \mapsto \Tr \left( A \,  \mathcal{L}_{\star}^{full} \left( t \right) \, \rho_t  \right) $
is a continuous function whenever 
$A:   \ell^2 \left(\mathbb{Z}_+ \right) \otimes \mathbb{C}^2 \rightarrow   \ell^2 \left(\mathbb{Z}_+ \right) \otimes \mathbb{C}^2$
is bounded and $ \rho_0$ is $N$-regular,
which was obtained by using probabilistic techniques.
A delicate issue in the proof of Theorem \ref{th:VariationConstants}
is to establish the continuity of 
$
t \mapsto \mathcal{L}_{\star}^{full} \left( t \right) \, \rho_t
$
and
$
t \mapsto \mathcal{L}_{\star}^{R} \left( t \right) \, \rho_t
$
with respect to the trace norm.
To this end,
we first restrict the initial condition to be $N^2$-regular.
As in  \cite{FagMora2019},
we next profit from the  probabilistic representation of (\ref{eq:AuxiliarGKSL}) given by 
\begin{equation}
 \label{eq:ProbRepAQME}
  \rho_t = \mathbb{E} \left|X_{t} \left( \xi \right)\right\rangle \left\langle X_{t} \left( \xi \right)\right| 
\end{equation}
with $X_{t} \left( \xi \right)$ being the solution of the linear stochastic Schr\"odinger equation on  $ \ell^2(\mathbb{Z}_+)  \otimes   \mathbb{C}^2 $:
\begin{equation}
\label {eq:SSE}
X_{t}\left( \xi \right) 
= \xi +\int_{0}^{t}G \left( s \right) X_{s}\left( \xi \right) ds 
+ \sum_{ \ell =1}^{3 }\int_{0}^{t}
L_{\ell}  \, X_{s}\left( \xi \right) dW_{s}^{\ell} 
\hspace{2cm}
\forall t\geq 0 ,
\end{equation}
where
$G \left( t \right), L_k  $ are the linear operators  in  $  \ell^2(\mathbb{Z}_+)\otimes \mathbb{C}^2$ 
defined by 
$G \left( t \right)  
 =
-  \mathrm{i} \, H \left( t \right) -  \sum_{\ell = 1}^3 L_{\ell}^* L_{\ell} / 2
$,
$ L_1  =  \sqrt{2 \kappa} \, a $, 
$L_2 =  \sqrt{ \gamma \left(1-d \right) } \, \sigma^{-} $,
$L_3 =  \sqrt{ \gamma \left(1 +d \right) } \, \sigma^{+} $,
\[
H \left( t \right) 
 =
 \omega \left(  a^\dagger a  +\sigma^{3} / 2 \right) 
+ \mathrm{i}  \left( \alpha \left( t \right) a^\dagger -  \overline{\alpha \left( t \right)} a 
+  \overline{\beta  \left( t \right) } \sigma^{-}  - \beta  \left( t \right) \sigma^+ \right)
,
\]
and 
$W^1, W^2, \ldots$ are independent real Brownian motions  on a filtered complete probability 
space $\left( \Omega ,\mathfrak{F}, \left(\mathfrak{F}_{t}\right) _{t\geq 0},\mathbb{P}\right) $
(see Theorem 6 of \cite{FagMora2019} for details).

\begin{lemma}
\label{le:Continuity}
Suppose that $\alpha, \beta : \left[ 0 , \infty \right[ \rightarrow \mathbb{C}$ are continuous functions,
and that $\left( \rho_t \right)_{t \geq 0}$ is a $N^2$-weak solution of (\ref{eq:AuxiliarGKSL}).
Then, for any $t \geq 0$ we have
$
\mathcal{L}_{\star}^{full} \left( s \right) \, \rho_s 
 \longrightarrow_{s \rightarrow t} 
 \mathcal{L}_{\star}^{full} \left( t \right) \, \rho_t ,
$
and
$
\mathcal{L}_{\star}^{R} \, \rho_s 
 \longrightarrow_{s \rightarrow t} 
 \mathcal{L}_{\star}^{R} \, \rho_t 
$,
where both limits are taken in $\mathfrak{L}_{1}\left(  \ell^2 \left(\mathbb{Z}_+ \right) \otimes \mathbb{C}^2  \right) $.
\end{lemma}

\begin{proof}
Since  $\rho_0$ is $N^2$-regular,
there exists  $\xi \in L_{N^2}^{2}\left( \mathbb{P},  \ell^2 \left(\mathbb{Z}_+ \right) \otimes \mathbb{C}^2 \right) $ such that
$
\rho_0 = \mathbb{E}\left\vert \xi \rangle \langle \xi \right\vert 
$
(see, e.g., Theorem 3.1 of \cite{MoraAP}).
Then 
$
 \rho_t = \mathbb{E} \left|X_{t} \left( \xi \right)\right\rangle \left\langle X_{t} \left( \xi \right)\right| 
$,
where $X_{t} \left( \xi \right)$ is the strong $N^2$-solution of (\ref{eq:SSE})
(see, e.g., \cite{FagMora2019}).

Using $X_{t} \left( \xi \right) \in L_{N^2}^{2}\left( \mathbb{P}, \ell^2 \left(\mathbb{Z}_+ \right) \otimes \mathbb{C}^2 \right) $
we deduce that 
\[
Y_t 
:= N \xi +\int_{0}^{t} N G \left( s \right) X_{s}\left( \xi \right) ds 
+ \sum_{\ell=1}^{3}\int_{0}^{t} N L_\ell  X_{s}\left( \xi \right) dW_{s}^{\ell} 
\hspace{2cm}
\forall t \geq 0 
\]
is a well-defined  continuous stochastic process.
As $N$ is a closed operator in $ \ell^2 \left(\mathbb{Z}_+ \right) \otimes \mathbb{C}^2$ we have
$ Y_t = N X_{t} \left( \xi \right)$ for all $t \geq 0$ $\mathbb{P}$-a.s.
(see, e.g., Proposition 4.15 of \cite{DaPrato}).
Moreover, 
$
 \mathbb{E} \left( \sup_{s \in \left[0, t+1 \right]}  \left\Vert  Y_s  \right\Vert ^{2} \right)
 < \infty
$
and
$
 \mathbb{E} \left( \sup_{s \in \left[0, t+1 \right]}  \left\Vert  X_s\left( \xi \right)  \right\Vert ^{2} \right)
 < \infty
$
(see, e.g., Theorem 4.2.5 of  \cite{Prevot}).
Then, applying the dominated convergence theorem we get
\[
\lim_{n \rightarrow + \infty} \mathbb{E} \left\Vert  N X_{s_n} \left( \xi \right) - N X_{t} \left( \xi \right)  \right\Vert^2
=
\lim_{n \rightarrow + \infty} \mathbb{E} \left\Vert  Y_{s_n} - Y_{t} \right\Vert^2
= 
0 
\]
and  
$
\lim_{n \rightarrow + \infty} \mathbb{E} \left\Vert  X_{s_n} \left( \xi \right) -  X_{t} \left( \xi \right)  \right\Vert^2
= 0
$,
where  $s_n \rightarrow t$  as $n \rightarrow + \infty$.
Therefore,
\begin{equation}
\label{eq:8.27}
 \lim_{s \rightarrow t} \mathbb{E} \left\Vert   X_{s} \left( \xi \right) -  X_{t} \left( \xi \right) \right\Vert_N^2 = 0
 \hspace{2cm} \forall t \geq 0 .
\end{equation}

Suppose  that $A,B$ are linear operators in $ \ell^2 \left(\mathbb{Z}_+ \right) \otimes \mathbb{C}^2  $,
which are relatively bounded with respect to $N$.
Then 
$
\mathbb{E} \left\vert  A \,   X_{s}\left( \xi\right) \rangle  \langle B \,  X_{t} \left( \xi \right) \right\vert
$
is well defined as a Bochner integral in  $ \ell^2 \left(\mathbb{Z}_+ \right) \otimes \mathbb{C}^2  $
for all $ s, t  \geq 0$
(see, e.g., \cite{MoraAP}).
Since
$  
\left\Vert     \, \left\vert  x \rangle  \langle  y  \right\vert    \,  
\right\Vert_{\mathfrak{L}_{1}\left(  \ell^2 \left(\mathbb{Z}_+ \right) \otimes \mathbb{C}^2  \right) }
=
\left\Vert   x  \right\Vert   \left\Vert  y \right\Vert  
$
for any $x, y \in  \ell^2 \left(\mathbb{Z}_+ \right) \otimes \mathbb{C}^2  $,
\begin{align*}
&  
\left\Vert    
\mathbb{E} \left\vert  A \,   X_{s}\left( \xi\right) \rangle  \langle B \,  X_{s}\left( \xi \right) \right\vert 
-
\mathbb{E} \left\vert  A \,   X_{t}\left( \xi\right) \rangle  \langle B \,  X_{t} \left( \xi \right) \right\vert 
\right\Vert_{\mathfrak{L}_{1}\left(  \ell^2 \left(\mathbb{Z}_+ \right) \otimes \mathbb{C}^2  \right) }
 \\
& \leq
\left\Vert    
\mathbb{E} \left\vert  A \,   X_{s}\left( \xi\right) -  A \,   X_{t}\left( \xi\right) \rangle  \langle B \,  X_{s}\left( \xi \right) \right\vert 
\right\Vert_{\mathfrak{L}_{1}\left(  \ell^2 \left(\mathbb{Z}_+ \right) \otimes \mathbb{C}^2  \right) }
\\
& \quad +
\left\Vert    
\mathbb{E} \left\vert  A \,   X_{t}\left( \xi\right)  \rangle  \langle B \,  X_{s} \left( \xi\right) -  B \,   X_{t}\left( \xi\right) \right\vert 
\right\Vert_{\mathfrak{L}_{1}\left(  \ell^2 \left(\mathbb{Z}_+ \right) \otimes \mathbb{C}^2  \right) }
\\
& \leq   
\mathbb{E} \left\Vert  \, \left\vert  A \,   X_{s}\left( \xi\right) -  A \,   X_{t}\left( \xi\right) \rangle  \langle B \,  X_{s}\left( \xi \right) \right\vert  \,
\right\Vert_{\mathfrak{L}_{1}\left(  \ell^2 \left(\mathbb{Z}_+ \right) \otimes \mathbb{C}^2  \right) }
\\
& \quad +
  \mathbb{E} \left\Vert   \, \left\vert  A \,   X_{t}\left( \xi\right)  \rangle  \langle B \,  X_{s} \left( \xi\right) -  B \,   X_{t}\left( \xi\right) \right\vert  \,
\right\Vert_{\mathfrak{L}_{1}\left(  \ell^2 \left(\mathbb{Z}_+ \right) \otimes \mathbb{C}^2  \right) }
 \\
 & =
 \mathbb{E} \left( 
 \left\Vert  A \,   X_{s}\left( \xi\right) -  A \,   X_{t}\left( \xi\right)  \right\Vert  \left\Vert B \,  X_{s}\left( \xi \right) \right\Vert 
 \right)
 +
\mathbb{E} \left(
 \left\Vert  A \,   X_{t}\left( \xi\right)  \right\Vert   \left\Vert B \,  X_{s} \left( \xi\right) -  B \,   X_{t}\left( \xi\right) \right\Vert 
 \right)
\end{align*}
for  any $s,t \geq 0$.
Therefore, 
using the Cauchy-Schwarz inequality gives 
\begin{align*}
&
\left\Vert    
\mathbb{E} \left\vert  A \,   X_{s}\left( \xi\right) \rangle  \langle B \,  X_{s}\left( \xi \right) \right\vert 
-
\mathbb{E} \left\vert  A \,   X_{t}\left( \xi\right) \rangle  \langle B \,  X_{t} \left( \xi \right) \right\vert 
\right\Vert_{\mathfrak{L}_{1}\left(  \ell^2 \left(\mathbb{Z}_+ \right) \otimes \mathbb{C}^2  \right) }
 \\
 & \leq
 \sqrt{\mathbb{E} \left( \left\Vert  A  \,   X_{s}\left( \xi\right) -  A \,   X_{t}\left( \xi\right)  \right\Vert^2  \right)}
 \sqrt{\mathbb{E} \left( \left\Vert B \,  X_{s}\left( \xi \right) \right\Vert^2 \right) }
 \\
 & \quad +
 \sqrt{\mathbb{E} \left( \left\Vert A \,  X_{t} \left( \xi \right) \right\Vert^2 \right) }
 \sqrt{\mathbb{E} \left( \left\Vert  B  \,   X_{s}\left( \xi\right) -  B \,   X_{t}\left( \xi\right)  \right\Vert^2  \right)} ,
  \end{align*}
and so combining (\ref{eq:8.27}) with 
$
  \sup_{s \in \left[0, t+1 \right]} \mathbb{E} \left( \left\Vert  X_s\left( \xi \right)  \right\Vert_N^{2} \right) < \infty
$
we get
\begin{equation}
\label{eq:8.277}
 \left\Vert    
\mathbb{E} \left\vert  A \,   X_{s}\left( \xi\right) \rangle  \langle B \,  X_{s}\left( \xi \right) \right\vert 
-
\mathbb{E} \left\vert  A \,   X_{t}\left( \xi\right) \rangle  \langle B \,  X_{t} \left( \xi \right) \right\vert 
\right\Vert_{\mathfrak{L}_{1}\left(  \ell^2 \left(\mathbb{Z}_+ \right) \otimes \mathbb{C}^2  \right) }
 \longrightarrow_{s \rightarrow t} 0 .
\end{equation}

According to
$
 \rho_t = \mathbb{E} \left|X_{t} \left( \xi \right)\right\rangle \left\langle X_{t} \left( \xi \right)\right| 
$
we have that
$
G  \left( t \right) \rho_t  =  \mathbb{E} \left\vert   G  \left( t \right)  X_{t}\left( \xi\right) \rangle  \langle X_{t}\left( \xi\right)  \right\vert 
$,
$
 \rho_t \, G  \left( t \right)^*
=
 \mathbb{E}  \left\vert X_{t}\left( \xi\right)  \rangle  \langle G  \left( t \right)  X_{t}\left( \xi\right) \right\vert 
$
and
$
L_{\ell} \, \rho_t \, L_{\ell}^*
=
\mathbb{E} \left\vert  L_{\ell}   X_{t}\left( \xi\right) \rangle  \langle L_{\ell}  X_{t}\left( \xi \right) \right\vert
$
(see, e.g., Theorem 3.2 of \cite{MoraAP}).
Hence,
\begin{align*}
 \mathcal{L}_{\star}^{full} \left( t \right) \, \rho_t
& =
G  \left( t \right) \rho_t  + \rho_t  \, G  \left( t \right)^* +  \sum_{\ell=1}^{3} L_{\ell} \, \rho_t \, L_{\ell}^*
\\
& =
\mathbb{E} \left\vert   G  \left( t \right)  X_{t}\left( \xi\right) \rangle  \langle X_{t}\left( \xi\right)  \right\vert 
+
\mathbb{E}  \left\vert X_{t}\left( \xi\right)  \rangle  \langle G  \left( t \right)  X_{t}\left( \xi\right) \right\vert 
+ 
\sum_{\ell=1}^{3}\mathbb{E} \left\vert  L_{\ell}   X_{t}\left( \xi\right) \rangle  \langle L_{\ell}  X_{t}\left( \xi \right) \right\vert .
\end{align*}
Since $\alpha, \beta : \left[ 0 , \infty \right[ \rightarrow \mathbb{C}$ are continuous functions
and 
 $\sigma^{-}, \sigma^{+}, \sigma^{3} \in \mathfrak{L}\left(   \ell^2 \left(\mathbb{Z}_+ \right) \otimes \mathbb{C}^2\right)$,
combining (\ref{eq:8.277}) with the fact that 
$ a^\dagger a $, $ a^\dagger $ and $ a $ are  relatively bounded with respect to $N$
we obtain that
$
t \mapsto
 \mathbb{E} \left\vert   G  \left( t \right)  X_{t}\left( \xi\right) \rangle  \langle X_{t}\left( \xi\right)  \right\vert 
+
 \mathbb{E}  \left\vert X_{t}\left( \xi\right)  \rangle  \langle G  \left( t \right)  X_{t}\left( \xi\right) \right\vert 
 + 
 \sum_{\ell=1}^{3}\mathbb{E} \left\vert  L_{\ell}   X_{t}\left( \xi\right) \rangle  \langle L_{\ell}  X_{t}\left( \xi \right) \right\vert 
$
is a continuous function from $\left[ 0 , \infty \right[$ to 
$\mathfrak{L}_{1}\left(  \ell^2 \left(\mathbb{Z}_+ \right) \otimes \mathbb{C}^2  \right)$.
Hence 
\begin{equation*}
 \mathcal{L}_{\star}^{full} \left( s \right) \, \rho_s 
 \longrightarrow_{s \rightarrow t} 
 \mathcal{L}_{\star}^{full} \left( t \right) \, \rho_t
 \hspace{2cm}
 \text{ in }  \mathfrak{L}_{1}\left(  \ell^2 \left(\mathbb{Z}_+ \right) \otimes \mathbb{C}^2  \right) .
\end{equation*}

Similarly,
\begin{align*}
\mathcal{L}_{\star}^{R} \left( t \right) \, \rho_t 
& =
 \mathbb{E} \left\vert   G^R  \left( t \right)  X_{t}\left( \xi\right) \rangle  \langle X_{t}\left( \xi\right)  \right\vert 
+
 \mathbb{E}  \left\vert X_{t}\left( \xi\right)  \rangle  \langle G^R  \left( t \right)  X_{t}\left( \xi\right) \right\vert 
 +
 \sum_{\ell=1}^{3}\mathbb{E} \left\vert  L_{\ell}   X_{t}\left( \xi\right) \rangle  \langle L_{\ell}  X_{t}\left( \xi \right) \right\vert ,
\end{align*}
where $G^R  \left( t \right)$ is defined by $G \left( t \right)$ with $\alpha \left( t \right) = \alpha_0$ and $\beta \left( t \right) = \beta_0$.
Then 
$
\mathcal{L}_{\star}^{R} \left( s \right) \, \rho_s 
$ 
converges in $\mathfrak{L}_{1}\left(  \ell^2 \left(\mathbb{Z}_+ \right) \otimes \mathbb{C}^2  \right)$
to 
$ \mathcal{L}_{\star}^{R} \left( t \right) \, \rho_t $
as $s \rightarrow t$.
\end{proof}

Now,
applying functional analysis techniques we show the assertion of Theorem \ref{th:VariationConstants}
in case $ \rho_0$ is a $N^2$-regular density operator.

\begin{lemma}
\label{lem:StrongSolution}
Under the assumptions of Lemma \ref{le:Continuity},  
\begin{equation*}
\lim_{s \rightarrow t} \frac{\rho_{s} - \rho_{t}}{s-t}
=
\mathcal{L}_{\star}^{full} \left( t \right) \, \rho_t 
\hspace{2cm} \text{ in } \mathfrak{L}_{1}\left(  \ell^2 \left(\mathbb{Z}_+ \right) \otimes \mathbb{C}^2  \right) .
\end{equation*} 
\end{lemma}

\begin{proof}
Fix $t \geq 0$.
Since $\left( \rho_t \right)_{t \geq 0}$ satisfies (\ref{eq:AuxiliarGKSL}) in the Bochner integral sense 
(see, e.g., Theorem 6 of \cite{FagMora2019}),
\begin{equation}
 \label{eq:2.28}
 \frac{\rho_{s} - \rho_{t}}{s-t} -  \mathcal{L}_{\star}^{full} \left( t \right) \, \rho_t
 = 
\frac{1}{ s - t } \int_{t}^{s} \left( 
\mathcal{L}_{\star}^{full} \left( u \right) \, \rho_u - \mathcal{L}_{\star}^{full} \left( t \right) \, \rho_t \right) du 
\qquad \text{in }  \mathfrak{L}_{1}\left( \mathfrak{h}\right) 
\end{equation}
for all  $s \geq 0$ with $s \neq t$ .
From Lemma \ref{le:Continuity}  it follows the continuity of the function 
$
 u \mapsto \mathcal{L}_{\star}^{full} \left( u \right) \, \rho_u - \mathcal{L}_{\star}^{full} \left( t \right) \, \rho_t 
$  
defined from  $\left[ 0 , + \infty \right[$ to 
$\mathfrak{L}_{1}\left(  \ell^2 \left(\mathbb{Z}_+ \right) \otimes \mathbb{C}^2  \right)$.
Therefore, 
\[
\lim_{s \rightarrow t} 
\frac{1}{ s - t} \int_{t }^{ s} \left( 
\mathcal{L}_{\star}^{full} \left( u \right) \, \rho_u - \mathcal{L}_{\star}^{full} \left( t \right) \, \rho_t \right) du
=
0 
\hspace{1cm} \text{ in } \mathfrak{L}_{1}\left(  \ell^2 \left(\mathbb{Z}_+ \right) \otimes \mathbb{C}^2  \right) ,
\]
and so  the lemma follows from (\ref{eq:2.28}). 
\end{proof}

\begin{lemma}
\label{th:VariationConstantsRegular}
Assume the hypothesis of Lemma \ref{le:Continuity}.
Then, for all $t \geq s \geq 0$ we have 
\begin{equation*}
 \rho_t 
 =
 \rho^R_{t-s} \left( \rho_s  \right) 
 + \int_s^t  \rho^R_{t-u} \left( 
 \left[  \left( \alpha_R \left( u \right) a^\dagger -  \overline{\alpha_R \left( u \right)} a \right)
\hspace{-1.2pt} + \hspace{-1.2pt}
\left( \overline{\beta_R  \left( u \right) } \sigma^{-}  - \beta_R  \left( u \right) \sigma^+ \right)   ,
\rho_u \right]  \right) du ,
\end{equation*}
where
$ \rho^R_t \left( \cdot \right)$ is as in Theorem \ref{th:ConvEqHE},
$ \alpha_R \left( u \right) = \alpha \left( u \right) - \alpha_0 $
and 
$ \beta_R  \left( u \right) = \beta \left( u \right) - \beta_0 $. 
\end{lemma}

\begin{proof}
Consider $t > s \geq 0$.
For any non-zero real number $\Delta$ such that $-s \leq \Delta < t-s$ we have
\begin{align*}
&
 \frac{1}{\Delta} \left( 
 \rho^R_{t - \left(s+ \Delta \right)}  \left( \rho_{s+\Delta} \right)
 - \rho^R_{t -s}  \left( \rho_{s} \right)
 \right)
 +  \mathcal{L}_{\star}^R \, \rho^R_{t -s}  \left( \rho_{s} \right)
 -
 \rho^R_{t -s}  \left( \mathcal{L}_{\star}^{full} \left( s \right) \,  \rho_{s} \right)
 \\
 & =
 \rho^R_{t - \left(s+ \Delta \right)}  \left( 
 \frac{1}{ \Delta} \left( \rho_{s+ \Delta} - \rho_{s}\right)  -  \mathcal{L}_{\star}^{full} \left( s \right) \,  \rho_{s}
 \right)
  +
  \rho^R_{t - \left(s+  \Delta \right)}  \left( \mathcal{L}_{\star}^{full} \left( s \right) \,  \rho_{s}  \right)
  -
   \rho^R_{t - s}  \left( \mathcal{L}_{\star}^{full} \left( s \right) \,  \rho_{s} \right)
 \\
 & \quad +
 \frac{1}{ \Delta} \left( 
 \rho^R_{t - \left(s+ \Delta \right)}  \left( \rho_{s} \right)
 - \rho^R_{t -s}  \left( \rho_{s} \right)
 \right)
 +  \mathcal{L}_{\star}^R \, \rho^R_{t -s}  \left( \rho_{s} \right) .
\end{align*}

Since $\rho^R_{t - \left(s+ \Delta \right)}$ is a contraction acting on 
$\mathfrak{L}_{1}\left(  \ell^2 \left(\mathbb{Z}_+ \right) \otimes \mathbb{C}^2  \right)$,
\begin{align*}
 & 
 \Tr \left( \left\vert  \rho^R_{t - \left(s+ \Delta \right)}  \left( 
 \frac{1}{ \Delta} \left( \rho_{s+ \Delta} - \rho_{s}\right)  -  \mathcal{L}_{\star}^{full} \left( s \right) \,  \rho_{s}
 \right) \right\vert  \right)
  \leq
 \Tr \left( \left\vert \frac{1}{ \Delta} \left( \rho_{s+ \Delta} - \rho_{s}\right)  -  \mathcal{L}_{\star}^{full} \left( s \right) \,  \rho_{s}  \right\vert  \right),
\end{align*}
and so applying Lemma \ref{lem:StrongSolution} yields 
\begin{equation*}
\Tr \left( \left\vert  \rho^R_{t - \left(s+ \Delta \right)}  \left( 
 \frac{1}{ \Delta} \left( \rho_{s+ \Delta} - \rho_{s}\right)  -  \mathcal{L}_{\star}^{full} \left( s \right) \,  \rho_{s}
 \right) \right\vert  \right)
 \longrightarrow_{ \Delta \rightarrow 0} 0 .
\end{equation*}
The operator $\mathcal{L}_{\star}^R$ is equal to $ \mathcal{L}_{\star}^{full} \left( t \right) $ 
in case
$\alpha \left( t \right) =  \alpha_0 $ and $\beta \left( t \right) = \beta_0 $ for all $t \geq 0$.
Hence,  
using  Lemma \ref{lem:StrongSolution} we deduce that 
$\mathcal{L}_{\star}^R$ coincides with
the infinitesimal generator of the strongly continuous semigroup $\left( \rho^R_t  \right)_{t \geq 0}$
on the subset $\mathfrak{L}_{1,N^2}^{+} \left( \ell^2 \left(\mathbb{Z}_+ \right) \otimes \mathbb{C}^2 \right) $,
as well as 
\begin{equation*}
 \Tr \left( \left\vert   
 \frac{1}{ \Delta} \left( 
 \rho^R_{t - \left(s+ \Delta \right)}  \left( \rho_{s} \right)
 - \rho^R_{t -s}  \left( \rho_{s} \right)
 \right)
 +  \mathcal{L}_{\star}^R \, \rho^R_{t -s}  \left( \rho_{s} \right)
 \right\vert  \right)
 \longrightarrow_{ \Delta \rightarrow 0} 0 .
\end{equation*}
Moreover,
the strong continuity of the semigroup $\left( \rho^R_u  \right)_{u \geq 0}$ implies 
\begin{equation*}
 \Tr \left( \left\vert 
 \rho^R_{t - \left(s+  \Delta \right)}  \left( \mathcal{L}_{\star}^{full} \left( s \right) \,  \rho_{s}  \right)
  -
 \rho^R_{t - s}  \left( \mathcal{L}_{\star}^{full} \left( s \right) \,  \rho_{s} \right)
   \right\vert  \right)
 \longrightarrow_{ \Delta \rightarrow 0} 0 .
\end{equation*}
Therefore,
$
\frac{1}{\Delta} \left( 
 \rho^R_{t - \left(s+ \Delta \right)}  \left( \rho_{s+\Delta} \right)
 - \rho^R_{t -s}  \left( \rho_{s} \right)
 \right)
 +  \mathcal{L}_{\star}^R \, \rho^R_{t -s}  \left( \rho_{s} \right)
 -
 \rho^R_{t -s}  \left( \mathcal{L}_{\star}^{full} \left( s \right) \,  \rho_{s} \right)
 $
converges to $0$ in $\mathfrak{L}_{1}\left(  \ell^2 \left(\mathbb{Z}_+ \right) \otimes \mathbb{C}^2  \right)$
as $\Delta \rightarrow 0$.
Thus
\begin{equation*}
 \frac{d}{ds} \rho^R_{t -s}  \left( \rho_{s} \right)
 =
\rho^R_{t -s}  \left( \mathcal{L}_{\star}^{full} \left( s \right) \,  \rho_{s} \right)
- \mathcal{L}_{\star}^R \, \rho^R_{t -s}  \left( \rho_{s} \right) 
 =
\rho^R_{t -s}  \left( \mathcal{L}_{\star}^{full} \left( s \right) \,  \rho_{s} \right)
-  \rho^R_{t -s}  \left( \mathcal{L}_{\star}^R \, \rho_{s} \right) ,
\end{equation*}
and consequently 
\begin{equation}
\label{eq:2.29}
 \frac{d}{ds} \rho^R_{t -s}  \left( \rho_{s} \right)
 =
\rho^R_{t -s}  \left( \mathcal{L}_{\star}^{full} \left( s \right) \,  \rho_{s} -  \mathcal{L}_{\star}^R \, \rho_{s} \right) .
\end{equation}

Now, we deduce the continuity of the map 
$
s \mapsto
\rho^R_{t -s}  \left( \mathcal{L}_{\star}^{full} \left( s \right) \,  \rho_{s} -  \mathcal{L}_{\star}^R \, \rho_{s} \right)
$.
The contraction property of  $\rho^R_{t -u}$ leads to 
\begin{align*}
 & 
 \Tr \left( \left\vert  
 \rho^R_{t -s}  \left( \mathcal{L}_{\star}^{full} \left( s \right) \,  \rho_{s} -  \mathcal{L}_{\star}^R \, \rho_{s} \right)
 -
 \rho^R_{t -u}  \left( \mathcal{L}_{\star}^{full} \left( u \right) \,  \rho_{u} -  \mathcal{L}_{\star}^R \, \rho_{u} \right)
  \right\vert  \right)
\\
& \leq
\Tr \left( \left\vert  
 \rho^R_{t -s}  \left( \mathcal{L}_{\star}^{full} \left( s \right) \,  \rho_{s} -  \mathcal{L}_{\star}^R \, \rho_{s} \right)
 -
 \rho^R_{t -u}  \left( \mathcal{L}_{\star}^{full} \left( s \right) \,  \rho_{s} -  \mathcal{L}_{\star}^R \, \rho_{s} \right)
\right\vert  \right)
\\
& \quad \quad +
\Tr \left( \left\vert  
 \rho^R_{t -u}  \left( \mathcal{L}_{\star}^{full} \left( s \right) \,  \rho_{s} -  \mathcal{L}_{\star}^R \, \rho_{s} \right)
 -
  \rho^R_{t -u}  \left( \mathcal{L}_{\star}^{full} \left( u \right) \,  \rho_{u} -  \mathcal{L}_{\star}^R \, \rho_{u} \right)
\right\vert  \right)
\\
& \leq
\Tr \left( \left\vert  
 \rho^R_{t -s}  \left( \mathcal{L}_{\star}^{full} \left( s \right) \,  \rho_{s} -  \mathcal{L}_{\star}^R \, \rho_{s} \right)
 -
 \rho^R_{t -u}  \left( \mathcal{L}_{\star}^{full} \left( s \right) \,  \rho_{s} -  \mathcal{L}_{\star}^R \, \rho_{s} \right)
\right\vert  \right)
\\
& \quad \quad +
\Tr \left( \left\vert  
  \left( \mathcal{L}_{\star}^{full} \left( s \right) \,  \rho_{s} -  \mathcal{L}_{\star}^R \, \rho_{s} \right)
 -
   \left( \mathcal{L}_{\star}^{full} \left( u \right) \,  \rho_{u} -  \mathcal{L}_{\star}^R \, \rho_{u} \right)
\right\vert  \right) .
\end{align*}
According to Lemma \ref{le:Continuity} we have
\[
\lim_{u \rightarrow s} \Tr \left( \left\vert  
  \left( \mathcal{L}_{\star}^{full} \left( s \right) \,  \rho_{s} -  \mathcal{L}_{\star}^R \, \rho_{s} \right)
 -
   \left( \mathcal{L}_{\star}^{full} c\left( u \right) \,  \rho_{u} -  \mathcal{L}_{\star}^R \, \rho_{u} \right)
\right\vert  \right)
= 0 ,
\]
and the strong continuity of $\left( \rho^R_r  \right)_{r \geq 0}$ yields
\[
\lim_{u \rightarrow s} 
\Tr \left( \left\vert  
 \rho^R_{t -s}  \left( \mathcal{L}_{\star}^{full} \left( s \right) \,  \rho_{s} -  \mathcal{L}_{\star}^R \, \rho_{s} \right)
 -
 \rho^R_{t -u}  \left( \mathcal{L}_{\star}^{full} \left( s \right) \,  \rho_{s} -  \mathcal{L}_{\star}^R \, \rho_{s} \right)
\right\vert  \right)
= 0 .
\]
Then, 
$
s \mapsto
\rho^h_{t -s}  \left( \mathcal{L}_{\star}^{full} \left( s \right) \,  \rho_{s} -  \mathcal{L}_{\star}^h \, \rho_{s} \right)
$
is continuous.
By the fundamental theorem of calculus for the Bochner integral,
integrating (\ref{eq:2.29}) gives
\[
\rho^R_{0}  \left( \rho_{t} \right) - \rho^R_{t -s}  \left( \rho_{s} \right)
=
\int_{s}^{t} 
\rho^R_{t -u}  \left( \mathcal{L}_{\star}^{full} \left( u \right) \,  \rho_{u} -  \mathcal{L}_{\star}^R \, \rho_{u} \right) 
 du ,
\]
which is the desired conclusion.
\end{proof}

Now,
we  extend Lemma \ref{th:VariationConstantsRegular} to any $N$-regular initial condition
by combining a limit procedure with the probabilistic representation (\ref{eq:ProbRepAQME}).

\begin{proof}[Proof of Theorem \ref{th:VariationConstants}]
We start by approximating the $N$-regular initial condition $\rho_0$ by 
$N^2$-regular density operators.
As $\rho_0 \in \mathfrak{L}_{1, N}^{+} \left( \ell^2(\mathbb{Z}_+)\otimes \mathbb{C}^2 \right) $
we have that
$
\rho_0 = \mathbb{E}\left\vert \xi \rangle \langle \xi \right\vert 
$
for certain $\xi \in  L_{N}^{2}\left( \mathbb{P},  \ell^2 \left(\mathbb{Z}_+ \right) \otimes \mathbb{C}^2 \right) $
(see, e.g., Theorem 3.1 of \cite{MoraAP}).
Let $pr_n$ denote the orthonormal projection of $ \ell^2 \left(\mathbb{Z}_+ \right) \otimes \mathbb{C}^2$ onto
the linear span of $e_0 \otimes e_{-}, e_0 \otimes e_{+},  \ldots, e_n \otimes e_{-}, e_n \otimes e_{+}$. 
Since
\[
\mathbb{E} \left\Vert pr_n \left( \xi \right) \right\Vert ^2 + \mathbb{E} \left\Vert N^2 \, pr_n \left( \xi \right) \right\Vert ^2
\leq
\mathbb{E} \left\Vert  \xi \right\Vert ^2 + n^2 \mathbb{E} \sum_{k=0}^n \sum_{\eta = \pm} 
\left\vert \langle e_{k} \otimes e_{\eta} , \xi \rangle \right\vert^2
\leq 1 + n^2 ,
\]
$ pr_n \left( \xi \right) \in L_{N^2}^{2}\left( \mathbb{P},  \ell^2 \left(\mathbb{Z}_+ \right) \otimes \mathbb{C}^2 \right) $.
The increasing sequence 
$ \mathbb{E} \left\Vert pr_n \left( \xi \right) \right\Vert ^2 $ converges to $\mathbb{E} \left\Vert \xi \right\Vert ^2 = 1$
as $n  \rightarrow + \infty$,
and so there exists $n_0 \in \mathbb{N}$ such that
$ \mathbb{E} \left\Vert pr_n \left( \xi \right) \right\Vert ^2  > 0$  for all $n \geq n_0$.
For any $n \geq n_0$ we set 
$\xi_n := pr_n \left( \xi \right)  / \sqrt{\mathbb{E} \left\Vert pr_n \left( \xi \right) \right\Vert ^2}$.
Then $ \xi_n \in L_{N^2}^{2}\left( \mathbb{P},  \ell^2 \left(\mathbb{Z}_+ \right) \otimes \mathbb{C}^2 \right) $.

Since $N$ commutes with $pr_n $, 
\begin{align*}
& 
\left\Vert 
 \sqrt{\mathbb{E} \left\Vert pr_n \left( \xi \right) \right\Vert ^2} N^p \, \xi 
 -
 N^p  pr_n \left( \xi \right) 
 \right\Vert ^2
 \\
 & =
 \left( \sqrt{\mathbb{E} \left\Vert pr_n \left( \xi \right) \right\Vert ^2} -1 \right)^2
 \left\Vert  pr_n \left( N^p \, \xi \right)  \right\Vert ^2
 +
\left( \mathbb{E} \left\Vert pr_n \left( \xi \right) \right\Vert ^2 \right)
 \left\Vert N^p \xi - pr_n \left( N^p \xi \right) \right\Vert^2 
\end{align*}
with $p = 0, 1$,
and so
\begin{align*}
 \mathbb{E}  \left( \left\Vert \xi - \xi_n  \right\Vert_{N}^2 \right)
& 
\leq 
 \frac{\left( \sqrt{\mathbb{E} \left\Vert pr_n \left( \xi \right) \right\Vert ^2} - 1 \right)^2}
 {\mathbb{E} \left\Vert pr_n \left( \xi \right) \right\Vert ^2}
 \left( 
 \mathbb{E}  \left( \left\Vert \xi \right\Vert^2 \right) +  \mathbb{E}  \left( \left\Vert N \xi \right\Vert^2 \right)
 \right)
 \\
 & \quad
 +
 \mathbb{E}  \left( \left\Vert \xi - pr_n \left( \xi \right) \right\Vert^2 \right) 
 +  \mathbb{E}  \left( \left\Vert N \xi - pr_n \left( N \xi \right) \right\Vert^2 \right) 
\longrightarrow_{n \rightarrow + \infty} 0 .
\end{align*}

Let  $X_{t} \left( \xi \right)$ be the strong $N$-solution of (\ref{eq:SSE}).
Since 
$
\rho_0 = \mathbb{E}\left\vert \xi \rangle \langle \xi \right\vert 
$,
$
 \rho_t = \mathbb{E} \left| X_{t} \left( \xi \right)\right\rangle \left\langle X_{t} \left( \xi \right)\right| 
$
(see, e.g., \cite{FagMora2019}).
Similarly, 
$
 \rho^n_t := \mathbb{E} \left| X_{t} \left( \xi_n \right)\right\rangle \left\langle X_{t} \left( \xi_n \right)\right| 
$
is the $N^2$-weak solution of (\ref{eq:AuxiliarGKSL}) with initial condition 
$
\mathbb{E} \left|  \xi_n \right\rangle \left\langle \xi_n \right| 
$
(see, e.g., \cite{FagMora2019}),
where
$X_{t} \left( \xi_n \right)$ is the strong $N^2$-solution of (\ref{eq:SSE}) with initial datum $\xi_n$.
Lemma \ref{th:VariationConstantsRegular} yields 
\begin{equation}
\label{eq:AuxiliarGKSL3}
\begin{aligned}
  \rho^n_t 
 & =
 \rho^R_{t-s} \left( \rho^n_s  \right) 
 \\
 & \quad
 + \int_s^t  \rho^R_{t-u} \left( 
 \left[  \left( \alpha_R \left( u \right) a^\dagger -  \overline{\alpha_R \left( u \right)} a \right)
+\left( \overline{\beta_R  \left( u \right) } \sigma^{-}  - \beta_R  \left( u \right) \sigma^+ \right)   ,
\rho^n_u \right]  \right) du 
\end{aligned}
\end{equation}
for all $t \geq s \geq 0$ and $n \geq n_0$.

The linearity of (\ref{eq:SSE})  leads to
\[
\mathbb{E} \left( \left\Vert X_{u} \left( \xi \right) - X_{u} \left( \xi_n \right) \right\Vert_{N}^2 \right)
=
\mathbb{E} \left( \left\Vert X_{u} \left( \xi - \xi_n \right) \right\Vert_{N}^2 \right)
\leq 
K \left( u \right) \mathbb{E}  \left( \left\Vert \xi - \xi_n  \right\Vert_{N}^2 \right) 
\]
for all $u \geq 0$
(see, e.g., \cite{FagMora2019}).
Therefore,  
\begin{equation}
 \label{eq:AuxiliarGKSL1}
 \limsup_{n \rightarrow + \infty}
\sup_{u \in \left[ 0 , t \right]}
\mathbb{E} \left( \left\Vert X_{u} \left( \xi \right) - X_{u} \left( \xi_n \right) \right\Vert_{N}^2 \right)
\leq
\lim_{n \rightarrow + \infty} 
K \left( t \right) \mathbb{E}  \left( \left\Vert \xi - \xi_n  \right\Vert_{N}^2 \right)
= 0 .
\end{equation}
Consider the linear operators $A,B$ in $ \ell^2 \left(\mathbb{Z}_+ \right) \otimes \mathbb{C}^2 $
that are relatively bounded with respect to $N$. Then
\begin{align*}
 & \Tr \left( \left\vert  
\mathbb{E} \left\vert  A \,   X_{u} \left( \xi \right) \rangle  \langle B \,  X_{u} \left( \xi \right) \right\vert 
-
\mathbb{E} \left\vert  A \,   X_{u} \left( \xi_n \right) \rangle  \langle B \,  X_{u} \left( \xi_n \right) \right\vert 
 \right\vert  \right)
 \\ 
& \leq
 \Tr \left( \left\vert  
\mathbb{E} \left\vert  A \,   X_{u}\left( \xi \right) -  A \,   X_{u}\left( \xi_n \right) \rangle  \langle B \,  X_{u}\left( \xi \right) \right\vert 
\right\vert  \right)
\\
& \quad
+
 \Tr \left( \left\vert  
\mathbb{E} \left\vert  A \,   X_{u}\left( \xi_n \right)  \rangle  \langle B \,  X_{u} \left( \xi \right) -  B \,   X_{u}\left( \xi_n \right) \right\vert 
 \right\vert  \right)
 \\
 & \leq
 \mathbb{E} \left( 
 \left\Vert  A \,   X_{u}\left( \xi \right) -  A \,   X_{u}\left( \xi_n \right)  \right\Vert  \left\Vert B \,  X_{u}\left( \xi \right) \right\Vert 
 \right)
 +
\mathbb{E} \left(
 \left\Vert  A \,   X_{u}\left( \xi_n \right)  \right\Vert   \left\Vert B \,  X_{u} \left( \xi \right) -  B \,   X_{u}\left( \xi_n \right) \right\Vert 
 \right)
 \\
 & \leq
 \sqrt{\mathbb{E} \left( \left\Vert  A  \,   X_{u}\left( \xi \right) -  A \,   X_{u}\left( \xi_n \right)  \right\Vert^2  \right)}
 \sqrt{\mathbb{E} \left( \left\Vert B \,  X_{u}\left( \xi \right) \right\Vert^2 \right) }
 \\
  & \quad \quad +
 \sqrt{\mathbb{E} \left( \left\Vert A \,  X_{u}\left( \xi_n \right) \right\Vert^2 \right) }
 \sqrt{\mathbb{E} \left( \left\Vert  B  \,   X_{u}\left( \xi \right) -  B \,   X_{u}\left( \xi_n \right)  \right\Vert^2  \right)} .
  \end{align*}
Since
\[
\mathbb{E} \left( \left\Vert A \, X_{u} \left( \xi_n \right) \right\Vert^2 \right)
\leq
K \, \mathbb{E} \left( \left\Vert X_{u} \left( \xi_n \right) \right\Vert_{N}^2 \right)
\leq 
K \left( u \right) \mathbb{E}  \left( \left\Vert \xi_n  \right\Vert_{N}^2 \right) 
\leq 
K \left( u \right) \mathbb{E}  \left(\left\Vert \xi  \right\Vert_{N}^2 \right)  
\]
(see, e.g., \cite{FagMora2019}),
applying (\ref{eq:AuxiliarGKSL1}) gives  
\begin{equation}
\label{eq:AuxiliarGKSL2}
\sup_{u \in \left[ 0 , t \right]}
 \Tr \left( \left\vert  
\mathbb{E} \left\vert  A \,   X_{u}\left( \xi \right) \rangle  \langle B \,  X_{u}\left( \xi \right) \right\vert 
-
\mathbb{E} \left\vert  A \,   X_{u}\left( \xi_n \right) \rangle  \langle B \,  X_{u} \left( \xi_n \right) \right\vert 
 \right\vert  \right)
 \longrightarrow_{n \rightarrow + \infty} 0 .
\end{equation}

We shall now proceed to tend $n$ to $+ \infty$ in (\ref{eq:AuxiliarGKSL3}).
Applying (\ref{eq:AuxiliarGKSL2}) we obtain
\[
\sup_{u \in \left[ 0 , t \right]} \hspace{-1pt} 
\Tr \left( \left\vert  \rho_u - \rho^n_u  \right\vert  \right)
\hspace{-1pt} = \hspace{-2pt} 
\sup_{u \in \left[ 0 , t \right]}
\hspace{-1pt}  \Tr \left( \left\vert  
\mathbb{E} \left\vert   X_{u} \left( \xi \right) \rangle  \langle X_{u} \left( \xi \right) \right\vert 
\hspace{-1pt}  - \hspace{-1pt} 
\mathbb{E} \left\vert   X_{u} \left( \xi_n \right) \rangle  \langle  X_{u} \left( \xi_n \right) \right\vert 
 \right\vert  \right)
 \longrightarrow_{n \rightarrow + \infty} 0 .
\]
Since $\rho^R_{t - s }$ is a contraction acting on 
$\mathfrak{L}_{1}\left(  \ell^2 \left(\mathbb{Z}_+ \right) \otimes \mathbb{C}^2  \right)$,
\[
\Tr \left( \left\vert \rho^R_{t-s} \left( \rho_s  \right) - \rho^R_{t-s} \left( \rho^n_s  \right)  \right\vert  \right)
 \leq 
\Tr \left( \left\vert  \rho_s - \rho^n_s  \right\vert  \right) \longrightarrow_{n \rightarrow + \infty} 0 .
\]
Moreover,
\begin{align*}
\Tr \left( \left\vert   \rho^R_{t-u} \left( 
 \left[  \overline{\beta_R  \left( u \right) } \sigma^{-}  - \beta_R  \left( u \right) \sigma^+   ,
\rho_u  - \rho^n_u \right] \right) 
\right\vert  \right)
 & \leq
\Tr \left( \left\vert    \left[  \overline{\beta_R  \left( u \right) } \sigma^{-}  - \beta_R  \left( u \right) \sigma^+   ,\rho_u  - \rho^n_u \right] 
\right\vert  \right)
\\
& \leq
2  \left\vert \beta_R  \left( u \right)  \right\vert \left(   \left\Vert  \sigma^{-} \right\Vert +   \left\Vert  \sigma^{+} \right\Vert \right) 
\Tr \left( \left\vert  \rho_u  - \rho^n_u \right\vert  \right) .
\end{align*}
This implies
\[
\sup_{u \in \left[ 0 , t \right]}
\Tr \left( \left\vert   \rho^R_{t-u} \left( 
 \left[  \overline{\beta_R  \left( u \right) } \sigma^{-}  - \beta_R  \left( u \right) \sigma^+   ,
\rho_u  - \rho^n_u \right] \right) 
\right\vert  \right)
\longrightarrow_{n \rightarrow + \infty} 0 .
\]

Using that $\rho^R_{t} \left( \cdot \right)$ is a  contraction semigroup gives
\begin{equation}
\label{eq:AuxiliarGKSL4}
\begin{aligned}
&
\Tr \left( \left\vert    \rho^R_{t-u} \left( 
 \left[   \alpha_R \left( u \right) a^\dagger -  \overline{\alpha_R \left( u \right)} a  ,  \rho_u  - \rho^n_u \right] 
\right\vert  \right)
 \right) 
  \leq
\Tr \left( \left\vert   
 \left[  \alpha_R \left( u \right) a^\dagger -  \overline{\alpha_R \left( u \right)} a  ,  \rho_u  - \rho^n_u \right] 
\right\vert  \right)  .
\end{aligned}
\end{equation}
Since
\begin{align*}
 &
 \left[   \alpha_R \left( u \right) a^\dagger -  \overline{\alpha_R \left( u \right)} a  ,  \rho_u  - \rho^n_u \right] 
 \\
 &
 =
 \mathbb{E} \left\vert   \left( \alpha_R \left( u \right) a^\dagger -  \overline{\alpha_R \left( u \right)} a \right)  X_{u} \left( \xi \right) \rangle  \langle X_{u} \left( \xi \right) \right\vert 
 -
 \mathbb{E} \left\vert   \left( \alpha_R \left( u \right) a^\dagger -  \overline{\alpha_R \left( u \right)} a \right)  X_{u} \left( \xi_n \right) \rangle  \langle X_{u} \left( \xi_n \right) \right\vert 
\\
 & \quad
+
 \mathbb{E} \left\vert  X_{u} \left( \xi \right) \rangle  \langle \left( \alpha_R \left( u \right) a^\dagger 
 -  \overline{\alpha_R \left( u \right)} a \right)  X_{u} \left( \xi \right) \right\vert 
 -
  \mathbb{E} \left\vert  X_{u} \left( \xi_n \right) \rangle  \langle \left( \alpha_R \left( u \right) a^\dagger 
  -  \overline{\alpha_R \left( u \right)} a \right)  X_{u} \left( \xi_n \right) \right\vert 
\end{align*}
(see, e.g.,  Theorem 3.2 of \cite{MoraAP}), 
combining (\ref{eq:AuxiliarGKSL2}) with  (\ref{eq:AuxiliarGKSL4}) we deduce that
\[
\sup_{u \in \left[ 0 , t \right]}
\Tr \left( \left\vert    \rho^R_{t-u} \left( 
 \left[  \alpha_R \left( u \right) a^\dagger -  \overline{\alpha_R \left( u \right)} a ,  \rho_u  - \rho^n_u \right] 
\right\vert  \right)
 \right) 
\longrightarrow_{n \rightarrow + \infty} 0 .
\]
Now, taking the limit as $n \rightarrow + \infty$ in (\ref{eq:AuxiliarGKSL3})
we obtain (\ref{eq:AuxiliarGKSL5}).
\end{proof}

\subsection{Proof of Corollary \ref{cor:Ineq-Conv}}
\label{sec:Proof:Ineq-Conv}

% Theorem \ref{th:VariationConstants} 

\begin{lemma}
\label{lem:Ineq-Trace}
Suppose that the operator $C : \mathcal{D}\left(C\right) \subset \mathfrak{h} \rightarrow \mathfrak{h} $ 
is positive and self-adjoint. 
Asume that  $\varrho$ is a $C$-regular density operator in $\mathfrak{h}$.
Consider the linear operator $A:  \mathcal{D}\left(A\right) \subset \mathfrak{h} \rightarrow \mathfrak{h}$.
Then:
\begin{itemize}
 
 \item
 $
\Tr \left( \left\vert A \varrho \right\vert \right)  
\leq
\sqrt{\Tr \left( \varrho A^{\star}A \right)  }
$
whenever 
$A,A^{\star}A  \in \mathfrak{L}\left( \left( \mathcal{D}\left(C\right) , \left\Vert \cdot \right\Vert_C \right) , \mathfrak{h}\right)$.

 \item
 $
\Tr \left( \left\vert \varrho A  \right\vert \right)  
\leq
\sqrt{\Tr \left(\varrho A A^{\star}  \right)  }
$
provided that
$A^{\star},AA^{\star}  \in \mathfrak{L}\left( \left( \mathcal{D}\left(C\right) , \left\Vert \cdot \right\Vert_C \right) , \mathfrak{h}\right)$. 
\end{itemize}
\end{lemma}

\begin{proof}
Since  $\varrho \in \mathfrak{L}_{1,C}^{+} \left( \mathfrak{h}\right) $,
there exists  $\xi \in L_{C}^{2}\left( \mathbb{P},\mathfrak{h}\right) $ such that
$
\varrho = \mathbb{E}\left\vert \xi \rangle \langle \xi \right\vert 
$
and 
$ \mathbb{E} \left( \left\Vert \xi  \right\Vert^2 \right) = 1 $
(see, e.g., Theorem 3.1 of \cite{MoraAP}).
If  
$A,A^{\star}A  \in \mathfrak{L}\left( \left( \mathcal{D}\left(C\right) , \left\Vert \cdot \right\Vert_C \right) , \mathfrak{h}\right)$,
then 
using Theorem 3.2 of of \cite{MoraAP} we obtain
\[
\Tr \left( \left\vert A \varrho \right\vert \right)
=
\sup_{\left\Vert B \right\Vert = 1} \left\vert \Tr \left( B A \varrho \right)  \right\vert
=
\sup_{\left\Vert B \right\Vert = 1} \left\vert \mathbb{E} \langle \xi , B A \, \xi \rangle  \right\vert
\leq
 \mathbb{E} \left( \left\Vert  \xi \right\Vert  \left\Vert A \, \xi \right\Vert \right) ,
\]
and so
\[
\mathbb{E} \left( \left\Vert  \xi \right\Vert  \left\Vert A \, \xi \right\Vert \right)
\leq
\sqrt{ \mathbb{E} \left( \left\Vert  \xi \right\Vert^2 \right) } \sqrt{ \mathbb{E} \left( \left\Vert  A \, \xi \right\Vert^2 \right) }
=
\sqrt{ \mathbb{E}  \langle A^{\star} A \, \xi ,  \xi \rangle }
=
\sqrt{\Tr \left( \varrho A^{\star}A  \right)  } ,
\]
because $\mathbb{E} \left( \left\Vert  \xi \right\Vert ^2 \right) = 1$.
Similarly, in case 
$A^{\star},AA^{\star}  \in \mathfrak{L}\left( \left( \mathcal{D}\left(C\right) , \left\Vert \cdot \right\Vert_C \right) , \mathfrak{h}\right)$
we have
$
\Tr \left( \left\vert  \varrho A \right\vert \right)
=
\sup_{\left\Vert B \right\Vert = 1} \left\vert \mathbb{E} \langle A^{\star} \, \xi , B  \xi \rangle  \right\vert
=
\sqrt{ \mathbb{E} \left( \left\Vert  A^{\star}  \, \xi \right\Vert^2 \right) }
=
\sqrt{\Tr \left( \varrho A  A^{\star}  \right)  } 
$.
\end{proof}

\begin{proof}[Proof of Corollary \ref{cor:Ineq-Conv}]
Theorem \ref{th:VariationConstants} leads to
\begin{align*}
& 
\Tr \left( \left\vert  \rho_t  - \varrho^f_{\infty} \otimes \varrho^a_{\infty} \right\vert  \right)
\leq 
\Tr \left( \left\vert   \rho^R_{t-s} \left( \rho_s  \right)   -  \varrho^f_{\infty} \otimes \varrho^a_{\infty} \right\vert  \right)
\\
& \hspace{1cm}
+ \Tr \left( \left\vert
\int_s^t  \rho^R_{t-u} \left( 
 \left[  \left( \alpha_R \left( u \right) a^\dagger -  \overline{\alpha_R \left( u \right)} a \right)
+\left( \overline{\beta_R  \left( u \right) } \sigma^{-}  - \beta_R  \left( u \right) \sigma^+ \right)   ,
\rho_u \right]  \right) du 
\right\vert  \right),
\end{align*}
where $t  \geq s \geq 0$.
Using that $\rho^R_{t -u}$ is a contraction acting on 
$\mathfrak{L}_{1}\left(  \ell^2 \left(\mathbb{Z}_+ \right) \otimes \mathbb{C}^2  \right)$
we obtain
\begin{align*}
 \Tr \left( \left\vert  \rho_t  - \varrho^f_{\infty} \otimes \varrho^a_{\infty} \right\vert  \right)
& \leq 
\Tr \left( \left\vert   \rho^R_{t-s} \left( \rho_s  \right)   -  \varrho^f_{\infty} \otimes \varrho^a_{\infty} \right\vert  \right)
\\
&  \hspace{-1cm}
+ 
\int_s^t  \Tr \left( \left\vert
 \left[  \left( \alpha_R \left( u \right) a^\dagger -  \overline{\alpha_R \left( u \right)} a \right)
+\left( \overline{\beta_R  \left( u \right) } \sigma^{-}  - \beta_R  \left( u \right) \sigma^+ \right)   ,
\rho_u \right] \right\vert   \right) du ,
\end{align*}
and so
\begin{align*}
\Tr \left( \left\vert  \rho_t  - \varrho^f_{\infty} \otimes \varrho^a_{\infty} \right\vert  \right)
& \leq 
\Tr \left( \left\vert   \rho^R_{t-s} \left( \rho_s  \right)   -  \varrho^f_{\infty} \otimes \varrho^a_{\infty} \right\vert  \right)
\\
&  \hspace{-0.3cm}
+ 
\int_s^t   \left\vert \alpha_R \left( u \right)  \right\vert \left( 
\Tr \left( \left\vert a^\dagger  \rho_u  \right\vert  \right) + \Tr \left( \left\vert  \rho_u  \, a^\dagger \right\vert  \right) 
+
\Tr \left( \left\vert a \, \rho_u \right\vert  \right) + \Tr \left( \left\vert \rho_u \, a  \right\vert  \right)
 \right) du 
\\
& \hspace{-1cm}
+ 
\int_s^t \left\vert \beta_R \left( u \right)  \right\vert \left( 
\Tr \left( \left\vert  \sigma^{-}  \rho_u \right\vert  \right)  + \Tr \left( \left\vert \rho_u \, \sigma^{-}  \right\vert  \right) 
+ \Tr \left( \left\vert \sigma^+  \rho_u  \right\vert  \right) + \Tr \left( \left\vert \rho_u \sigma^+  \right\vert  \right)
\right) du .
\end{align*}

As $\sigma^{\pm}$ are bounded operators,
\begin{align*}
  \Tr \left( \left\vert  \sigma^{-}  \rho_u \right\vert  \right)  + \Tr \left( \left\vert \rho_u \, \sigma^{-}  \right\vert  \right) 
+ \Tr \left( \left\vert \sigma^+  \rho_u  \right\vert  \right) + \Tr \left( \left\vert \rho_u \, \sigma^+  \right\vert  \right)
& \leq
\left( 2 \left\Vert  \sigma^{-} \right\Vert +  2 \left\Vert  \sigma^{+} \right\Vert \right) \Tr \left( \rho_u \right)
\\
& 
=
2 \left\Vert  \sigma^{-} \right\Vert +  2 \left\Vert  \sigma^{+} \right\Vert .
\end{align*}
By $a^\dagger a = N$, $a \, a^\dagger = N + I$ and $\Tr \left( \rho_u \right) =1$,
applying Lemma \ref{lem:Ineq-Trace} yields
\begin{align*}
& \Tr \left( \left\vert a^\dagger  \rho_u  \right\vert  \right) + \Tr \left( \left\vert  \rho_u  \, a^\dagger \right\vert  \right) 
+
\Tr \left( \left\vert a \, \rho_u \right\vert  \right) + \Tr \left( \left\vert \rho_u \, a  \right\vert  \right)
\\
&
\leq
 2 \sqrt{ \Tr \left( \rho_u \left( N + I \right) \right) } + 2 \sqrt{ \Tr \left( \rho_u \, N  \right) }
\leq
4 \sqrt{ \Tr \left( \rho_u \, N  \right) +1 }
.
\end{align*}
We thus get  (\ref{eq:30n}). 
\end{proof}

\subsubsection{Proof of Theorem \ref{th:ConvAuxiliarQME}}
\label{sec:Proofth:th:ConvAuxiliarQME}

We rewrite (\ref{eq:3.2g}) as
\begin{equation*}
 \frac{d}{dt}  \rho^R_t \left( \varrho \right)
=
\mathcal{L}_{\star}^f \otimes I \left( \rho^R_t \left( \varrho \right) \right) 
+ 
I \otimes \mathcal{L}_{\star}^a \left( \rho^R_t \left( \varrho \right) \right)  ,
\end{equation*}
where $\mathcal{L}_{\star}^f $ is the unbounded linear operator in $ \mathfrak{L}_1 \left(  \ell^2 \left(\mathbb{Z}_+ \right)\right) $ given by
\begin{equation}
 \label{eq:R8.4f}
 \mathcal{L}_{\star}^f \left( \tilde{\varrho} \right)
= 
\left[  
- \left( \kappa + \mathrm{i} \omega \right) a^\dagger a +  \left( \alpha_0 a^\dagger - \overline{\alpha_0} \, a \right)
, \tilde{ \varrho } \right] 
+
2 \kappa \, a\, \tilde{ \varrho } \, a^\dagger  
\end{equation}
and for any $  \hat{\varrho}   \in \mathbb{C}^{2 \times 2}$ we set
\begin{align}
\nonumber
\mathcal{L}_{\star}^a \left(  \hat{\varrho} \right)
& = 
- \mathrm{i} \omega\left[  \frac{1}{2}  \sigma^{3} ,  \hat{\varrho} \, \right]  
+ \left[ \, \overline{\beta_0} \sigma^{-} - \beta_0 \sigma^{+} ,  \hat{\varrho} \right] 
+  \frac{\gamma \left(1-d \right)}{2} \left( 2 \,  \sigma^{-}  \hat{\varrho} \,\sigma^+  -  \sigma^+\sigma^{-}  \hat{\varrho}
 -  \hat{\varrho} \,\sigma^+\sigma^{-}\right)
\\
\label{eq:R8.4a}
& \quad 
+ \frac{ \gamma \left(1+d \right) }{ 2 } \left(  2 \, \sigma^+  \hat{\varrho} \,\sigma^{-} - \sigma^{-}\sigma^+  \hat{\varrho}
-  \hat{\varrho} \,\sigma^{-}\sigma^+\right) .
\end{align}
Here, 
$d,\omega \in\mathbb{R}$, $\alpha_0 , \beta_0 \in\mathbb{C}$ and $\gamma , \kappa > 0$.
Using matrix analysis tools we now study the long time behavior of the semigroup of bounded operators on $\mathbb{C}^{2 \times 2}$
generated by $\mathcal{L}_{\star}^a$.

\begin{lemma}
\label{lem:ConvAtom}
Consider the linear ordinary differential equation
\begin{equation}
\label{eq:9.8}
 \left\{ 
\frac{d}{dt}  \rho^a_t 
  = 
\mathcal{L}_{\star}^a \left( \rho^a_t \right)
\hspace{2cm}
\forall t \geq 0 ,
 \qquad
  \rho^a_0   =   \varrho^a
\right. ,
\end{equation}
where $ \rho^a_t  \in \mathbb{C}^{2 \times 2}$
and $\mathcal{L}_{\star}^a$ is as in (\ref{eq:R8.4a}) with $d,\omega \in\mathbb{R}$,  $\gamma > 0$ and $\beta_0 \in\mathbb{C}$.
Then
\[
\Tr \left( \left\vert  
 \rho^a_t  
-
\Tr \left(\varrho^a \right) \varrho^a_{\infty} \right\vert \right) 
\leq
4 \exp \left( - \gamma  \, t \right) \left(
 \left\Vert  \varrho^a  \right\Vert_F
 +
 \left\vert d \, \Tr \left(\varrho^a \right) \right\vert 
 \right)  
 \qquad
 \forall t \geq 0 ,
\]
where $\left\Vert  \varrho^a  \right\Vert_F$ stands for the Frobenius norm of $\varrho^a $,
and 
$\varrho^a_{\infty}$ is as in Corollary \ref{cor:Ineq-Conv}.
\end{lemma}

\begin{proof}
Decomposing $\rho^a_t $ in the canonical basis of $ \mathbb{C}^{2 \times 2}$ we obtain
\[
 \rho^a_t 
 =
 \alpha_{+ + }\left( t \right) \ketbra{e_+}{e_+} 
 +  \alpha_{+ - }\left( t \right) \ketbra{e_+}{e_-}
 +  \alpha_{- + }\left( t \right) \ketbra{e_-}{e_+}
 +  \alpha_{- - }\left( t \right) \ketbra{e_-}{e_-} ,
\]
where  $\alpha_{\pm  \pm } \left( t \right)$ and $ \alpha_{\pm  \mp } \left( t \right)$ belong to  $\mathbb{C}$.
Then
\begin{align*}
\frac{d}{dt}  \rho^a_t 
=
 \mathcal{L}_{\star}^a \left( \rho^a_t \right)
& =
 \alpha_{+ + }\left( t \right) \mathcal{L}_{\star}^a \left(  \ketbra{e_+}{e_+}   \right) 
+  \alpha_{+ - } \left( t \right)  \mathcal{L}_{\star}^a \left( \ketbra{e_+}{e_-} \right)
\\
& \quad
+  \alpha_{- + } \left( t \right) \mathcal{L}_{\star}^a \left(  \ketbra{e_-}{e_+}  \right)
+  \alpha_{- - } \left( t \right)  \mathcal{L}_{\star}^a \left( \ketbra{e_-}{e_-}   \right) .
\end{align*}
Computing explicitly $\mathcal{L}_{\star}^a \left(  \ketbra{e_\pm}{e_\pm}   \right) $ and 
$\mathcal{L}_{\star}^a \left(  \ketbra{e_\pm}{e_\mp}   \right) $ yields
\[
 \left\{ 
 \begin{aligned}
 \frac{d}{dt} \alpha_{+ + }\left( t \right)  
 & =
 - \overline{\beta_0} \, \alpha_{+ - }\left( t \right)  -  \beta_0 \, \alpha_{- + }\left( t \right)
 - \gamma \left( 1 -d \right) \alpha_{+ + }\left( t \right) + \gamma \left( 1 + d \right) \alpha_{- -}\left( t \right)
 \\
 \frac{d}{dt} \alpha_{- - }\left( t \right)  
 & =
  \overline{\beta_0} \, \alpha_{+ - }\left( t \right)  +  \beta_0 \, \alpha_{- + }\left( t \right)
 + \gamma \left( 1 -d \right) \alpha_{+ + }\left( t \right) - \gamma \left( 1 + d \right) \alpha_{- -}\left( t \right)
 \\
 \frac{d}{dt} \alpha_{+ - }\left( t \right)  
 & =
 - \left( \gamma + \mathrm{i} \omega \right)  \alpha_{+ - }\left( t \right) +  \beta_0 \, \alpha_{+ + }\left( t \right)
 - \beta_0 \, \alpha_{- - }\left( t \right)
  \\
 \frac{d}{dt} \alpha_{- + }\left( t \right)  
 & =
  \left(- \gamma + \mathrm{i} \omega \right)  \alpha_{- + }\left( t \right) +  \overline{\beta_0} \, \alpha_{+ + }\left( t \right)
 - \overline{\beta_0} \, \alpha_{- - }\left( t \right)
 \end{aligned}
\right. .
\]
Adding the first two equations we get
\begin{equation}
 \label{eq:9.2}
 \alpha_{+ + }\left( t \right) + \alpha_{- - }\left( t \right)
=
\alpha_{+ + }\left( 0 \right) + \alpha_{- - }\left( 0 \right)
=
\Tr \left(\varrho^a \right) ,
\end{equation}
and so subtracting the first two equations we deduce that
\[
  \frac{d}{dt} \left( \alpha_{+ + }\left( t \right)  -  \alpha_{- - }\left( t \right)  \right)
  =
  - 2  \overline{\beta_0} \, \alpha_{+ - }\left( t \right)  - 2 \beta_0 \, \alpha_{- + }\left( t \right)
  - 2 \gamma \left( \alpha_{+ + }\left( t \right)  -  \alpha_{- - }\left( t \right)  \right)
  + 2 \gamma d \, \Tr \left(\varrho^a \right) .
\]
Therefore
\begin{equation}
 \label{eq:9.1}
  \frac{d}{dt} 
  \begin{pmatrix}  \alpha_{+ + }\left( t \right)  -  \alpha_{- - }\left( t \right)
  \\ 
  \alpha_{+ - }\left( t \right) 
  \\ \alpha_{- + }\left( t \right) 
  \end{pmatrix}
  =
  A 
  \begin{pmatrix}
  \alpha_{+ + }\left( t \right)  -  \alpha_{- - }\left( t \right)
  \\ 
  \alpha_{+ - }\left( t \right) 
  \\ \alpha_{- + }\left( t \right) 
  \end{pmatrix}
  +
  \begin{pmatrix}
   2 \gamma d \, \Tr \left(\varrho^a \right)  \\ 0 \\ 0 
  \end{pmatrix},
\end{equation}
where
$
A = 
\begin{pmatrix}
  - 2 \gamma &  - 2  \overline{\beta_0} & - 2 \beta_0 
 \\  
 \beta_0 & - \gamma -  \mathrm{i} \omega & 0
 \\
   \overline{\beta_0} & 0 & - \gamma + \mathrm{i} \omega
\end{pmatrix} .
$

Solving explicitly (\ref{eq:9.1}), together the calculation of 
$ A^{-1} 
\begin{pmatrix}  
2 \gamma d \, \Tr \left(\varrho^a \right)  & 0 & 0 
\end{pmatrix}^{\top} $,
gives 
%\begin{align*}
%\begin{pmatrix}
%  \alpha_{+ + }\left( t \right)  -  \alpha_{- - }\left( t \right)
%  \\ 
%  \alpha_{+ - }\left( t \right) 
%  \\ \alpha_{- + }\left( t \right) 
% \end{pmatrix}
% & = 
% \exp \left( A t \right)
% \begin{pmatrix}
%  \alpha_{+ + }\left( 0 \right)  -  \alpha_{- - }\left( 0 \right)
%  \\ \alpha_{+ - }\left( 0 \right) 
%  \\ \alpha_{- + }\left(0 \right) 
%  \end{pmatrix}
% - A^{-1} 
% \begin{pmatrix}
%  2 \gamma d \, \Tr \left(\varrho^a \right)  \\ 0 \\ 0 
% \end{pmatrix}
% \\
% & \quad
% +  A^{-1}  
% \exp \left( A t \right) 
% \begin{pmatrix}
%  2 \gamma d \, \Tr \left(\varrho^a \right)  \\ 0 \\ 0 
%\end{pmatrix} .
%\end{align*}
%Calculating 
%$A^{-1} 
%\begin{pmatrix}  
%2 \gamma d \, \Tr \left(\varrho^a \right)  \\ 0 \\ 0 
%\end{pmatrix}
%$ 
%we obtain
\begin{equation}
\label{eq:9.11}
 \begin{aligned}
&
 \begin{pmatrix} 
 \alpha_{+ + }\left( t \right)  -  \alpha_{- - }\left( t \right)
  \\ 
  \alpha_{+ - }\left( t \right) \\ \alpha_{- + }\left( t \right) 
 \end{pmatrix}
  - \frac{ d \, \Tr \left(\varrho^a \right)}{ \gamma^2 + \omega^2 + 2 \left\vert \beta_0 \right\vert ^2}
  \begin{pmatrix}
  \gamma^2 + \omega^2 
  \\ 
   \beta_0 \left( \gamma - \mathrm{i} \omega \right) 
  \\
    \overline{\beta_0} \left( \gamma + \mathrm{i} \omega \right) 
  \end{pmatrix}
  \\
 & = 
 \exp \left( A t \right) \left( 
 \begin{pmatrix}
  \alpha_{+ + }\left( 0 \right)  -  \alpha_{- - }\left( 0 \right)
  \\ \alpha_{+ - }\left( 0 \right) 
  \\ \alpha_{- + }\left(0 \right) 
  \end{pmatrix}
 - 
 \frac{d \, \Tr \left(\varrho^a \right)}{ \gamma^2 + \omega^2 + 2 \left\vert \beta_0 \right\vert ^2}
 \begin{pmatrix} 
 \gamma^2 + \omega^2 
  \\
  \beta_0 \left( \gamma - \mathrm{i} \omega \right) 
  \\ 
   \overline{\beta_0} \left( \gamma + \mathrm{i} \omega \right) 
 \end{pmatrix}
\right) .
 \end{aligned}
\end{equation}

Consider $v \in \mathbb{C}^{3}$ and 
$
M = 
\begin{pmatrix}
  1 &  0 & 0
 \\  
0 & 2 & 0
 \\
0 & 0 & 2
\end{pmatrix}
$.
Then
\begin{align*}
 \frac{d}{dt} \langle \exp \left( A t \right) v , M \exp \left( A t \right) v \rangle
& =
\langle \exp \left( A t \right) v , \left( A^{\star} M + M \, A \right) \exp \left( A t \right) v \rangle
\\
& =
- 4 \gamma \left\Vert   \exp \left( A t \right) v \right\Vert^2 
\leq
- 2 \gamma  \langle \exp \left( A t \right) v , M \exp \left( A t \right) v \rangle ,
\end{align*}
and hence for all $t \geq 0$,
\[
\left\Vert   \exp \left( A t \right) v \right\Vert^2 
\hspace{-0.8pt} \leq \hspace{-0.8pt}
 \langle \exp \left( A t \right) v , M \exp \left( A t \right) v \rangle
 \hspace{-0.5pt} \leq
 \exp \left( - 2 \gamma  \, t \right) \langle v , M  \, v \rangle
 \leq
 2 \exp \left( - 2 \gamma  \, t \right)  \left\Vert v \right\Vert^2 \hspace{-1pt}  .
\]
From (\ref{eq:9.11}) it follows 
\begin{align*}
 & \left\Vert  
 \begin{pmatrix} 
 \alpha_{+ + }\left( t \right)  -  \alpha_{- - }\left( t \right)
  \\ 
  \alpha_{+ - }\left( t \right) 
  \\ \alpha_{- + }\left( t \right) 
 \end{pmatrix}
  - \frac{d \, \Tr \left(\varrho^a \right)}{ \gamma^2 + \omega^2 + 2 \left\vert \beta_0 \right\vert ^2}
 \begin{pmatrix} 
 \gamma^2 + \omega^2 
  \\ 
   \beta_0 \left( \gamma - \mathrm{i} \omega \right) 
  \\
    \overline{\beta_0} \left( \gamma + \mathrm{i} \omega \right) 
 \end{pmatrix}
 \right\Vert
 \\
 & \leq
 \sqrt{2} \exp \left( - \gamma  \, t \right) \left(
 \left\Vert  
 \begin{pmatrix}  \alpha_{+ + }\left( 0 \right)  -  \alpha_{- - }\left( 0 \right)
  \\ 
  \alpha_{+ - }\left( 0 \right) 
  \\ 
  \alpha_{- + }\left(0 \right) 
 \end{pmatrix}
 \right\Vert
 +
 \frac{\left\vert d \, \Tr \left(\varrho^a \right) \right\vert  \sqrt{ \gamma^2 + \omega^2 }
 }{ \sqrt{  \gamma^2 + \omega^2 + 2 \left\vert \beta_0 \right\vert ^2 }  }
 \right) 
 \\
 & \leq
2 \exp \left( - \gamma  \, t \right) \left(
 \left\Vert  
 \begin{pmatrix}
 \alpha_{+ + }\left( 0 \right)  
 \\  \alpha_{- - }\left( 0 \right)
  \\ \alpha_{+ - }\left( 0 \right) \\ \alpha_{- + }\left(0 \right) 
 \end{pmatrix}
 \right\Vert
 +
 \left\vert d \, \Tr \left(\varrho^a \right) \right\vert 
 \right).
\end{align*}
%
%The characteristic polynomial of $A$ is 
%$$
%\det \left( A - \lambda I \right)
%=
%- \lambda^3 - 4 \gamma \lambda^2 - \left( 5 \gamma^2 +  \omega^2 + 4 \left\vert \beta \right\vert ^2  \right) \lambda 
%- \left( 2 \gamma^3  + 2 \gamma \omega^2 + 4 \left\vert \beta \right\vert ^2 \gamma \right) .
%$$
%Since $4 \gamma > 0$, $2 \gamma^3  + 2 \gamma \omega^2 + 4 \left\vert \beta \right\vert ^2 \gamma > 0$
%and 
%$$
%4 \gamma  \left( 5 \gamma^2 +  \omega^2 + 4 \left\vert \beta \right\vert ^2  \right)
%>
%2 \gamma^3  + 2 \gamma \omega^2 + 4 \left\vert \beta \right\vert ^2 \gamma ,
%$$
% applying the Routh-Hurwitz criterion we obtain that the roots  $\lambda_1, \lambda_2, \lambda_3$ of 
% the polynomial $ \lambda \mapsto \det \left( A - \lambda I \right)$
%have negative real part.
%Therefore
%$$
%\exp \left( c \, t \right)
%\left(
% \begin{pmatrix}  \alpha_{+ + }\left( t \right)  -  \alpha_{- - }\left( t \right)
%  \\ \alpha_{+ - }\left( t \right) \\ \alpha_{- + }\left( t \right) \end{pmatrix}
% + A^{-1} \begin{pmatrix}  2 \gamma d \, \Tr \left(\varrho \right)  \\ 0 \\ 0 \end{pmatrix} 
% \right)
%\longrightarrow_{t \rightarrow + \infty} ,
%$$
%provided that
%$c < \min \left\{ \left\vert \Re \left(  \lambda_1  \right) \right\vert ,  \left\vert \Re \left(  \lambda_2  \right) \right\vert ,  \left\vert \Re \left(  \lambda_3  \right) \right\vert \right\}$.
Using (\ref{eq:9.2}) we deduce that 
\begin{align*}
& \Tr \left( 
\left\vert  
\rho^a_t 
 -
 \frac{ \Tr \left(\varrho^a \right)}{2} 
 \begin{pmatrix}  1  & 0 \\ 0 & 1  \end{pmatrix}
  - \frac{d \, \Tr \left(\varrho^a \right)}{ \gamma^2 + \omega^2 + 2 \left\vert \beta_0 \right\vert ^2}
 \begin{pmatrix} 
\left( \gamma^2 + \omega^2 \right)/2
&
 \beta_0 \left( \gamma - \mathrm{i} \omega \right) 
\\
  \overline{\beta_0} \left( \gamma + \mathrm{i} \omega \right) 
 &
- \left( \gamma^2 + \omega^2 \right)/2
 \end{pmatrix}
 \right\vert
\right)
\\
& \leq \sqrt{2}
\left\Vert  
\begin{pmatrix} 
 \alpha_{+ + }\left( t \right)  \\  \alpha_{- - }\left( t \right)
  \\ \alpha_{+ - }\left( t \right) \\ \alpha_{- + }\left( t \right) 
 \end{pmatrix}
 -
 \frac{ \Tr \left(\varrho^a \right)}{2} 
  \begin{pmatrix}  1 \\ 1 \\  0 \\  0  \end{pmatrix}
  - \frac{d \, \Tr \left(\varrho^a \right)}{ \gamma^2 + \omega^2 + 2 \left\vert \beta_0 \right\vert ^2}
\begin{pmatrix}
\left( \gamma^2 + \omega^2 \right)/2
\\
- \left( \gamma^2 + \omega^2 \right)/2
  \\
   \beta_0 \left( \gamma - \mathrm{i} \omega \right) 
  \\ 
    \overline{\beta_0} \left( \gamma + \mathrm{i} \omega \right) 
  \end{pmatrix}
 \right\Vert
 \\
&  \leq
4 \exp \left( - \gamma  \, t \right) \left(
 \left\Vert  
\begin{pmatrix} 
 \alpha_{+ + }\left( 0 \right)  
 \\  \alpha_{- - }\left( 0 \right)
  \\ \alpha_{+ - }\left( 0 \right) \\ \alpha_{- + }\left(0 \right) 
\end{pmatrix}
 \right\Vert
 +
 \left\vert d \, \Tr \left(\varrho^a \right) \right\vert 
 \right)  .
\end{align*}
\end{proof}

Next,
by means of the Weyl operator we connect $\mathcal{L}_{\star}^f$ with 
a GKSL master equation in 
$\mathfrak{L}^+_{1}\left(  \ell^2 \left(\mathbb{Z}_+ \right)\right) $ 
whose coefficients only involve annihilation, number and identity operators.
Then,
applying techniques used to get the convergence of quantum dynamical semigroups to the ground state 
we  obtain 
the exponential convergence to the equilibrium of the regular solution of (\ref{eq:9.6}) given below.

\begin{lemma}
 \label{lem:ConvFoton}
 Suppose that $\left( \rho^f_t \right)_{t \geq 0}$ is the  $N$-weak solution to 
 \begin{equation}
 \label{eq:9.6}
\frac{d}{dt}  \rho^f_t    = \mathcal{L}_{\star}^f \left( \rho^f_t \right)
\hspace{1cm} \forall t \geq 0 ,
\qquad
\rho^f_0  = \varrho^f  ,
\end{equation}
where $ \varrho^f \in \mathfrak{L}^+_{1}\left(  \ell^2 \left(\mathbb{Z}_+ \right) \right) $ 
is a $N$-regular density operator and 
$\mathcal{L}_{\star}^f$ is described by (\ref{eq:R8.4f}) with $\kappa > 0$, $\alpha_0 \in \mathbb{C}$ and $\omega \in \mathbb{R}$.
Then
\[
\Tr \left( \left\vert
 \rho^f_t 
 -
 \ketbra{ \mathcal{E} \left( \frac{ \alpha_0 }{ \kappa  + \mathrm{i} \omega } \right)} 
{ \mathcal{E} \left( \frac{ \alpha_0 }{ \kappa  + \mathrm{i} \omega } \right)}
  %  \varrho^f_{\infty}
 \right\vert  \right) 
\leq
 2 \hbox{\rm e}^{- \kappa \, t} \left( \sqrt{ \Tr \left( \varrho^f \, N  \right) } 
 +  \left\vert  \alpha_0  \right\vert  / \sqrt{ \kappa^2   +  \omega^2 } \right)
\]
for all $  t \geq 0 $, 
where  
$ \mathcal{E} \left( \cdot \right) $ is defined by \eqref{eq:Coherent-Vector}.
\end{lemma}

\begin{proof}
Consider the unitary Weyl operator $W \left(  u \right)$ defined by
\[
W \left(  u \right) e \left(  z \right)
=
\exp \left(  - \left\vert u \right\vert ^2 /2 - \overline{u} z \right)
e \left(  z + u \right)
\hspace{1cm}
\forall z \in \mathbb{C} ,
\]
where $u \in \mathbb{C}$ and the exponential vector associated with $\zeta \in \mathbb{C}$ is given by 
$
 e \left( \zeta \right) 
=
\sum_{n \geq 0} \zeta^n e_n / \sqrt{ n !} 
$
(see, e.g., \cite{Parthasarathy1992}).
Applying the well-known relations
\[
W \left(  u \right) W \left(  - u \right) = I,
\quad
W \left(  u \right) a \, W \left(  - u \right) = a - u I,
\quad
W \left(  u \right) a^\dagger \, W \left(  - u \right) = a^\dagger - \overline{u} I
\]
we obtain
$
W \left(  u \right) a^\dagger a \, W \left(  - u \right)
=
a^\dagger a - u a^\dagger - \overline{u} a +  \left\vert u \right\vert ^2
$.
Take  
\[
v=  \alpha_0 / \left( \kappa  + \mathrm{i} \omega \right).
\]
For any $\xi \in  L_{N}^{2}\left( \mathbb{P}, \ell^2 \left(\mathbb{Z}_+ \right) \right)$,
$
W \left(  - v \right) \mathbb{E}\left\vert \xi \rangle \langle \xi \right\vert  \, W \left(  v \right)
=
\mathbb{E}\left\vert W \left(  - v \right) \xi \rangle \langle W \left(  - v \right) \xi \right\vert 
$
and
\[
\mathbb{E} \left( \left\Vert N \, W \left(  - v \right) \xi \right\Vert^2 \right)
\leq
\left\Vert W \left(  - v \right) \right\Vert^2  \mathbb{E} \left\Vert \left(
a^\dagger a - v a^\dagger - \bar{v} a +  \left\vert v \right\vert ^2
\right)  \xi \right\Vert^2 
\hspace{-1pt} \leq \hspace{-1pt}
K \left( \left\vert v \right\vert \right) \mathbb{E} \left( \left\Vert \xi \right\Vert_N^2 \right) \hspace{-1.5pt}.
\]
Hence,
the application $ \widetilde{  \varrho } \mapsto W \left(  - v \right)  \widetilde{  \varrho } \, W \left(  v \right)$
preserves the property of being $N$-regular.

Set 
$
\widetilde{L} = \sqrt{2 \kappa} \, a + \sqrt{2 \kappa} \, v \, I
$
and
\[
\widetilde{G} =
- \left( \kappa + \mathrm{i} \omega \right) a^\dagger a 
- \frac{2 \kappa \, \overline{\alpha_0} }{ \kappa  - \mathrm{i} \omega } a
+ \left\vert \alpha_0 \right\vert^2 \left( \frac{1 }{ \kappa  - \mathrm{i} \omega } - \frac{2 \kappa }{ \kappa^2  +\omega^2 }\right) I .
\]
Then, 
for all $x$  in the domain of $ a^\dagger a$ we have
\begin{equation*}
 \sqrt{2 \kappa} \, a \, x
=
W \left(   \frac{ \alpha_0 }{ \kappa  + \mathrm{i} \omega }\right) \widetilde{L} \,
W \left(  - \frac{ \alpha_0 }{ \kappa  + \mathrm{i} \omega }\right) x
\end{equation*}
and 
\begin{equation*}
 - \left( \kappa + \mathrm{i} \omega \right) a^\dagger a \, x+  \left( \alpha_0 a^\dagger - \overline{\alpha_0} a \right) x
=
W \left(   \frac{ \alpha_0 }{ \kappa  + \mathrm{i} \omega }\right) \widetilde{G} \,
W \left(  - \frac{ \alpha_0 }{ \kappa  + \mathrm{i} \omega }\right) x .
\end{equation*}
This gives
\begin{equation}
\label{eq:9.3}
 \begin{aligned}
 \mathcal{L}_{\star}^f \left( \widetilde{  \varrho } \right)
& =
W \left(  v \right) \widetilde{G} \, W \left(  - v \right) \widetilde{  \varrho }
+
\widetilde{  \varrho } \, W \left(  v \right) \widetilde{G}^{\, *} \, W \left(  - v \right) 
 \\
& \quad +
W \left(  v \right) \widetilde{L} \, W \left(  - v \right) \widetilde{  \varrho } \, W \left(  v \right) \widetilde{L}^{\, *} \, W \left(  - v \right) ,
\end{aligned}
\end{equation}
for any $N$-regular density operator $\widetilde{  \varrho }$ in $ \ell^2 \left(\mathbb{Z}_+ \right)$.

Choose 
$
\widetilde{\rho}_t
=
W \left(  - v \right)  \rho^f_t \, W \left(  v \right)
$.
Then, the density operator $\widetilde{\rho}_t$ is  $N$-regular.
Combining (\ref{eq:9.6}) with (\ref{eq:9.3}) we obtain that $\left( \widetilde{\rho}_t \right)_{t \geq 0}$ 
is the  $N$-weak solution to 
 \begin{equation}
 \label{eq:9.4}
 \left\{ 
\frac{d}{dt} \widetilde{\rho}_t    = \widetilde{ \mathcal{L}_{\star} } \left( \widetilde{\rho}_t \right)
\hspace{2cm} \forall t \geq 0 ,
\qquad
\widetilde{\rho}_0   =
W \left(  - v \right)  \varrho^f  \, W \left(   v \right)
\right. ,
\end{equation}
where 
$
 \widetilde{ \mathcal{L}_{\star} } \left( \widetilde{  \varrho } \right)
 =
 \widetilde{G} \, \widetilde{  \varrho } + \widetilde{  \varrho } \, \widetilde{G}^{\, *} 
 + \widetilde{L} \, \widetilde{  \varrho } \, \widetilde{L}^{\, *}
$.
A computation yields 
\[
\left\langle   e_j ,   \widetilde{ \mathcal{L}_{\star} } \left( \widetilde{  \varrho } \right)   e_j \right\rangle
=
- 2 \kappa j \left\langle   e_{j} ,   \widetilde{  \varrho } \, e_{j} \right\rangle
+ 
2 \kappa \left(j+1\right) \left\langle   e_{j+1} ,   \widetilde{  \varrho } \, e_{j+1} \right\rangle 
\quad \quad \forall j \geq 0
\]
whenever $\widetilde{  \varrho }$ is a $N$-regular density operator in  $\ell^2 \left(\mathbb{Z}_+ \right)$.
Applying (\ref{eq:9.4}) we deduce that the functions
$
\varphi_j  \left( t \right) := \left\langle   e_j ,    \widetilde{\rho}_t \, e_j \right\rangle 
$
satisfy 
\begin{equation}
\label{eq:8.21}
 \varphi_j^\prime \left( t \right) = - 2 \kappa j \varphi_j \left( t \right) + 2 \kappa (j+1) \varphi_{j+1}\left( t \right) ,
\end{equation}
which are the Kolmogorov equations for a pure-death process.
In case  $\varphi_j(0)=\delta_{jn}$, for all $j \geq 0$,
the solution of (\ref{eq:8.21})  is  
$
\varphi_j \left( t \right) 
=
\begin{pmatrix} n \\ j \end{pmatrix} \hbox{\rm e}^{- 2\kappa \, j \, t} \left( 1-\hbox{\rm e}^{-2 \kappa \, t} \right)^{n- j} 
$
for  $0 \le j \le n $,
and
$
\varphi_j \left( t \right) = 0
$
if $ j > n $.
%\begin{equation*}
% \varphi_j \left( t \right) =
%\begin{cases}  
%\begin{pmatrix} n \\ j \end{pmatrix}
%\hbox{\rm e}^{- 2\kappa \, j  \, t} \left( 1-\hbox{\rm e}^{-2 \kappa t} \right)^{n- j} 
% & \text{ if }  0 \le j \le n 
% \\
% 0 & \text{ if }   j > n
%\end{cases}   .
%\end{equation*}
Therefore,
$
 \left\langle   e_0 ,    \widetilde{\rho}_t \, e_0 \right\rangle 
=
\varphi_0 \left( t \right)
=
 \sum_{n \geq 0} \varphi_n \left( 0 \right)  \left( 1- \hbox{\rm e}^{-2 \kappa t} \right)^{n} 
$.

According to Theorem 4.2 of \cite{AccardiFagnolaHachicha2006} we have 
$
 \Tr \left( \left\vert   
 \widetilde{\rho}_t -  \ketbra{e_0}{e_0} 
\right\vert  \right)
\leq 
2\left(1 -  \left\langle   e_0 ,   \widetilde{\rho}_t e_0 \right\rangle \right)^{1/2} 
$.
Using that $ \varphi_n \left( 0 \right) \geq 0 $ and $ \sum_{n \geq 0} \varphi_n \left( 0 \right) = 1$
we obtain
\[
0 \leq  1 -  \left\langle   e_0 ,   \widetilde{\rho}_t \, e_0 \right\rangle
=
 \sum_{n \geq 1} \varphi_n \left( 0 \right)  \left( 1 - \left( 1- \hbox{\rm e}^{-2 \kappa t} \right)^{n} \right)
 \leq 
\hbox{\rm e}^{-2 \kappa t}  \sum_{n \geq 1} n \, \varphi_n \left( 0 \right)  ,
\]
because 
$
1 - \left( 1 - x \right)^n \leq n \, x
$
for any $n \in \mathbb{N}$ and $x \in \left[ 0 , 1 \right]$.
Hence 
\[
 \Tr \left( \left\vert   
 \widetilde{\rho}_t -  \ketbra{e_0}{e_0} 
\right\vert  \right)
\leq 2 \hbox{\rm e}^{- \kappa t} \left(  \sum_{n \geq 0}  \left\langle   n \, e_n ,    \widetilde{\rho}_0 \, e_n \right\rangle \right)^{1/2}
=
2 \hbox{\rm e}^{- \kappa t} \left(  \sum_{n \geq 0}  \left\langle   N \, e_n ,    \widetilde{\rho}_0 \, e_n \right\rangle \right)^{1/2} ,
\]
and so
\[
\Tr \left( \left\vert   
 \widetilde{\rho}_t -  \ketbra{e_0}{e_0} 
\right\vert  \right)
\leq
2 \hbox{\rm e}^{- \kappa t} \left( \Tr \left(\widetilde{\rho}_0 N \right) \right)^{1/2}
=
2 \hbox{\rm e}^{- \kappa t} \left( \Tr \left( \varrho^f  \, W \left(   v \right) N \, W \left(  - v \right) \right)  \right)^{1/2}
\]
(see, e.g., Theorem 3.2 of \cite{MoraAP}).
Then
\begin{equation}
\label{eq:8.25}
\Tr \left( \left\vert   
 \widetilde{\rho}_t -  \ketbra{e_0}{e_0} 
\right\vert  \right)
\leq
2 \hbox{\rm e}^{- \kappa t} \left( \Tr \left( \varrho^f    \left(   N - v a^\dagger - \bar{v} a +  \left\vert v \right\vert ^2 \right) \right) \right)^{1/2} . 
\end{equation}

Due to $\varrho^f =  \mathbb{E}\left\vert \xi \rangle \langle \xi \right\vert $
for certain $\xi \in  L_{N}^{2}\left( \mathbb{P}, \ell^2 \left(\mathbb{Z}_+ \right) \right)$,
\[
\left\vert \Tr \left( \varrho^f  \, a^\dagger \right)  \right\vert
=
\left\vert  \mathbb{E}   \langle  a \, \xi , \xi \rangle \right\vert
\leq
\sqrt{ \mathbb{E} \left\vert a \, \xi  \right\vert^2 } \sqrt{  \mathbb{E} \left\vert \xi  \right\vert^2 } 
=
\sqrt{  \mathbb{E} \left\vert a \, \xi  \right\vert^2 }
=
\sqrt{  \mathbb{E}  \langle  N \, \xi , \xi \rangle }
=
\sqrt{  \Tr \left( \varrho^f  \, N \right) }
\]
and
$
\left\vert \Tr \left( \varrho^f  \, a \right)  \right\vert
=
\left\vert  \mathbb{E}   \langle  a^\dagger \, \xi , \xi \rangle \right\vert
\leq
\sqrt{  \mathbb{E} \left\vert a \, \xi  \right\vert^2 }
=
\sqrt{  \Tr \left( \varrho^f  \, N \right) }
$
(see, e.g., Theorem 3.2 of \cite{MoraAP}).
From (\ref{eq:8.25}) we deduce that
$
\Tr \left( \left\vert   
 \widetilde{\rho}_t -  \ketbra{e_0}{e_0} 
\right\vert  \right)
\leq
2 \hbox{\rm e}^{- \kappa t} \left( \sqrt{ \Tr \left( \varrho^f \, N  \right) } +  \left\vert v \right\vert \right) , 
$
and consequently 
\begin{align*}
& \Tr \left( \left\vert   \rho^f_t  - W \left( v \right) \ketbra{e_0}{e_0} W \left( - v \right) \right\vert  \right)
=
\Tr \left( \left\vert W \left( v \right) \left(
\widetilde{\rho}_t - \ketbra{e_0}{e_0}  \right) W \left( - v \right) \right\vert  \right)
\\
& \leq 
\left\Vert W \left( v \right) \right\Vert \left\Vert W \left( - v \right) \right\Vert 
\Tr \left( \left\vert    \widetilde{\rho}_t -  \ketbra{e_0}{e_0} \right\vert  \right) 
 =
\Tr \left( \left\vert    \widetilde{\rho}_t -  \ketbra{e_0}{e_0} \right\vert  \right) 
\leq
2 \hbox{\rm e}^{- \kappa t} \left( \sqrt{ \Tr \left( \varrho^f \, N  \right) } +  \left\vert v \right\vert \right) .
\end{align*} 
Since 
$
W \left( v \right) e_0 = W \left( v \right)  e \left( 0 \right)
=
\exp \left(  - \left\vert v \right\vert ^2 /2  \right) e \left(  v \right) 
$,
\[
W \left( v \right) \ketbra{e_0}{e_0} W \left( - v \right)
=
\exp \left(- \left\vert \alpha_0  / \left( \kappa  + \mathrm{i} \omega \right) \right\vert ^2  \right) 
\ketbra{e \left( \frac{\alpha_0 }{ \kappa  + \mathrm{i} \omega } \right)} {e \left( \frac{\alpha_0 }{ \kappa  + \mathrm{i} \omega } \right)} . 
\]
\end{proof}

Finally,
applying Lemma \ref{lem:ConvAtom} we deduce the convergence to $0$ of  the non-diagonal components of some representation of $ \rho^R_t \left( \varrho \right) $ as 
$ \mathfrak{L}_{1}\left(  \ell^2 \left(\mathbb{Z}_+ \right)\right) ^{2,2}$ matrix.
Then,
using Lemmata \ref{lem:ConvAtom}  and \ref{lem:ConvFoton} we get (\ref{eq:5.1}).

\begin{proof}[Proof of Theorem \ref{th:ConvAuxiliarQME}]

The  solution of (\ref{eq:9.8}) is denoted by  $\rho^a_t \left( \varrho^a \right)$,
and 
we write $\left( \rho^f_t \right)_{t \geq 0}$ for the semigroup $N$-solution  of 
the quantum master equation (\ref{eq:9.6}) (see  \cite{MoraAP} for details).
Due to $\varrho$ is $N$-regular,
$
\varrho = \mathbb{E}\left\vert \xi_+ \otimes e_+ + \xi_- \otimes e_- \rangle \langle \xi_+ \otimes e_+ + \xi_- \otimes e_- \right\vert
$
with $\xi_{\pm} \in  L_{N}^{2}\left( \mathbb{P}, \ell^2 \left(\mathbb{Z}_+ \right) \right)$,
and so
\begin{equation}
\label{eq:9.7}
 \varrho   = 
\varrho_{++} \otimes\ketbra{e_{+}}{e_{+}}  + \varrho_{+-} \otimes\ketbra{e_{+}}{e_{-}}  +  \varrho_{-+} \otimes\ketbra{e_{-}}{e_{+}}
 + \varrho_{--} \otimes\ketbra{e_{-}}{e_{-}}  ,
\end{equation}
where
$
\varrho_{\eta \widetilde{\eta} } =  
\mathbb{E}\left\vert \xi_{\eta}  \rangle \langle \xi_{ \widetilde{\eta}} \right\vert
$.
Since the right-hand term of  (\ref{eq:3.2g}) is equal to 
$\mathcal{L}_{\star}^f \otimes I \left(  \rho^R_t \left( \varrho \right) \right) 
+ I \otimes \mathcal{L}_{\star}^a \left(  \rho^R_t \left( \varrho \right) \right)$,
where $\mathcal{L}_{\star}^f$ and  $\mathcal{L}_{\star}^a $ 
are as in (\ref{eq:R8.4f}) and (\ref{eq:R8.4a}), respectively, 
from (\ref{eq:9.7}) we obtain
\begin{equation}
\label{eq:9.9}
 \begin{aligned}
 \rho^R_t \left( \varrho \right) & = 
  \rho^f_t \left( \varrho_{++}  \right)  \otimes \rho^a_t \left( \ketbra{e_{+}}{e_{+}}   \right)
  + \rho^f_t \left( \varrho_{+-}  \right)  \otimes \rho^a_t \left( \ketbra{e_{+}}{e_{-}} \right)
\\ 
& \quad
+ \rho^f_t \left(  \varrho_{-+}  \right)  \otimes \rho^a_t \left( \ketbra{e_{-}}{e_{+}} \right)
+ \rho^f_t \left(  \varrho_{--}  \right)  \otimes \rho^a_t \left( \ketbra{e_{-}}{e_{-}}  \right)   .
 \end{aligned}
\end{equation}

Combining  $\Tr \left( \ketbra{e_{\pm}}{e_{\mp}} \right) = 0$ with Lemma  \ref{lem:ConvAtom} 
we deduce that
\begin{align*}
 \Tr \left( \left\vert   
 \rho^f_t \left( \varrho_{ \substack{+ - \\ - +} }  \right)  \otimes \rho^a_t \left( \ketbra{e_{\pm}}{e_{\mp}} \right)
\right\vert  \right)
& =
\Tr \left( \left\vert   \rho^f_t \left( \varrho_{\substack{+ - \\ - +} }  \right) \right\vert  \right)
\Tr \left( \left\vert  \rho^a_t \left( \ketbra{e_{\pm}}{e_{\mp}} \right)
\right\vert  \right)
\\
& 
 \leq
\Tr \left( \left\vert    \varrho_{\substack{+ - \\ - +} }  \right\vert  \right)
\Tr \left( \left\vert  \rho^a_t \left( \ketbra{e_{\pm}}{e_{\mp}} \right)
\right\vert  \right) 
\leq
4 \exp \left( - \gamma  \, t \right) \Tr \left( \left\vert    \varrho_{\substack{+ - \\ - +} }  \right\vert  \right) ,
\end{align*}
and so
\begin{equation}
 \label{eq:9.10}
 \Tr \left( \left\vert   
 \rho^f_t \left( \varrho_{ \substack{+ - \\ - +} }  \right)  \otimes \rho^a_t \left( \ketbra{e_{\pm}}{e_{\mp}} \right)
\right\vert  \right)
\leq
4 \exp \left( - \gamma  \, t \right) 
\hspace{2cm} \forall t \geq 0 ,
\end{equation}
because
$
\Tr \left( \left\vert    \varrho_{ \substack{+ - \\ - +} }  \right\vert  \right)
\leq 
\mathbb{E} \Tr \left( \left\vert    \xi_{\pm}  \rangle \langle \xi_{ \mp}  \right\vert  \right)
=
\mathbb{E} \left\Vert   \xi_{\pm} \right\Vert \left\Vert   \xi_{\mp} \right\Vert 
\leq
\sqrt{\mathbb{E} \left\Vert   \xi_{\pm} \right\Vert^2 } \sqrt{\mathbb{E} \left\Vert   \xi_{\mp} \right\Vert^2  }
\leq 
1
$.

Since
\begin{align*}
& 
\Tr \left( \left\vert   
 \rho^f_t \left( \varrho_{\pm \pm}  \right)  \otimes \rho^a_t \left( \ketbra{e_{\pm}}{e_{\pm}} \right)
-
 \Tr \left( \varrho_{\pm \pm}  \right)  \varrho^f_{\infty} \otimes \varrho^a_{\infty} 
\right\vert  \right)
\\
& \leq
\Tr \left( \left\vert   
 \rho^f_t \left( \varrho_{\pm \pm}  \right)  \otimes \rho^a_t \left( \ketbra{e_{\pm}}{e_{\pm}} \right)
-
 \Tr \left( \varrho_{\pm \pm}  \right)  \varrho^f_{\infty} \otimes \rho^a_t \left( \ketbra{e_{\pm}}{e_{\pm}} \right)
\right\vert  \right)
\\
&
\quad +
\Tr \left( \left\vert   
 \Tr \left( \varrho_{\pm \pm}  \right)  \varrho^f_{\infty} \otimes \rho^a_t \left( \ketbra{e_{\pm}}{e_{\pm}} \right)
-
 \Tr \left( \varrho_{\pm \pm}  \right)  \varrho^f_{\infty} \otimes \varrho^a_{\infty} 
\right\vert  \right)
\\
& =
\Tr \left( \left\vert   
 \rho^f_t \left( \varrho_{\pm \pm}  \right)  - \Tr \left( \varrho_{\pm \pm}  \right)  \varrho^f_{\infty} 
\right\vert  \right)
\Tr \left( \left\vert  \rho^a_t \left( \ketbra{e_{\pm}}{e_{\pm}} \right)
\right\vert  \right)
\\
&
\quad +
 \Tr \left( \varrho_{\pm \pm}  \right) 
 \Tr \left( \varrho^f_{\infty}  \right)
\Tr \left( \left\vert   
 \rho^a_t \left( \ketbra{e_{\pm}}{e_{\pm}} \right) - \varrho^a_{\infty} 
\right\vert  \right)
\\
& =
\Tr \left( \left\vert   
 \rho^f_t \left( \varrho_{\pm \pm}  \right)  - \Tr \left( \varrho_{\pm \pm}  \right)  \varrho^f_{\infty} 
\right\vert  \right)
+
\Tr \left( \varrho_{\pm \pm}  \right) \Tr \left( \left\vert   
 \rho^a_t \left( \ketbra{e_{\pm}}{e_{\pm}} \right) - \varrho^a_{\infty} 
\right\vert  \right) ,
\end{align*}
applying Lemmata \ref{lem:ConvAtom} and  \ref{lem:ConvFoton} yields 
\begin{align*} 
&
\Tr \left( \left\vert   
 \rho^f_t \left( \varrho_{\pm \pm}  \right)  \otimes \rho^a_t \left( \ketbra{e_{\pm}}{e_{\pm}} \right)
-
 \Tr \left( \varrho_{\pm \pm}  \right)  \varrho^f_{\infty} \otimes \varrho^a_{\infty} 
\right\vert  \right)
\\ 
& \leq
2 \sqrt{\Tr \left( \varrho_{\pm \pm}  \right)}  \hbox{\rm e}^{- \kappa t} \sqrt{ \Tr \left( \varrho_{\pm \pm}  N  \right) }
+ \Tr \left( \varrho_{\pm \pm}  \right) \hbox{\rm e}^{- \kappa t} 
\frac{2 \left\vert  \alpha_0  \right\vert}{ \sqrt{ \kappa^2   +  \omega^2 }}
+ 4 \hbox{\rm e}^{- \gamma  t} \Tr \left( \varrho_{\pm \pm}  \right) \left( 1 +  \left\vert d \right\vert \right) 
\\
& \leq
2   \hbox{\rm e}^{- \kappa t} \sqrt{ \Tr \left( \varrho_{\pm \pm}  N  \right) }
+ \Tr \left( \varrho_{\pm \pm}  \right) \hbox{\rm e}^{- \kappa t} 
\frac{2 \left\vert  \alpha_0  \right\vert}{ \sqrt{ \kappa^2   +  \omega^2 }}
+ 4 \hbox{\rm e}^{- \gamma  t} \Tr \left( \varrho_{\pm \pm}  \right) \left( 1 +  \left\vert d \right\vert \right) .
\end{align*}
Now, using (\ref{eq:9.9}), (\ref{eq:9.10}),
$
 \Tr \left( \varrho_{++}  \right) + \Tr \left( \varrho_{--}  \right) =1
$
and 
$\Tr \left( \varrho_{\pm \pm}  N  \right) \leq \Tr \left( \varrho \, N  \right)$ we  get (\ref{eq:5.1}). 
\end{proof}

\subsection{Proof of Theorem \ref{th:StatState-LaserE}}
\label{sec:Proof:StatState-LaserE}

%\begin{lemma}
%\label{lem:invLS}
%Suppose that $\varrho$ is a $N^p$-regular density operator
%which is a stationary state for  (\ref{eq:Laser1}) with $\omega \neq 0$,
%where $p \in \mathbb{N}$.
%Then
%$  \Tr\left( a  \varrho \right) = \Tr\left( \sigma^{-}  \varrho  \right) = 0 $ and  $\Tr\left( \sigma^{3} \varrho  \right) = d $.
%\end{lemma}
%
%\begin{proof}
%According to Theorem \ref{th:EyU-LaserE} we have that
%$ A \left( t \right) = \Tr\left( a \, \varrho \right) $,
%$ S \left( t  \right)  = \Tr\left( \sigma^{-} \varrho  \right) $
%and
%$ D \left( t  \right) = \Tr\left(  \sigma^{3}  \varrho \right) $
%satisfy (\ref{eq:Lorenz}.
%Thus,
%the lemma follows from Theorem \ref{th:EquLorenzEquations-Laser}.
%\end{proof}

For completeness, 
we start by examining the fix points of the Maxwell-Bloch equations \eqref{eq:Lorenz}.

\begin{lemma}
\label{lem:ConstSolution}
Assume that
$d\in \left]-1,1 \right[$,  
$\kappa,\gamma \in \left] 0, + \infty \right[ $
and 
that 
$ g, \omega $ are real numbers  different  from $0$.
Then,
the unique constant solution of \eqref{eq:Lorenz} is 
 $\left( A \left( t \right), S \left( t \right), D \left( t \right)  \right) = \left( 0, 0, d \right)$.
\end{lemma}

\begin{proof}
 Let $\left( A \left( t \right), S \left( t \right), D \left( t \right)  \right) = \left( A, S, D \right)$
be a constant solution of (\ref{eq:Lorenz}).
Then
\begin{subequations}
 \begin{align}
 \label{eq:8.15a-n}
 - \left( \kappa + \mathrm{i} \omega \right)  A  + g \ S
  = 0 ,
 \\
 \label{eq:8.15b-n}
 - \left( \gamma + \mathrm{i} \omega \right)   \ S  + g \ A \  D 
  = 0 ,
 \\
  \label{eq:8.15c-n}
 - 4 g \ \Re \left( A \  \overline{S}   \right)
 - 2 \gamma \left(  D - d \right) 
  = 0 .
\end{align} 
\end{subequations}
Combining (\ref{eq:8.15a-n}) with (\ref{eq:8.15b-n}) we deduce that
\begin{equation}
\label{eq:9.16_n}
 A \Bigl( 
 -  \left( \gamma + \mathrm{i} \omega \right) \left( \kappa + \mathrm{i} \omega  \right) + g^2  \  D
\Bigr)
 = 0 .
\end{equation}
Using $\omega \neq 0$ and $ \kappa, \gamma > 0$ we get 
$\left( \gamma + \mathrm{i} \omega \right) \left( \kappa + \mathrm{i} \omega  \right) \notin \mathbb{R}$.
Since $D \in \mathbb{R}$,   
$
-    \left( \gamma + \mathrm{i} \omega \right) \left( \kappa + \mathrm{i} \omega  \right)  + g^2  \  D
\neq 0 
$,
and so (\ref{eq:9.16_n}) yields $A = 0$.
From (\ref{eq:8.15a-n})-(\ref{eq:8.15c-n}) we obtain $S = 0$ and  $D = d $.
\end{proof}

\begin{proof}[Proof of Theorem \ref{th:StatState-LaserE}]
First,
we check by direct computation that  $\varrho_{\infty}$, given by  (\ref{eq:I10}),
is a constant solution of  (\ref{eq:Laser1}).
Since 
$
 \Tr\left( \sigma^{-} \varrho_{\infty}  \right) 
= 
\frac{d+1}{2} \langle e_+ , \sigma^{-}  e_+ \rangle
+
\frac{1-d}{2} \langle e_- , \sigma^{-}  e_- \rangle
=
0 
$
and
$ \Tr\left( a \, \varrho_{\infty}  \right) = \langle e_0 , a e_0 \rangle = 0 $,
$
\left[ 
\Tr\left( \sigma^{-}   \varrho_{\infty}   \right) a^\dagger  -  \Tr\left( \sigma^{+}  \varrho_{\infty}   \right) a 
+
 \Tr\left( a^\dagger   \varrho_{\infty}   \right) \sigma^{-} -  \Tr\left( a \,   \varrho_{\infty}  \right) \sigma^+ 
,  \varrho_{\infty}  \right] = 0 
$.
Moreover,
using the fact that  $A  \left| x\rangle \langle y \right| B  =  \left| A  x\rangle \langle B^{\star} y \right|  $
for any operators $A,B$ in $\mathfrak{h}$, $x \in \mathcal{D}\left( A\right)$
and  $y \in \mathcal{D}\left( B^{\star} \right)$,
we obtain 
$ \mathcal{L}_{\star}^h \, \varrho_{\infty} = 0 $,
where $\mathcal{L}_{\star}^h$ is defined by (\ref{eq:3.21}).
Therefore,
\[
\mathcal{L}_{\star}^h \, \varrho_{\infty}
+   
g  \left[ 
\Tr\left( \sigma^{-}   \varrho_{\infty}   \right) a^\dagger  -  \Tr\left( \sigma^{+}  \varrho_{\infty}   \right) a 
+
 \Tr\left( a^\dagger   \varrho_{\infty}   \right) \sigma^{-} -  \Tr\left( a \,   \varrho_{\infty}  \right) \sigma^+ 
,  \varrho_{\infty}  \right]
=
0 ,
\]
and so $ \varrho_{\infty} $ is a stationary state for  (\ref{eq:Laser1}),
which is $N^p$-regular for all $p \in \mathbb{N}$.

Next,
we deal with the uniqueness of the $N$-regular stationary state for  (\ref{eq:Laser1}) with $\omega \neq 0$.
Let $ \widetilde{\varrho} $ be a $N$-regular stationary state for  (\ref{eq:Laser1}).
Then $\rho_t \equiv \widetilde{\varrho}$ satisfies  (\ref{eq:Laser1}),
and so 
$ A \left( t \right)  \equiv \Tr\left( a \,  \widetilde{\varrho}  \right) $,
$ S \left( t  \right) \equiv \Tr\left( \sigma^{-}  \widetilde{\varrho}   \right) $
and
$ D \left( t  \right) \equiv \Tr\left(  \sigma^{3}   \widetilde{\varrho}  \right) $
is a constant solution to the Maxwell-Bloch equations (\ref{eq:Lorenz})
(see, e.g., \cite{FagMora2019}).
According to Lemma \ref{lem:ConstSolution} it follows that
$ \Tr\left( a \, \widetilde{\varrho} \right) = \Tr\left( \sigma^{-}  \widetilde{\varrho}  \right) = 0 $ 
and  $\Tr\left( \sigma^{3} \widetilde{\varrho}  \right) = d $,
because $\omega \neq 0$.
Therefore, 
\[
0
=
\mathcal{L}_{\star}^h \, \widetilde{\varrho} 
+   
g  \left[ 
\Tr\left( \sigma^{-}   \widetilde{\varrho}   \right) a^\dagger  -  \Tr\left( \sigma^{+}  \widetilde{\varrho}    \right) a 
+
 \Tr\left( a^\dagger  \widetilde{\varrho}   \right) \sigma^{-} -  \Tr\left( a \,  \widetilde{\varrho}   \right) \sigma^+ 
,  \widetilde{\varrho}  \right]
=
\mathcal{L}_{\star}^h \, \widetilde{\varrho} ,
\]
that is, $ \widetilde{\varrho} $ satisfies the linear equation 
$\mathcal{L}_{\star}^h \, \widetilde{\varrho} = 0$.
Hence, $ \widetilde{\varrho} $ is a stationary state for the linear quantum master equation  (\ref{eq:3.2}).
Using Corollary \ref{cor:ConvEMCSimple} we obtain 
$
\widetilde{\varrho} =  \rho^h_t \left( \widetilde{\varrho}  \right) \longrightarrow_{t \rightarrow + \infty}  \varrho_{\infty} 
$.
\end{proof}

\subsection{Proof of Theorem \ref{th:FreeS-LaserE}}
\label{sec:Proof:FreeS-LaserE}

Applying Theorem  \ref{th:ConvAuxiliarQME}, 
together with Lemmata  \ref{lem:ConvAtom} and \ref{lem:ConvFoton} used in the proof of Theorem  \ref{th:ConvAuxiliarQME},
we now find the $N$-regular invariant states of the  linear quantum master equation (\ref{eq:3.2g}).

\begin{lemma}
 \label{th:SolEst-Lineal} 
Let $\left( \rho^R_t \left( \varrho \right)  \right)_{t \geq 0}$ be the $N$-weak solution of the linear quantum master equation 
(\ref{eq:3.2g}) with $ \alpha_0,  \beta_0 \in \mathbb{C}$
and initial datum  $\varrho \in \mathfrak{L}_{1,N}^+ \left( \ell^2 \left(\mathbb{Z}_+ \right) \otimes \mathbb{C}^2 \right) $.
Consider the operators $\varrho^f_{\infty}$ and $\varrho^a_{\infty}$ defined in Corollary \ref{cor:Ineq-Conv}.
Then $ \varrho^f_{\infty} \otimes \varrho^a_{\infty}$ is the unique operator
$\hat{\varrho}_{\infty} \in \mathfrak{L}_{1,N}^+ \left( \ell^2 \left(\mathbb{Z}_+ \right) \otimes \mathbb{C}^2 \right) $ for which  
\begin{equation}
 \label{eq:9.12}
 \rho^R_t  \left( \hat{\varrho}_{\infty}  \right)  = \hat{\varrho}_{\infty} 
  \hspace{3cm}
  \forall t \geq 0  .
\end{equation}
\end{lemma}

\begin{proof}
Since (\ref{eq:9.8}) is a complex ordinary differential equation,
using Lemma \ref{lem:ConvAtom} we deduce that $ \varrho^a_{\infty}$ is a fix point of (\ref{eq:9.8}),
and so 
$\mathcal{L}_{\star}^a \left( \varrho^a_{\infty} \right) = 0$.
Moreover,
from  the proof of Lemma \ref{lem:ConvFoton} we obtain that 
for any $N$-regular density operator $\widetilde{  \varrho }$ in $ \ell^2 \left(\mathbb{Z}_+ \right)$,
\[
\mathcal{L}_{\star}^f \left(  
W \left( \frac{\alpha_0}{\kappa  + \mathrm{i} \omega} \right)  \widetilde{  \varrho }  \, 
W \left( -  \frac{\alpha_0}{\kappa  + \mathrm{i} \omega} \right) 
\right) 
=
W \left( \frac{\alpha_0}{\kappa  + \mathrm{i} \omega} \right) 
 \widetilde{ \mathcal{L}_{\star} } \left(  \widetilde{  \varrho }  \right)
W \left( -  \frac{\alpha_0}{\kappa  + \mathrm{i} \omega} \right) 
\]
(see relation (\ref{eq:9.3}), 
where $ \widetilde{ \mathcal{L}_{\star} }$ is as in (\ref{eq:9.4}).
As  $ \widetilde{ \mathcal{L}_{\star} } \left( \ketbra{e_0}{e_0}  \right) = 0$ we have
\[
 \mathcal{L}_{\star}^f \left(  
W \left( \frac{\alpha_0}{\kappa  + \mathrm{i} \omega} \right) \ketbra{e_0}{e_0} 
W \left( -  \frac{\alpha_0}{\kappa  + \mathrm{i} \omega} \right) 
\right) 
 =
W \left( \frac{\alpha_0}{\kappa  + \mathrm{i} \omega} \right)
\widetilde{ \mathcal{L}_{\star} } \left( \ketbra{e_0}{e_0}  \right)
W \left( -  \frac{\alpha_0}{\kappa  + \mathrm{i} \omega} \right) 
=
0 .
\]
Hence,
$ 
\mathcal{L}_{\star}^f \left(   \varrho^f_{\infty} \right) = 0 
$
since
$
W \left( \frac{\alpha_0}{\kappa  + \mathrm{i} \omega} \right) e_{0}
=
\exp \left(  - \left\vert \frac{\alpha_0}{\kappa  + \mathrm{i} \omega} \right\vert ^2 /2  \right)
e \left(  \frac{\alpha_0}{\kappa  + \mathrm{i} \omega} \right)
$.
Therefore, 
\[
\mathcal{L}_{\star}^f \otimes I \left( \varrho^f_{\infty} \otimes \varrho^a_{\infty} \right) 
+ 
I \otimes \mathcal{L}_{\star}^a \left( \varrho^f_{\infty} \otimes \varrho^a_{\infty} \right)
= 0 .
\]
This gives 
$
 \rho^R_t \left( \varrho^f_{\infty} \otimes \varrho^a_{\infty}  \right)  = \varrho^f_{\infty} \otimes \varrho^a_{\infty}
$
for all $t \geq 0$.

In order to prove the uniqueness of the $N$-regular invariant state of  $\rho^R_t \left( \cdot \right)$,
we now consider 
$\hat{\varrho}_{\infty} \in \mathfrak{L}_{1,N}^+ \left( \ell^2 \left(\mathbb{Z}_+ \right) \otimes \mathbb{C}^2 \right) $
satisfying (\ref{eq:9.12}).
Then, applying Theorem \ref{th:ConvAuxiliarQME} yields 
$
 \hat{\varrho}_{\infty} 
 = \lim_{t \rightarrow + \infty}  \rho^R_t \left( \hat{\varrho}_{\infty}  \right) 
 = \varrho^f_{\infty} \otimes \varrho^a_{\infty} 
 $
 in 
 $ \mathfrak{L}_{1} \left(  \ell^2 \left(\mathbb{Z}_+ \right) \otimes \mathbb{C}^2 \right) $.
\end{proof}

\begin{proof}[Proof of Theorem \ref{th:FreeS-LaserE}]
By Stone's theorem,
the self-adjoint operator $  \omega  \left( N +  \sigma^3 / 2 \right)$ generates 
the strongly continuous one-parameter unitary group 
$
\left(\hbox{\rm e}^{   \mathrm{i}  \omega  \left( N +  \sigma^3 / 2 \right) t  }\right)_{t \in \mathbb{R}}.
$
% $\left( \exp \left(  \mathrm{i}  \omega  \left( N +  \sigma^3 / 2 \right) t \right) \right)_{t \in \mathbb{R}}$.
In order to describe the physical system in the interaction picture we set 
\begin{equation*}
\widetilde{ \rho }_t =
\exp \left(  \mathrm{i}  \omega  \left( N +  \sigma^3 / 2 \right) t \right)
\rho_t
\exp \left( - \mathrm{i}  \omega  \left( N +  \sigma^3 / 2 \right) t \right)
\hspace{1cm}
\forall t \geq 0 .
\end{equation*}
Since $N$ commutes with  $\sigma^3$,
$\rho_t \in \mathfrak{L}_{1,N}^+ \left( \ell^2 \left(\mathbb{Z}_+ \right) \otimes \mathbb{C}^2 \right)$ iff 
$\widetilde{ \rho }_t \in \mathfrak{L}_{1,N}^+ \left( \ell^2 \left(\mathbb{Z}_+ \right) \otimes \mathbb{C}^2 \right)$.
Hence, 
$\rho_t $ is a $N$-regular  free interaction solution  to  (\ref{eq:Laser1}) iff
\begin{equation}
\label{eq:I12}
\widetilde{ \rho }_t = \rho_0 \in  \mathfrak{L}_{1,N}^+ \left( \ell^2 \left(\mathbb{Z}_+ \right) \otimes \mathbb{C}^2 \right)
\hspace{2cm}
\forall t \geq 0 .
\end{equation}
A careful computation shows that $\rho_t $ is a $N$-weak solution to (\ref{eq:Laser1})
iff 
$\widetilde{ \rho }_t $ is a $N$-weak solution to
\begin{equation}
 \label{eq:Laser6}
 \begin{aligned}
  \frac{d }{dt} \widetilde{ \rho }_t
 &  = 
  g \left[  \Tr\left( \sigma^{-}  \widetilde{ \rho }_t  \right) a^\dagger
                   -  \Tr\left( \sigma^{+}  \widetilde{ \rho }_t  \right) a
                   +  \Tr\left( a^\dagger  \widetilde{ \rho }_t  \right) \sigma^{-}
                   -  \Tr\left( a \,  \widetilde{ \rho }_t  \right) \sigma^+
                                      ,  \widetilde{ \rho }_t \right]  
 \\
 & \quad +
 \kappa \left( 2 \, a\,  \widetilde{ \rho }_t  a^\dagger 
 -  a^\dagger a  \widetilde{ \rho }_t -   \widetilde{ \rho }_t  a^\dagger a\right) 
+  \frac{ \gamma(1-d) }{ 2 }\left( 2 \, \sigma^{-}  \widetilde{ \rho }_t \,\sigma^+ 
- \sigma^+\sigma^{-}  \widetilde{ \rho }_t 
- \widetilde{ \rho }_t \,\sigma^+\sigma^{-}\right) 
\\ 
& \quad 
+  \frac{ \gamma(1+d) }{ 2 }\left( 2 \, \sigma^+  \widetilde{ \rho }_t \,\sigma^{-} 
- \sigma^{-}\sigma^+  \widetilde{ \rho }_t
- \widetilde{ \rho }_t \,\sigma^{-}\sigma^+\right) .
 \end{aligned}
\end{equation}
Therefore,
$\rho_t $ is a $N$-regular  free interaction solution  to  (\ref{eq:Laser1}) iff 
$ \rho_0 $ is a $N$-regular  stationary state for (\ref{eq:Laser6}).

Suppose that $ \rho_0 $ is a  constant $N$-regular solution to the non-linear evolution equation  (\ref{eq:Laser6}).
Then
$  \Tr\left( a \, \rho_t \right) \equiv \Tr\left( a \, \rho_0 \right) $,
$   \Tr\left( \sigma^{-} \rho_t  \right) \equiv \Tr\left( \sigma^{-} \rho_0  \right) $
and
$ \Tr\left(  \sigma^{3} \rho_t \right)   \equiv  \Tr\left(  \sigma^{3} \rho_0 \right) $,
and so  $ \rho_0 $ is a $N$-regular stationary state of the linear quantum master equation
(\ref{eq:3.2g}) with 
$ \omega = 0$,
$ \alpha_0  =  g \, \Tr\left( \sigma^{-} \rho_0  \right) $
and
$ \beta_0  =  g \, \Tr\left( a \, \rho_0 \right) $.
Moreover,
$\Tr\left( \sigma^{-} \rho_0  \right) $ and $\Tr\left( a \, \rho_0 \right) $
are given by the constant solutions of the Maxwell-Bloch equations \eqref{eq:Lorenz} with $ \omega = 0$.
Thus,
we next obtain all  $N$-regular  stationary states for non-linear evolution equation (\ref{eq:Laser6})
by finding the constant $N$-regular  solutions of the linear evolution equation (\ref{eq:3.2g}) with $ \omega = 0$,
$ \alpha_0  =  g \, S \left( 0 \right)$ and $ \beta_0  = g \, A  \left( 0 \right)$,
where $A  \left( 0 \right)$, $S \left( 0 \right)$ and  $D \left( 0 \right)$ is a fix-point of \eqref{eq:Lorenz} with $ \omega = 0$.

Suppose that (\ref{eq:I12}) holds.
Since the functions
$ t \mapsto  \Tr\left( a \, \rho_0 \right) $,
$  t \mapsto \Tr\left( \sigma^{-} \rho_0  \right) $
and
$  t \mapsto \Tr\left(  \sigma^{3} \rho_0 \right) $
satisfy (\ref{eq:Lorenz}) with $\omega = 0$
(see \cite{FagMora2019}),
\begin{subequations}
 \begin{align}
 \label{eq:8.15a}
  - \kappa \,  \Tr\left( a \, \rho_{0} \right) + g \ \Tr\left( \sigma^{-} \rho_{0} \right)
  = 0 ,
 \\
 \label{eq:8.15b}
  - \gamma  \,  \Tr\left( \sigma^{-} \rho_{0} \right) 
 + g \ \Tr\left( a \, \rho_{0}  \right)  \Tr\left( \sigma^{3} \rho_{0}   \right) 
  = 0 ,
 \\
  \label{eq:8.15c}
  2 g \ \Re \left(
  \Tr\left( a \, \rho_{0}  \right) \  \overline{ \Tr\left(  \sigma^{-}  \rho_{0}  \right) }
 \right)
 +  \gamma \left(  \Tr\left(  \sigma^{3} \rho_{0} \right) -d \right) 
   = 0 .
\end{align} 
\end{subequations}
Combining (\ref{eq:8.15a}) with (\ref{eq:8.15b}) we obtain
$
 \Tr\left( a \,  \rho_{0} \right) \Bigl(
-   \gamma \kappa + g^2  \  \Tr\left( \sigma^{3}  \rho_{0}  \right)
\Bigr)
 = 0 
$.
Then  $\Tr\left( a \,  \rho_{0} \right) = 0$ or  $  g^2  \,  \Tr\left( \sigma^{3}  \rho_{0} \right) =  \gamma \kappa $.

Asume $ \Tr\left( a \,  \rho_{0} \right) = 0$, together with (\ref{eq:I12}).
Then (\ref{eq:8.15a}) and (\ref{eq:8.15c}) lead to
$ \ \Tr\left( \sigma^{-} \rho_{0} \right) = 0$ and $ \Tr\left(  \sigma^{3} \rho_{0} \right) = d$.
So,
\begin{equation}
\label{eq:9.13}
\Tr\left( a  \, \widetilde{ \rho }_t \right) = \Tr\left( \sigma^{-}  \widetilde{ \rho }_t  \right) = 0 
\text{ and } \Tr\left( \sigma^{3}  \widetilde{ \rho }_t  \right) = d .
\end{equation}
Therefore,
$ \rho_0 $ is a N-regular stationary state for  (\ref{eq:3.2}) with $ \omega = 0 $.
Using Corollary \ref{cor:ConvEMCSimple}  gives 
$
\rho_0
=
\varrho_{\infty}
$,
where $\varrho_{\infty}$ is defined by (\ref{eq:I10}).
Since
$\Tr\left( a \, \varrho_{\infty} \right) = \Tr\left( \sigma^{-}   \varrho_{\infty}  \right) = 0 $, 
$\Tr\left( \sigma^{3}   \varrho_{\infty}  \right) = d $
and  
$ \mathcal{L}_{\star}^h \, \varrho_{\infty} = 0 $,
(\ref{eq:I10}) is indeed a constant $N$-regular solution to  (\ref{eq:Laser6}).
Summarizing,
$\varrho_{\infty} $, given by  (\ref{eq:I10}),
is the unique $N$-regular stationary state for (\ref{eq:Laser6}) 
satisfying $ \Tr\left( a \,  \rho_{0} \right) = 0$.
This yields  the free interaction solution  to  (\ref{eq:Laser1}):
 \begin{align*}
 \rho_t
 & =
 \exp \left( - \mathrm{i}  \omega  \left( N +  \sigma^3 / 2 \right) t \right) 
\widetilde{ \rho }_t 
\exp \left(  \mathrm{i}  \omega  \left( N +  \sigma^3 / 2 \right) t \right)
\\
& =
\exp \left( - \mathrm{i}  \omega  \left( N +  \sigma^3 / 2 \right) t \right) 
\varrho_{\infty}
\exp \left(  \mathrm{i}  \omega  \left( N +  \sigma^3 / 2 \right) t \right) 
\\
&
= 
\ketbra{e_0}{e_0}
\otimes 
\left( \frac{d+1}{2}  \ketbra{e_+}{e_+}  + \frac{1-d}{2}   \ketbra{e_-}{e_-}  \right) 
=
\varrho_{\infty} .
\end{align*}

On the other hand, 
suppose that $ \Tr\left( a \,  \rho_{0} \right) \neq 0$ and that (\ref{eq:I12}) holds.
Then, 
$  g^2  \,  \Tr\left( \sigma^{3}  \rho_{0} \right) =  \gamma \kappa $,
and (\ref{eq:8.15a}) implies that $g \neq 0$.
Hence, 
$ \Tr\left( \sigma^{3}  \rho_{0}  \right) = \gamma \kappa / \left( g^2 \right) $.
Using  (\ref{eq:8.15b})  and (\ref{eq:8.15c}) we deduce that
\begin{equation}
 \label{eq:9.31}
  \left\vert \Tr\left( a  \, \rho_0 \right)  \right\vert ^2 
 =
 \frac{\gamma}{2 \kappa g^2} \left( d g^2 - \gamma \kappa \right) .
\end{equation}
Therefore $ d g^2 > \gamma \kappa$, i.e., $ C_\mathfrak{b}  > 1 $.
Hence,
there are no $N$-regular  free interaction solution  to  (\ref{eq:Laser1}) 
with  $ \Tr\left( a \,  \rho_{0} \right) \neq 0$ in case  $ d g^2 \leq \gamma \kappa$,
and so from the previous paragraph we conclude  that 
the state (\ref{eq:I10})  is the unique $N$-regular  free interaction solution  to  (\ref{eq:Laser1}) 
whenever $ C_\mathfrak{b} \leq 1 $.

Let $ C_\mathfrak{b}  > 1 $ and $ \Tr\left( a \,  \rho_{0} \right) \neq 0$.
According to (\ref{eq:9.31}) we have that
there exists  $z \in \mathbb{C}$ with $\left\vert z \right\vert = 1$ such that
$
\Tr\left( a \, \widetilde{ \rho }_t \right)   = z \sqrt{\frac{\gamma}{2 \kappa g^2} \left( d g^2 - \gamma \kappa \right)} 
$,
and so (\ref{eq:8.15a}) yields
$
\Tr\left( \sigma^{-}  \widetilde{ \rho }_t  \right) 
=
 \frac{\kappa \, z }{g} \sqrt{\frac{\gamma}{2 \kappa g^2} \left( d g^2 - \gamma \kappa \right)} 
$.
Since (\ref{eq:I12}) holds, from (\ref{eq:Laser6}) it follows that
$\rho_0$ is a $N$-regular stationary state  for 
(\ref{eq:3.2g}) with   
\begin{equation}
 \label{eq:ValoresAlphaBeta}
 \alpha_0
 =
  z \, \kappa \,  \sqrt{\frac{\gamma}{2 \kappa g^2} \left( d g^2 - \gamma \kappa \right)},
 \qquad 
 \beta_0
 =
 g \, z \, \sqrt{\frac{\gamma}{2 \kappa g^2} \left( d g^2 - \gamma \kappa \right)} ,
\end{equation}
and $\omega = 0$.
Applying Lemma \ref{th:SolEst-Lineal} we obtain  $\varrho_0  = \varrho_{\infty} \left( z \right) $ with 
\begin{equation}
 \label{eq:9.22}
\varrho_{\infty} \left( z \right) 
=
\ketbra{ \mathcal{E} \left( \frac{z \gamma \sqrt{ C_\mathfrak{b} -1 } }{ \sqrt{2} \left| g \right| }  \right)} 
{ \mathcal{E} \left( \frac{z \gamma \sqrt{ C_\mathfrak{b} -1 }}{ \sqrt{2} \left| g \right| } \right)}
\otimes 
\begin{pmatrix}
 \frac{1}{2} \left( 1 + \frac{ d }{ C_\mathfrak{b} }  \right)
 &
 \frac{z \kappa \gamma }{ \sqrt{2} g \left| g \right|} \sqrt{ C_\mathfrak{b} -1 }
 \\
 \frac{ \bar{z} \kappa \gamma }{ \sqrt{2} g \left| g \right|} \sqrt{ C_\mathfrak{b} -1 }
 &
 \frac{1}{2} \left( 1 - \frac{d }{ C_\mathfrak{b} }  \right) 
\end{pmatrix} .
\end{equation}
Then,
the only candidate for 
$N$-regular  stationary states of (\ref{eq:Laser6})  with the property   $ \Tr\left( a \,  \rho_{0} \right) \neq 0$ are:
$ \varrho_{\infty} \left( z \right) $ for any  $ \left\vert z \right\vert = 1$.

Consider $\left\vert z \right\vert = 1$, and let $ d g^2 > \gamma \kappa$.
By $ a \, \mathcal{E} \left( \zeta \right)  =  \zeta \, \mathcal{E} \left( \zeta \right)$
for any  $\zeta \in \mathbb{C}$,
a direct computation yields 
$
\Tr\left( a \, \varrho_{\infty} \left( z \right) \right)   = z \sqrt{\frac{\gamma}{2 \kappa g^2} \left( d g^2 - \gamma \kappa \right)}
$.
Moreover, 
a direct calculation gives  
$
\Tr\left( \sigma^{3} \varrho_{\infty} \left( z \right) \right) = \frac{ \gamma \kappa }{ g^2 }
$,
and
$
\Tr\left( \sigma^{-}  \varrho_{\infty} \left( z \right)  \right) 
= \frac{\kappa \, z }{g}  \sqrt{\frac{\gamma}{2 \kappa g^2} \left( d g^2 - \gamma \kappa \right)}  
$.
Therefore,
$
g \, \Tr\left( \sigma^{-}  \varrho_{\infty} \left( z \right)  \right)  =  \alpha_0 
$
and
$
g \, \Tr\left( a \, \varrho_{\infty} \left( z \right) \right) =  \beta_0
$,
where $ \alpha_0 $ and $  \beta_0 $ are as in (\ref{eq:ValoresAlphaBeta}).
Using Lemma \ref{th:SolEst-Lineal} we get that 
$\varrho_{\infty} \left( z \right)$ is a  $N$-regular  stationary state of (\ref{eq:Laser6}).
Then,
in addition to  (\ref{eq:I10}),
the only $N$-regular  stationary states for (\ref{eq:Laser6}) with $ d g^2 > \gamma \kappa$
are given by (\ref{eq:9.22}) 
for any complex number $\left\vert z \right\vert = 1$.

Since
$
\rho_t
=
\exp \left( - \mathrm{i}  \omega  \left( N +  \sigma^3 / 2 \right) t \right) 
\widetilde{ \rho }_t 
\exp \left(  \mathrm{i}  \omega  \left( N +  \sigma^3 / 2 \right) t \right) 
$,
all  non-constant $N$-regular  free interaction solution  to  (\ref{eq:Laser1}) are:
\begin{align*}
&  
\ketbra{  \hbox{\rm e}^{  - \mathrm{i} \omega N t } \mathcal{E} \left( \frac{z \gamma \sqrt{ C_\mathfrak{b} -1 } }{ \sqrt{2} \left| g \right| }  \right)} 
{  \hbox{\rm e}^{  - \mathrm{i} \omega N t } \mathcal{E} \left( \frac{z \gamma \sqrt{ C_\mathfrak{b} -1 }}{ \sqrt{2} \left| g \right|  } \right)}
\otimes 
\\
& \hspace{6cm}
 \hbox{\rm e}^{ - \mathrm{i} \frac{\omega}{2}  \sigma^3  t }
\begin{pmatrix} 
 \frac{1}{2} \left( 1 + \frac{ d }{ C_\mathfrak{b} }  \right)
 &
 \frac{z \kappa \gamma }{ \sqrt{2} g \left| g \right| } \sqrt{ C_\mathfrak{b} -1 }
 \\
 \frac{ \bar{z} \kappa \gamma }{ \sqrt{2} g \left| g \right|  } \sqrt{ C_\mathfrak{b} -1 }
 &
 \frac{1}{2} \left( 1 - \frac{d }{ C_\mathfrak{b} }  \right) 
\end{pmatrix} 
 \hbox{\rm e}^{ \mathrm{i} \frac{\omega}{2}  \sigma^3  t } ,
\end{align*}
where $\left\vert z \right\vert = 1$,
and therefore they are:
\begin{align*}
& 
\ketbra{ \mathcal{E} \left( \frac{z \gamma \sqrt{ C_\mathfrak{b} -1 } }{ \sqrt{2} \left| g \right|  }   \hbox{\rm e}^{-i \omega t}   \right)} 
{ \mathcal{E} \left(  \frac{z \gamma \sqrt{ C_\mathfrak{b} -1 }}{ \sqrt{2} \left| g \right|  }  \hbox{\rm e}^{-i \omega t}  \right)}
\otimes 
\begin{pmatrix}
 \frac{1}{2} \left( 1 + \frac{ d }{ C_\mathfrak{b} }  \right)
 &
  \hbox{\rm e}^{-i \omega t}  \frac{z \kappa \gamma }{ \sqrt{2} g \left| g \right| } \sqrt{ C_\mathfrak{b} -1 }
 \\
  \hbox{\rm e}^{i \omega t}  \frac{ \bar{z} \kappa \gamma }{ \sqrt{2} g \left| g \right| } \sqrt{ C_\mathfrak{b} -1 }
 &
 \frac{1}{2} \left( 1 - \frac{d }{ C_\mathfrak{b} }  \right) 
\end{pmatrix} 
\end{align*}
for any  $\left\vert z \right\vert = 1$.
\end{proof}

\subsection{Proof of Theorem \ref{th:LongTimeD}}
\label{sec:Proof:LongTimeD}

\begin{proof}[Proof of Theorem \ref{th:LongTimeD}]
Recall that  $\Tr\left( a \, \rho_{t} \right)$, $\Tr\left( \sigma^{-} \rho_{t}  \right)$ and $\Tr\left(  \sigma^{3}  \rho_{t} \right)$
are given by the Maxwell-Bloch equations (\ref{eq:Lorenz})
(see \cite{FagMora2019}).
Since
$\Tr\left( a \, \varrho \right) = \Tr\left( \sigma^{-} \varrho  \right) = 0$,
from (\ref{eq:Lorenz}) it follows that $\Tr\left( a \, \rho_{t} \right ) = \Tr\left( \sigma^{-} \rho_{t}  \right) = 0$ for all $t \geq 0$.
Therefore, 
$ \rho_{t} $ solves  (\ref{eq:3.2}) with initial condition $\varrho$,
and hence $  \rho_{t} = \rho^h_t \left( \varrho \right)$,
where $\rho^h_t \left( \varrho \right)$ is the $N$-weak solution of  (\ref{eq:3.2}).
Applying Corollary \ref{cor:ConvEMCSimple}  gives (\ref{eq:8.22}).
\end{proof}

\subsection{Proof of Theorem \ref{th:LongTime}}
\label{sec:Proof:LongTime}

% Corollary \ref{cor:Ineq-Conv},  

First, we establish the equation of motion of the mean value of the number operator
by applying an Ehrenfest-type theorem developed in \cite{FagMora2013}.

\begin{lemma}
\label{lem:EvolNumero}
Let  $\left( \rho_t \right)_{t \geq 0}$ be the $N$-weak solution to (\ref{eq:AuxiliarGKSL}) with  $N$-regular  initial datum
and  $\alpha, \beta : \left[ 0 , \infty \right[ \rightarrow \mathbb{C}$ continuous.
Then for all $t \geq 0$,
\begin{equation}
\label{eq:8.32}
\frac{d}{dt} \Tr\left( \rho_{t} N  \right)
=
- 2  \kappa \,  \Tr \left( \rho_t \, N \right)
+
2  \Re \left(  \overline{\alpha \left( t \right)}  \, \Tr \left( \rho_t \, a  \right) \right) .
\end{equation}
\end{lemma}

\begin{proof}
Let  $X_{t} \left( \xi \right)$  be the  strong $N$-solution of (\ref{eq:SSE}) with initial datum  
$\xi \in L_{N}^{2}\left( \mathbb{P},  \ell^2 \left(\mathbb{Z}_+ \right) \otimes \mathbb{C}^2 \right) $ satisfying 
$
\rho_0 = \mathbb{E}\left\vert \xi \rangle \langle \xi \right\vert 
$.
According Theorem 4.1 of  \cite{FagMora2013} we have 
\begin{equation}
\label{eq:8.7}
 \begin{aligned}
\Tr \left(N \rho_t \right)
& = 
\Tr \left(N \rho_0 \right)
+
\int_0^t  \mathbb{E} \left( 2 \ \Re \left\langle N  X_{s} \left( \xi \right) , G \left( s \right) X_{s} \left( \xi \right) \right\rangle \right) ds
\\
& \quad
+
\sum_{\ell = 1}^{3} \int_0^t 
\mathbb{E} \left\langle  N^{1/2} L_{\ell}  X_{s} \left( \xi \right) , N^{1/2} L_{\ell}  X_{s} \left( \xi \right) \right\rangle
ds ,
 \end{aligned}
\end{equation}
where  $G \left( t \right)$, $H \left( t \right)$, $L_1$, $L_2$, $L_3$ is defined as in (\ref{eq:SSE}).

We write $\mathfrak{D}$ for the set of all  $x \in  \ell^2(\mathbb{Z}_+)\otimes \mathbb{C}^2$ 
such that 
$ \left\langle e_n \otimes e_{\eta}, x \right\rangle = 0$ 
for all combinations of  $n \in \mathbb{Z}_{+}$ and $\eta = \pm$ except a finite number. 
Then,
for all $x \in \mathfrak{D}$ we have
\begin{align*}
 & 2 \ \Re \left\langle N  x , G \left( t \right) x \right\rangle 
 + \sum_{\ell = 1}^{3} \left\langle  N^{1/2} L_{\ell} \, x , N^{1/2} L_{\ell}  \, x \right\rangle
 \\
 & = 
 \left\langle  x , 
 \left( \mathrm{i} \left[ H \left( t \right) ,  N \right] 
 + \sum_{\ell = 1}^{3} \left( \frac{1}{2} \left[ L_{\ell}^{\ast } ,  N \right] L_{\ell} 
 +  \frac{1}{2}  L_{\ell}^{\ast } \left[ N,  L_{\ell} \right] \right)
\right) x   \right\rangle 
 \\
 & = 
 \left\langle  x , 
 \left( -   \left[ 
     \alpha \left( t \right) a^\dagger -  \overline{\alpha \left( t \right)} a  ,
  N \right] 
 +  \kappa  \left[ a^\dagger ,  N \right] a
 +  \kappa \,  a^\dagger  \left[ N,  a \right] 
\right) x   \right\rangle ,
\end{align*}
and so 
\begin{equation}
\label{eq:8.30}
2 \ \Re \left\langle N  x , G \left( t \right) x \right\rangle 
 + \sum_{\ell = 1}^{3} \left\langle  N^{1/2} L_{\ell} \, x , N^{1/2} L_{\ell}  \, x \right\rangle
= 
 \left\langle  x , 
 \left(     \alpha \left( t \right) a^\dagger +  \overline{\alpha \left( t \right)} a 
 -2  \kappa   N  
\right) x   \right\rangle 
\end{equation}
since
$\left[ N, a^\dagger \right] = a^\dagger $
and
$\left[ a, N \right] = a $.
As $\mathfrak{D}$ is a core for $N$,
(\ref{eq:8.30}) holds for all $x \in \mathcal{D}\left( N \right) $,
and hence (\ref{eq:8.7}) gives 
\begin{equation*}
\Tr \left(N \rho_t \right)
\hspace{-0.8pt} = \hspace{-0.8pt}
\Tr \left(N \rho_0 \right)
+
\int_0^t  \left( 2  \, \Re \left(  \overline{\alpha \left( t \right)}  \mathbb{E} \left\langle  X_{s} \left( \xi \right) , a X_{s} \left( \xi \right) \right\rangle \right) 
- 2  \kappa \,   \mathbb{E} \left\langle  X_{s} \left( \xi \right) , N X_{s} \left( \xi \right) \right\rangle
\right) ds.
\end{equation*}
This, together with $\rho_ s =   \mathbb{E} \ketbra{ X_{s} \left( \xi \right)}{ X_{s} \left( \xi \right)}$, implies 
 \begin{equation*}
\Tr \left(N \rho_t \right)
 = 
\Tr \left(N \rho_0 \right)
+
\int_0^t  \left( 2 \, \Re \left(  \overline{\alpha \left( s \right)} \, \Tr \left(a \rho_s \right) \right) 
- 2  \kappa \, \Tr \left(N \rho_s \right)
\right) ds 
\end{equation*}
(see, e.g, \cite{MoraAP}).
The continuity of $\alpha \left( t \right)$, $\Tr \left(a \rho_t \right)$ and $\Tr \left(N \rho_t \right)$ yields (\ref{eq:8.32}).
\end{proof}

Let $d  \geq 0$. 
From, for instance Theorem 8 of \cite{FagMora2019},  we have that 
\begin{equation}
 \label{eq:LongTimeLE}
 \begin{aligned} 
 & \left\vert A \left( t  \right) \right\vert ^2 + g^2  \left\vert S \left( t  \right)  \right\vert ^2 / \left( \gamma \kappa \right)
+ g^2 \left( D \left( t \right) - d \right)^2  / \left(  4 \gamma \kappa \right)
\\
& \quad \leq
\exp \left( 
- t \min \left\{ \kappa - \frac{g^2d}{\gamma}   ,   \gamma - \frac{g^2d}{\kappa}   \right\} 
\right)
\left(
\left\vert A \left( 0  \right) \right\vert ^2 + \frac{g^2}{\gamma \kappa} \left\vert S \left( 0  \right)  \right\vert ^2
+ \frac{g^2}{4 \gamma \kappa} \left( D \left( 0 \right) - d \right)^2
\right) 
\end{aligned}
\end{equation}
for any $t \geq 0$.
Next,
we improve the upper bound of $\left\vert S \left( t  \right)  \right\vert $ and $ \left( D \left( t \right) - d \right)^2$
given by (\ref{eq:LongTimeLE}) in case $g \approx 0$ and $C_\mathfrak{b} < 1$.

\begin{lemma}
 \label{th:LorenzEquations-Laser}
 Assume that $S \left( t \right)$, $Z \left( t \right)$ and $D \left( t \right)$ is the solution of \eqref{eq:Lorenz} with
$\omega \in\mathbb{R}$,  $d\in \left[ 0 ,1 \right[$, $ g \in \mathbb{R} \smallsetminus \left\{ 0 \right\}$
 and  $\kappa,\gamma>0$.
 Let $C_\mathfrak{b} < 1$.
 Then for any $t \geq 0$ we have
 \begin{align*}
& \left\vert S \left( t  \right)  \right\vert ^2 + \left( D \left( t \right) - d \right)^2 /4
\\  
& \leq
\hbox{\rm e}^{-  \left( 1- C_\mathfrak{b} \right) \min \left\{ \kappa ,   \gamma \right\}  t}
\left(
\frac{4 \kappa  d}{\gamma} \left\vert A \left( 0  \right) \right\vert ^2 
+ \left( \frac{4 \kappa }{\gamma} + 1 \right) \left\vert S \left( 0  \right) \right\vert ^2
+  \left( \frac{ \kappa }{\gamma} + \frac{1}{4} \right) \left( D \left( 0 \right) - d \right)^2
\right) .
\end{align*}
\end{lemma}

\begin{proof}
Define 
$X \left( t \right) =  \hbox{\rm e}^{ i \omega t }  A \left( t \right)$, 
$ Y \left( t \right) =  \hbox{\rm e}^{ i \omega t }  S \left( t \right) $
and 
$ Z \left( t \right)  =  D \left( t \right) - d $
for  any $t \geq 0$.
By  (\ref{eq:Lorenz}), 
computing the derivatives of $ \left\vert Y \left( t \right) \right\vert^2$ and $ Z \left( t \right)^2$ gives 
\begin{equation}
\label{eq:L2}
4 \, \frac{d}{dt} \left\vert Y \left( t \right) \right\vert^2 + \frac{d}{dt}  Z \left( t \right)^2
=
 8 \, d \, g \, \Re \left( X \left( t \right)  \overline{ Y \left( t \right) } \right) 
- 8 \gamma \left\vert  Y \left( t  \right)  \right\vert ^2
- 4 \gamma Z \left( t \right)^2 
\end{equation}
(see, e.g., proof of Theorem 8 of   \cite{FagMora2019}).
Since
\[
2 \, \Re \left( X \left( t \right)  \overline{ Y \left( t \right) } \right) 
=
\frac{4 \, d \, g }{\gamma} \, \Re \left( X \left( t \right)  \overline{ \frac{\gamma}{2 \, d \, g } Y \left( t \right) } \right) 
\leq
\frac{2 \, d \, g }{\gamma} \left\vert X \left( t \right) \right\vert^2 +  \frac{\gamma}{2 \, d \, g } \left\vert Y \left( t \right) \right\vert^2,
\]
(\ref{eq:L2}) leads to
\begin{equation*}
\frac{d}{dt} \left( 4  \left\vert Y \left( t \right) \right\vert^2 +   Z \left( t \right)^2 \right)
\leq
 \frac{8 d^2  g^2 }{\gamma} \left\vert X \left( t \right) \right\vert^2 
- \frac{3}{2} \gamma \left( 4  \left\vert Y \left( t \right) \right\vert^2 +   Z \left( t \right)^2 \right) .
\end{equation*}
This implies
\[
4  \left\vert Y \left( t \right) \right\vert^2 +  Z \left( t \right)^2 
\leq 
\hbox{\rm e}^{ - \frac{3}{2} \gamma \, t } \left( 4  \left\vert Y \left( 0 \right) \right\vert^2 +  Z \left( 0 \right)^2 \right)
+
\frac{8 d^2  g^2 }{\gamma} \hbox{\rm e}^{ - \frac{3}{2} \gamma \, t } \int_{0}^t \hbox{\rm e}^{ \frac{3}{2} \gamma \, s } \left\vert X \left( s \right) \right\vert^2 ds ,
\]
and so
\begin{equation}
\label{eq:L3}
4  \left\vert Y \left( t \right) \right\vert^2 +  Z \left( t \right)^2 
\leq 
\hbox{\rm e}^{ - \frac{3}{2} \gamma \, t } \left( 4  \left\vert Y \left( 0 \right) \right\vert^2 +  Z \left( 0 \right)^2 \right)
+
8 \kappa \, \hbox{\rm e}^{ - \frac{3}{2} \gamma \, t } \int_{0}^t \hbox{\rm e}^{ \frac{3}{2} \gamma \, s } d \left\vert X \left( s \right) \right\vert^2 ds ,
\end{equation}
because $C_\mathfrak{b} < 1$.

Combining (\ref{eq:LongTimeLE}) with $g^2 \, d / \left( \kappa \, \gamma \right) < 1$ we deduce that
\begin{equation*}
d \left\vert X \left( t  \right) \right\vert ^2 
\leq
\hbox{\rm e}^{
-  \left( 1- C_\mathfrak{b} \right) \min \left\{ \kappa ,   \gamma \right\}  t}
\left(
d \left\vert A \left( 0  \right) \right\vert ^2 + \left\vert S \left( 0  \right)  \right\vert ^2
+ \frac{1}{4} \left( D \left( 0 \right) - d \right)^2
\right) .
\end{equation*}
Using (\ref{eq:L3}), together with  
$ 
3 \gamma / 2- \left( 1- C_\mathfrak{b} \right) \min \left\{ \kappa ,   \gamma \right\}  \geq \gamma / 2
$,
yields
\begin{align*}
  4  \left\vert Y \left( t \right) \right\vert^2 +  Z \left( t \right)^2 
 & \leq 
\hbox{\rm e}^{ - \frac{3}{2} \gamma \, t } \left( 4  \left\vert Y \left( 0 \right) \right\vert^2 +  Z \left( 0 \right)^2 \right)
\\
 &  +
\frac{16 \kappa}{\gamma} 
\left(
d \left\vert A \left( 0  \right) \right\vert ^2 + \left\vert S \left( 0  \right)  \right\vert ^2
+ \frac{1}{4} \left( D \left( 0 \right) - d \right)^2
\right)
\left( 
\hbox{\rm e}^{-  \left( 1- C_\mathfrak{b} \right) \min \left\{ \kappa ,   \gamma \right\}  t}
- \hbox{\rm e}^{ - \frac{3}{2} \gamma \, t }
\right) ,
\end{align*}
and the lemma follows.
 \end{proof}

\begin{proof}[Proof of Theorem \ref{th:LongTime}]
First,
we shift the analysis from the non-linear quantum master equation (\ref{eq:Laser1})
to the linear  quantum master equation (\ref{eq:AuxiliarGKSL}).
To this end,
we consider the solution $ A \left( t \right) $, $ S \left( t \right) $ and $ D \left( t \right) $ to (\ref{eq:Lorenz})
with initial datum $  A \left( 0 \right) =  \Tr\left( a \, \rho_0 \right) $, $ S \left( 0 \right) = \Tr\left( \sigma^{-}  \, \rho_0 \right)$ 
and $ D \left( 0 \right) = \Tr\left(  \sigma^{3} \, \rho_0 \right) $.
Since $\rho_0 \in \mathfrak{L}_{1,N}^{+} \left( \ell^2 \left(\mathbb{Z}_+ \right) \otimes \mathbb{C}^2 \right) $, 
the functions 
$ t \mapsto \Tr\left( a \, \rho_{t} \right)$, $  t \mapsto  \Tr\left( \sigma^{-} \rho_{t}  \right)$ 
and $  t \mapsto  \Tr\left(  \sigma^{3}  \rho_{t} \right)$
satisfy (\ref{eq:Lorenz}) (see, e.g., \cite{FagMora2019}),
and so 
the uniqueness of the solution to (\ref{eq:Lorenz})
implies  
$ 
\left( A \left( t \right) ,  S \left( t  \right), D \left( t  \right) \right)
=
\left( \Tr\left( a \, \rho_{t} \right),  \Tr\left( \sigma^{-} \rho_{t}  \right) ,  \Tr\left(  \sigma^{3}  \rho_{t} \right) \right)
$
for all $t \geq 0$.
%$ A \left( t \right) = \Tr\left( a \, \rho_{t} \right) $,
%$ S \left( t  \right)  = \Tr\left( \sigma^{-} \rho_{t}  \right) $
%and
%$ D \left( t  \right) = \Tr\left(  \sigma^{3}  \rho_{t} \right) $.
By  the uniqueness of the $N$-weak solution to (\ref{eq:AuxiliarGKSL}),
the  $N$-weak solution $\rho_t$ to (\ref{eq:Laser1})
is equal to the $N$-weak solution to 
the non-homogeneous  linear evolution equation (\ref{eq:AuxiliarGKSL}) with initial condition 
$\rho_0$ and coefficients  $\alpha \left( t \right) = g \,  S\left( t  \right)$ and $\beta \left( t \right) = g \,  A\left( t  \right)$.

Now, 
we apply Corollary \ref{cor:Ineq-Conv}  with 
$\alpha \left( t \right) = g \,  S\left( t  \right)$ and $\beta \left( t \right) = g \,  A\left( t  \right)$,
together with 
$\alpha_0 = \beta_0 = 0$
since 
$
\Tr \left(   \varrho_{\infty} \, a \right) 
=
\Tr \left(   \varrho_{\infty} \, \sigma^{-}  \right) 
= 0
$.
This gives
\begin{equation}
\label{eq:30}
 \begin{aligned}
 \Tr \left( \left\vert  \rho_t  - \varrho_{\infty} \right\vert  \right)
& \leq 
\Tr \left( \left\vert   \rho^h_{t-s} \left( \rho_s  \right)   - \varrho_{\infty} \right\vert  \right)
+
4 \left\vert g \right\vert  \int_s^t \left\vert S \left( u \right) \right\vert \sqrt{ \Tr \left( \rho_u \, N  \right) +1 } \, du
\\
& \quad
+ 
2 \left\vert g \right\vert  \left( \left\Vert  \sigma^{-} \right\Vert +   \left\Vert  \sigma^{+} \right\Vert \right)
 \int_s^t \left\vert A \left( u \right) \right\vert du 
 \end{aligned}
\end{equation}
for all  $t  \geq s \geq 0$,
because 
$ \mathcal{E} \left( 0 \right) = e_0 $
and 
$
\begin{pmatrix}
 \frac{1}{2} +  \frac{d }{2}
 &
 0
 \\
 0
 &
  \frac{1}{2} -  \frac{d }{ 2}
\end{pmatrix}
=
\left( \frac{d+1}{2}  \ketbra{e_+}{e_+}  + \frac{1-d}{2}   \ketbra{e_-}{e_-}  \right)
$.
Here, 
$
\varrho_{\infty}
$
is defined by (\ref{eq:I10}) and $\left( \rho^h_{u} \left( \rho_s  \right) \right)_{ u \geq 0}$ 
is  the $N$-weak solution of (\ref{eq:3.2}) with initial datum $\rho_s$.
Combining (\ref{eq:30}) with Corollary \ref{cor:ConvEMCSimple}  yields
\begin{align}
 \label{eq:8.38}
\Tr \left( \left\vert  \rho_t  - \varrho_{\infty} \right\vert  \right)
 & \leq 
12 \, \hbox{\rm e}^{- \gamma  \left(t-s \right) }    \left( 1 +  \left\vert d \right\vert \right)
+
4 \, \hbox{\rm e}^{- \kappa  \left(t-s \right) }  \sqrt{ \Tr \left( \rho_{s}  \, N  \right) }
\\
\nonumber
&  \quad +
4 \left\vert g \right\vert  \int_s^t \left\vert S \left( u \right) \right\vert \sqrt{ \Tr \left( \rho_u  \, N  \right) +1 } \, du
+ 
2 \left\vert g \right\vert  \left( \left\Vert  \sigma^{-} \right\Vert +   \left\Vert  \sigma^{+} \right\Vert \right)
\int_s^t \left\vert A \left( u \right) \right\vert du .
\end{align}

Next,
we estimate the right-hand of (\ref{eq:8.38}).
Applying Lemma \ref{lem:EvolNumero} we obtain
\begin{align}
\nonumber
 \Tr\left( \rho_{t} \, N  \right) 
& =
\hbox{\rm e}^{- 2  \kappa \, t} \Tr\left( \rho_{0} \,  N  \right) 
+
2 g \, \int_0^t \hbox{\rm e}^{- 2  \kappa \left( t - s \right)}  \Re \left(  \overline{ S \left( t \right)}  \, \Tr \left( \rho_s \, a  \right) \right) ds
\\
\nonumber
& =
\hbox{\rm e}^{- 2  \kappa \, t} \Tr\left( \rho_{0} \,  N  \right) 
+
2 g   \int_0^t \hbox{\rm e}^{- 2  \kappa \left( t - s \right)}  \Re \left(  \overline{ S \left( s \right)}  A \left( s \right)  \right) ds 
\\
\label{eq:8.31}
&
\leq
\hbox{\rm e}^{- 2  \kappa \, t} \Tr\left( \rho_{0} \, N  \right) 
+
\left\vert g \right\vert \int_0^t \hbox{\rm e}^{- 2  \kappa \left( t - s \right)} 
\left(  \left\vert S \left( s \right) \right\vert^2 + \left\vert A \left( s \right) \right\vert^2   \right) 
ds .
\end{align}
If $d < 0$,
then for all  $t \geq 0$  we have 
\[
\left\vert d \right\vert \left\vert A \left( t  \right) \right\vert ^2 
+ \left\vert S \left( t  \right)  \right\vert ^2
+ \left( D \left( t \right) - d \right)^2 / 4
\leq
\hbox{\rm e}^{
-  2 t\, \min \left\{  \kappa  ,  \gamma \right\} 
}
\left(
\left\vert d \right\vert \left\vert A \left( 0  \right) \right\vert ^2 +  \left\vert S \left( 0  \right)  \right\vert ^2
+  \left( D \left( 0 \right) - d \right)^2 / 4 \right)
\]
(see, e.g., \cite{FagMora2019}).
Using this inequality, (\ref{eq:LongTimeLE}) and Lemma \ref{th:LorenzEquations-Laser},
together with $ d g^2 / \left( \gamma \kappa \right) < 1$,
we deduce that  
\begin{equation}
\label{eq:8.33}
 \left\vert A \left( t  \right) \right\vert^2 \leq K_A  \exp \left( - c_{sys} \, t \right)
 \hspace{1cm} \text{and}  \hspace{1cm}
\left\vert S \left( t  \right)  \right\vert^2 \leq K_S  \exp \left( - c_{sys} \, t \right) 
\qquad \forall t \geq 0 ,
\end{equation}
where:
\begin{itemize}
 
 \item In case $d < 0$, 
$
c_{sys}
=
2 \, \min \left\{  \kappa  ,  \gamma \right\}
$,
$
K_S
=
\left\vert d \right\vert \left\vert A \left( 0  \right) \right\vert ^2 +  \left\vert S \left( 0  \right)  \right\vert ^2
+  \left( D \left( 0 \right) - d \right)^2/4 
$,
and
$
K_A
=
\left\vert A \left( 0  \right) \right\vert ^2 +  \left\vert S \left( 0  \right)  \right\vert ^2 / \left\vert d \right\vert 
+  \left( D \left( 0 \right) - d \right)^2 / \left( 4 \left\vert d \right\vert \right)
$.
 
 \item In case $d \geq 0 $, 
$
c_{sys}
=
 \left( 1- C_\mathfrak{b} \right) \min \left\{ \kappa ,   \gamma \right\}
$,
$
K_A
=
\left\vert A \left( 0  \right) \right\vert ^2 + \frac{g^2}{\gamma \kappa} \left\vert S \left( 0  \right)  \right\vert ^2
+ \frac{g^2}{4 \gamma \kappa} \left( D \left( 0 \right) - d \right)^2 
$
and
$
K_S
=
\frac{4 \kappa  d}{\gamma} \left\vert A \left( 0  \right) \right\vert ^2 
+ \left( \frac{4 \kappa }{\gamma} + 1 \right) \left\vert S \left( 0  \right) \right\vert ^2
+  \left( \frac{ \kappa }{\gamma} + \frac{1}{4} \right) \left( D \left( 0 \right) - d \right)^2  
$.

\end{itemize}

Suppose that either $d \geq 0$ or $d < 0$ with $\kappa >  \gamma$.
Then $2 \kappa > c_{sys}$ and 
\[
 \int_0^t   \hbox{\rm e}^{- 2  \kappa \left( t - u \right)} 
\left( \left\vert S \left( u \right) \right\vert^2 + \left\vert A \left( u \right) \right\vert^2   \right) du
<
\frac{K_A + K_S }{2 \kappa - c_{sys}} \left(  \hbox{\rm e}^{-  c_{sys} \, t } -  \hbox{\rm e}^{- 2  \kappa \, t } \right) 
<
\frac{K_A + K_S }{2 \kappa - c_{sys}}  \hbox{\rm e}^{-  c_{sys} \, t } .
\]
From (\ref{eq:8.31}) it follows that
\begin{equation}
 \label{eq:8.34}
 \Tr\left( \rho_{t} \, N  \right) 
\leq 
\left(\Tr\left( \rho_{0} \, N  \right) +  \frac{ \left\vert g \right\vert \left( K_A + K_S \right)}{ 2 \kappa - c_{sys} }
\right) 
\hbox{\rm e}^{-  c_{sys} \, t } .
\end{equation}
Consider $t  \geq s \geq 0$.
Applying 
(\ref{eq:8.34}) we get
\[
 \Tr\left( \rho_{t} \, N  \right) 
\leq 
\Tr\left( \rho_{0} \, N  \right) +  \left\vert g \right\vert \left( K_A + K_S \right) / \left( 2 \kappa - c_{sys} \right) ,
\]
and hence (\ref{eq:8.33}) gives
\begin{align*}
 \int_s^t \left\vert S \left( u \right) \right\vert \sqrt{ \Tr \left( \rho_u \, N  \right) +1 } \, du
 \leq
\left( 1 + \Tr\left( \rho_{0} \, N  \right) + \frac{ \left\vert g \right\vert  \left( K_A + K_S \right) }{ 2 \kappa - c_{sys} } \right)^{1/2}
\int_s^t \left\vert S \left( u \right) \right\vert  du
\\
\leq
\left(
2 \sqrt{K_S} \left( 1 + \Tr\left( \rho_{0} \, N  \right) + \frac{ \left\vert g \right\vert  \left( K_A + K_S \right) }{ 2 \kappa - c_{sys} } \right)^{1/2}
\right)
\frac{\hbox{\rm e}^{- c_{sys} \, s /2 } - \hbox{\rm e}^{- c_{sys} \, t / 2}}{c_{sys}}  .
\end{align*}
Using (\ref{eq:8.33})  we also obtain 
\begin{equation}
\label{eq:8.61}
\left( \left\Vert  \sigma^{-} \right\Vert +   \left\Vert  \sigma^{+} \right\Vert \right)
 \int_s^t \left\vert A \left( u \right) \right\vert du 
\leq
2  \int_s^t \left\vert A \left( u \right) \right\vert du 
\leq
4  \sqrt{K_A} \frac{\hbox{\rm e}^{- c_{sys} \, s /2 } - \hbox{\rm e}^{- c_{sys} \, t / 2}}{c_{sys}} .
\end{equation}
Then,
from  (\ref{eq:8.38}) and (\ref{eq:8.34})  we get
\begin{align*}
 \Tr \left( \left\vert  \rho_t  - \varrho_{\infty} \right\vert  \right)
 & \leq 
12 \, \hbox{\rm e}^{- \gamma  \left(t-s \right) }    \left( 1 +  \left\vert d \right\vert \right)
+
4 \, \hbox{\rm e}^{- \kappa  \left(t-s \right) - \frac{c_{sys}}{2} s}  
\sqrt{ \Tr\left( \rho_{0} \, N  \right) +  \frac{ \left\vert g \right\vert  \left( K_A + K_S \right)}{ 2 \kappa - c_{sys} } }
\\
&  \quad +
\frac{8 \left\vert g \right\vert }{c_{sys}}
\left(
 \sqrt{K_S} \left( 1 + \Tr\left( \rho_{0} \, N  \right) + \frac{ \left\vert g \right\vert  \left( K_A + K_S \right) }{ 2 \kappa - c_{sys} } \right)^{1/2}
 +
  \sqrt{K_A} 
\right)
\hbox{\rm e}^{- \frac{ c_{sys} }{2 } s  } .
\end{align*}

In case $d \geq 0$, taking $t = 3s/2$ yields
\begin{align*}
 \Tr \left( \left\vert  \rho_{3 s / 2}  - \varrho_{\infty} \right\vert  \right)
 & \leq 
12 \, \hbox{\rm e}^{- \gamma  s / 2}    \left( 1 +  \left\vert d \right\vert \right)
+
4 \, \hbox{\rm e}^{- \frac{c_{sys}}{2} s}  
\sqrt{ \Tr\left( \rho_{0} \, N  \right) +  \frac{ \left\vert g \right\vert  \left( K_A + K_S \right)}{ 2 \kappa - c_{sys} } }
\\
&  \hspace{-10pt} +
\frac{8 \left\vert g \right\vert }{c_{sys}}
\left(
 \sqrt{K_S} \left( 1 + \Tr\left( \rho_{0} \, N  \right) + \frac{ \left\vert g \right\vert  \left( K_A + K_S \right) }{ 2 \kappa - c_{sys} } \right)^{1/2}
 +
  \sqrt{K_A} 
\right)
\hbox{\rm e}^{- \frac{c_{sys} }{ 2} s  } , 
\end{align*}
and so for all $t \geq 0$,
\begin{align*}
 \Tr \left( \left\vert  \rho_{ t}  - \varrho_{\infty} \right\vert  \right)
 & \leq 
 \hbox{\rm e}^{- \frac{c_{sys} }{ 3} t  }
 \left( 
 12 \left( 1 +  \left\vert d \right\vert \right) 
+
4 
\sqrt{ \Tr\left( \rho_{0} \, N  \right) +  \frac{ \left\vert g \right\vert  \left( K_A + K_S \right)}{ 2 \kappa - c_{sys} } }
+
\frac{8 \left\vert g \right\vert }{c_{sys}} \sqrt{K_A} 
\right.
\\
& \hspace{80pt}
\left. 
 +
\frac{8 \left\vert g \right\vert }{c_{sys}}
 \sqrt{K_S} \left( 1 + \Tr\left( \rho_{0} \, N  \right) + \frac{ \left\vert g \right\vert  \left( K_A + K_S \right) }{ 2 \kappa - c_{sys} } \right)^{1/2}
\right).
\end{align*}

In case $d < 0$ with $\kappa >  \gamma$, choosing  $t = 2 s$ we deduce that 
\begin{align*}
 \Tr \left( \left\vert  \rho_{2 s}  - \varrho_{\infty} \right\vert  \right)
 & \leq 
12 \, \hbox{\rm e}^{- \gamma  s }    \left( 1 +  \left\vert d \right\vert \right)
+
4 \, \hbox{\rm e}^{- \kappa  s - \frac{c_{sys}}{2} s}  
\sqrt{ \Tr\left( \rho_{0} \, N  \right) +  \frac{ \left\vert g \right\vert  \left( K_A + K_S \right)}{ 2 \kappa - c_{sys} } }
\\
&  \quad +
\frac{8 \left\vert g \right\vert }{c_{sys}}
\left(
 \sqrt{K_S} \left( 1 + \Tr\left( \rho_{0} \, N  \right) + \frac{ \left\vert g \right\vert  \left( K_A + K_S \right) }{ 2 \kappa - c_{sys} } \right)^{1/2}
 +
  \sqrt{K_A} 
\right)
\hbox{\rm e}^{- \frac{ c_{sys} }{2 } s  } ,
\end{align*}
and consequently 
\begin{align*}
 \Tr \left( \left\vert  \rho_{ t}  - \varrho_{\infty} \right\vert  \right)
 & \leq 
 \hbox{\rm e}^{- \frac{c_{sys} }{ 4} t  }
 \left( 
 12 \left( 1 +  \left\vert d \right\vert \right) 
+
4 
\sqrt{ \Tr\left( \rho_{0} \, N  \right) +  \frac{ \left\vert g \right\vert  \left( K_A + K_S \right)}{ 2 \kappa - c_{sys} } }
+
\frac{8 \left\vert g \right\vert }{c_{sys}} \sqrt{K_A} 
\right.
\\
& \hspace{3cm}
\left. 
 +
\frac{8 \left\vert g \right\vert }{c_{sys}}
 \sqrt{K_S} \left( 1 + \Tr\left( \rho_{0} \, N  \right) + \frac{ \left\vert g \right\vert  \left( K_A + K_S \right) }{ 2 \kappa - c_{sys} } \right)^{1/2}
\right).
\end{align*}
for any $t \geq 0$.

On the other hand, we assume that  $d < 0$ and $\kappa \leq  \gamma$.
Then
\[
 \int_0^t  \hbox{\rm e}^{- 2  \kappa \left( t - u \right)} 
\left( \left\vert S \left( u \right) \right\vert^2 + \left\vert A \left( u \right) \right\vert^2   \right) du
\leq
2 t \,  \left( K_A + K_S \right) \exp \left(- 2  \kappa \, t \right)  ,
\]
and so  (\ref{eq:8.31}) leads to 
\begin{equation}
\label{eq:8.39}
 \Tr\left( \rho_{t} \, N  \right) 
\leq 
\exp \left(- 2  \kappa \, t \right)  \Tr\left( \rho_{0} \, N  \right) 
+
2 \left\vert g \right\vert \left( K_A + K_S \right)  t \, \exp \left(- 2  \kappa \, t \right)  .
\end{equation}
Since 
$
t \, \exp \left(- 2  \kappa \, t \right) \leq 1 / \left( 2 \, \hbox{\rm e}  \, \kappa  \right)
$,
according to (\ref{eq:8.33}) we have that for all $t  \geq s \geq 0$,
\[
\int_s^t \left\vert S \left( u \right) \right\vert \sqrt{ \Tr \left( \rho_u \, N  \right) +1 } \, du
 \leq
\sqrt{K_S} \sqrt{ \Tr\left( \rho_{0} \, N  \right)  + \frac{ \left\vert g \right\vert \left( K_A + K_S \right) }{ \kappa \, \hbox{\rm e} }}
  \frac{\hbox{\rm e}^{- \kappa \, s  } - \hbox{\rm e}^{- \kappa \, t }}{\kappa} .
\]
Moreover, (\ref{eq:8.61}) gives
\[
\left( \left\Vert  \sigma^{-} \right\Vert +   \left\Vert  \sigma^{+} \right\Vert \right)
 \int_s^t \left\vert A \left( u \right) \right\vert du 
\leq
2  \sqrt{K_A} \frac{\hbox{\rm e}^{- \kappa \, s  } - \hbox{\rm e}^{- \kappa \, t }}{\kappa}  .
\]
Therefore, (\ref{eq:8.38}) yields
\begin{align*}
 \Tr \left( \left\vert  \rho_t  - \varrho_{\infty} \right\vert  \right)
 & \leq 
12 \, \hbox{\rm e}^{- \gamma  \left( t-s \right) }    \left( 1 +  \left\vert d \right\vert \right)
+
4 \, \hbox{\rm e}^{- \kappa  \left(  t-s \right) }  
\sqrt{ \Tr\left( \rho_{0} \, N  \right) 
+
 \left\vert g \right\vert  \left( K_A + K_S \right)  /  \left(  \kappa \, \hbox{\rm e} \right)  }
\\
& \quad +
4 \left\vert g \right\vert 
\left(
\sqrt{K_S} \sqrt{ \Tr\left( \rho_{0} \, N  \right)  + \left\vert g \right\vert \left( K_A + K_S \right) /  \left(  \kappa \hbox{\rm e} \right) }
+
  \sqrt{K_A}
\right)
\hbox{\rm e}^{- \kappa \, s  } / \kappa.
\end{align*}
Hence
\begin{align*}
 \Tr \left( \left\vert  \rho_{2 s}  - \varrho_{\infty} \right\vert  \right)
 & \leq 
12 \, \hbox{\rm e}^{- \gamma  s }    \left( 1 +  \left\vert d \right\vert \right)
+
4 \, \hbox{\rm e}^{- \kappa  s }  
\sqrt{ \Tr\left( \rho_{0} \, N  \right) 
+
 \left\vert g \right\vert  \left( K_A + K_S \right)  /  \left(  \kappa \hbox{\rm e} \right)  }
\\
& \quad +
4 \left\vert g \right\vert 
\left(
\sqrt{K_S} \sqrt{ \Tr\left( \rho_{0} \, N  \right)  + \left\vert g \right\vert \left( K_A + K_S \right) /  \left(  \kappa \hbox{\rm e} \right) }
+
  \sqrt{K_A}
\right)
\hbox{\rm e}^{- \kappa \, s  } / \kappa ,
\end{align*}
which implies 
\begin{align*}
 \Tr \left( \left\vert  \rho_t  - \varrho_{\infty} \right\vert  \right)
 & \leq 
4 \hbox{\rm e}^{- \frac{c_{sys}}{4} t  }
 \left(
 \left( 1 + \frac{\left\vert g \right\vert \sqrt{K_S} }{\kappa} \right)
 \sqrt{ \Tr\left( \rho_{0} \, N  \right)  + \frac{ \left\vert g \right\vert \left( K_A + K_S \right) }{ \kappa \, \hbox{\rm e} }}
\right. 
\\
& \hspace{6.8cm}
\left.
+ 3   \left( 1 +  \left\vert d \right\vert \right)
+
\frac{ \left\vert g \right\vert \sqrt{K_A} }{ \kappa}
\right) .
\end{align*}
This completes the proof of (\ref{eq:LTB.1}).

We are now in position to show (\ref{eq:8.35}).
We decompose $A$ as 
\[
A = A \, P + P  A  \left( I - P \right) + \left( I - P \right)  A  \left( I - P \right) ,
\]
where $P$ is the orthogonal projection of $\ell^2 \left(\mathbb{Z}_+ \right) \otimes \mathbb{C}^2 $ onto the linear span of $e_0 \otimes e_+$ and $e_0 \otimes e_-$,
i.e.,
$
P  \, x 
=
\langle e_0 \otimes e_+, x \rangle \, e_0 \otimes e_+ 
+
\langle e_0 \otimes e_-, x \rangle  \, e_0 \otimes e_- 
$.
From (\ref{eq:I10}) it follows 
\begin{align*}
\quad
 \Tr \left( \varrho_{\infty}  \, A \, P \right)  
& =
\sum_{n = 0}^{+ \infty} \sum_{\eta = \pm}
\langle e_n \otimes e_{\eta}, \varrho_{\infty}  \, A \, P   e_n \otimes e_{\eta} \rangle 
=
\langle \varrho_{\infty} e_0 \otimes e_{\eta},   A  \,  e_0 \otimes e_{\eta} \rangle
\\
&
 =
 \frac{d+1}{2} \langle e_0 \otimes e_+, A \, e_0 \otimes e_+ \rangle
+
\frac{1-d}{2} \langle e_0 \otimes e_-, A \, e_0 \otimes e_- \rangle .
\end{align*}
We can extend $P  A  \left( I - P \right) $ to the  bounded linear operator 
\[
P  A  \left( I - P \right)  x 
=
\langle A^{\star} e_0 \otimes e_+,  \left( I - P \right)  x \rangle \, e_0 \otimes e_+
+
\langle A^{\star} e_0 \otimes e_-,  \left( I - P \right)  x \rangle \, e_0 \otimes e_- .
\]
Using (\ref{eq:I10}) yields
$
\Tr \left( \varrho_{\infty}  P  A  \left( I - P \right) \right)  
=
\sum_{n = 0}^{+ \infty} \sum_{\eta = \pm}
\langle  \varrho_{\infty} e_n \otimes e_{\eta},  P  A  \left( I - P \right)   e_n \otimes e_{\eta} \rangle 
=
0 
$.
Applying  (\ref{eq:LTB.1}) we deduce that for all $t \geq 0 $,
 \begin{align}
 \nonumber
& 
\left\vert 
\Tr \left( \rho_t \left( A \, P + P  A  \left( I - P \right) \right) \right)
-
\frac{d+1}{2} \langle e_0 \otimes e_+, A \, e_0 \otimes e_+ \rangle
-
\frac{1-d}{2} \langle e_0 \otimes e_-, A \, e_0 \otimes e_- \rangle
\right\vert
\\
 \label{eq:8.36}
& =
\left\vert 
\Tr \left( \rho_t \left( A \, P + P  A  \left( I - P \right) \right) \right)
-
\Tr \left( \varrho_{\infty} \left( A \, P + P  A  \left( I - P \right) \right) \right)
\right\vert
\\
\nonumber
& \leq 
\left( \left\Vert  A \, P \right\Vert + \left\Vert P  A  \left( I - P \right) \right\Vert
\right)
K_{sys} \left( \left\vert g \right\vert \right) \exp \left( - \delta_{sys} \, t \right) .
\end{align}

According to \eqref{eq:LTB.c} we have 
\[
\max \left\{\left\Vert   \frac{1}{2} \left( A + A^{\star} \right)  x  \right\Vert, \left\Vert  \frac{i}{2} \left( A^{\star} - A  \right)  x  \right\Vert \right\}
\leq 
\widetilde{K}   \left\Vert  x  \right\Vert_{N}
\qquad
\forall x \in \mathcal{D}\left( N \right) .
\]
Hence, for any $ x \in \mathcal{D}\left( N \right)$,
$
\left\vert
 \langle  x ,   A  \, x  \rangle 
 \right\vert
 =
 \left\vert
 \langle  x ,   \frac{1}{2} \left( A + A^{\star} \right) x  \rangle + i \langle  x ,   \frac{i}{2} \left( A^{\star} - A  \right) x  \rangle
 \right\vert
 \leq
4 \widetilde{K}   \left\Vert  x  \right\Vert_{N} 
$
(see, e.g., proof of Theorem VI.1.38 of \cite{Kato}).
Therefore,
\begin{align*}
\left\vert
 \langle   \left( I - P \right)  x  ,   \left( I - P \right)  A  \left( I - P \right)  x  \rangle 
 \right\vert
&  \leq
4 \widetilde{K}   \left(  \left\Vert  \left( I - P \right)  x   \right\Vert^2 +  \langle  \left( I - P \right)  x  ,  N  \left( I - P \right)  x  \rangle \right)
\\
&  \leq
8 \widetilde{K}    \, \langle  \left( I - P \right)  x  ,  N  \left( I - P \right)  x  \rangle 
=
8 \widetilde{K}    \, \langle  x  ,  N  x  \rangle 
\end{align*}
for any $x \in  \mathcal{D}\left( N \right) $,
and so  
\begin{equation*}
\left\vert \Tr \left( \rho_t  \left( I - P \right)  A  \left( I - P \right) \right) \right\vert
\leq
8 \widetilde{K}    \, \Tr \left( \rho_t N \right)
\hspace{3cm} \forall t \geq 0 .
\end{equation*}
Then, 
using (\ref{eq:8.34}), (\ref{eq:8.39}) and (\ref{eq:8.36}) we obtain (\ref{eq:8.35}),
because
$ \left\Vert A P \right\Vert \leq \widetilde{K}$ 
and 
$ \left\Vert P A \left( I - P \right) \right\Vert 
=
\left\Vert  \left( I - P \right) A^* P  \right\Vert 
\leq 
\widetilde{K} 
$.
\end{proof}

\subsection{Proof of Theorem \ref{th:LimitCycle}}
\label{sec:Proof:LimitCycle}

For the sake of completeness, 
we now study the local stability of the nonzero equilibrium points of \eqref{eq:Lorenz} with $\omega = 0$.
As in the physical literature (see, e.g., \cite{Khanin2006,NingHaken1990}),
we combine  a change of  variables with linear stability analysis.

\begin{lemma}
\label{lem:LorenzBifurcation}
Let $C_\mathfrak{b} > 1$.
Consider \eqref{eq:Lorenz} with $\omega = 0$,
$d\in \left]-1,1 \right[$,  $ g \in \mathbb{R} \smallsetminus \left\{ 0 \right\}$
and  $\kappa,\gamma>0$.
Suppose that $\kappa \leq 3 \gamma $ 
or that $\kappa > 3 \gamma $ and $ \kappa^2 +  5 \kappa \gamma >  \gamma \left( \kappa - 3 \gamma \right) C_\mathfrak{b}$. 
Then,
there exist constants $\epsilon, \lambda, K > 0$ such that for all $t \geq 0$,
\[
\left\vert 
A \left( t \right)  
-
\frac{\gamma \sqrt{ C_\mathfrak{b} - 1} }{\sqrt{2} \left\vert g \right\vert} 
\frac{ A  \left( 0 \right) }{  \left\vert  A  \left( 0 \right)  \right\vert } 
\exp \left(  \mathrm{i}  g \int_0^{+ \infty} \Im \left( \frac{S  \left( s \right) }{  A \left( s \right) } \right) ds \right) 
\right\vert 
\leq K \exp \left(  - \lambda t \right) ,
\]
\[
\left\vert 
S\left( t \right)  
-
\frac{\kappa}{g}
\frac{\gamma \sqrt{ C_\mathfrak{b} - 1} }{\sqrt{2} \left\vert g \right\vert} 
\frac{ A  \left( 0 \right) }{  \left\vert  A  \left( 0 \right)  \right\vert } 
\exp \left(  \mathrm{i}  g \int_0^{+ \infty} \Im \left( \frac{S  \left( s \right) }{  A \left( s \right) } \right) ds \right) 
\right\vert 
\leq K \exp \left(  - \lambda t \right)
\]
and
$
\left\vert 
D \left( t \right)  
-
d / C_\mathfrak{b}
\right\vert 
\leq K \exp \left(  - \lambda t \right)
$
provided that 
\[
\max \left\{ 
  \left\vert    \left\vert A \left( 0 \right)  \right\vert  -  \frac{\gamma \sqrt{ C_\mathfrak{b} - 1} }{\sqrt{2} \left\vert g \right\vert}  \right\vert  ,  
 \left\vert  S \left(  0 \right)   -   \frac{ \kappa }{ g } \frac{\gamma \sqrt{ C_\mathfrak{b} - 1} }{\sqrt{2} \left\vert g \right\vert}  \frac{A \left( 0 \right) }{ \left\vert A \left( 0 \right)  \right\vert }\right\vert,
  \left\vert D \left( 0\right) -  \frac{d}{C_\mathfrak{b}} \right\vert 
\right\}
< \epsilon .
\]
\end{lemma}

\begin{proof}
Consider the change of functions 
$ 
A \left( t \right) 
 =
 r \left( t \right) \hbox{\rm e}^{\mathrm{i}  \, \phi \left( t \right) }
$,
$
S \left( t \right) 
=
 \left( S_R  \left( t \right)  + \mathrm{i}  \,S_I  \left( t \right)  \right) \hbox{\rm e}^{\mathrm{i}  \, \phi \left( t \right) }
 $
 and 
$
D \left( t \right) 
=
D_R \left( t \right) + d
$,
where
the unknown   $A \left( t \right)$, $S \left( t \right)$, $D \left( t \right)$ are replaced by 
the real functions 
$r \left( t \right) $, $\phi \left( t \right)$, $S_R  \left( t \right)$, $S_I  \left( t \right)$, $D_R \left( t \right)$.
Let  $A \left( 0 \right) \neq 0 $. Then,
from \eqref{eq:Lorenz}  it follows that 
 \begin{equation}
 \label{eq:L4}
 \left\{
 \begin{aligned}
 r ^{\prime} \left( t \right) 
 & =
 - \kappa \, r \left( t \right)  + g \, S_R  \left( t \right)
\\
S_R ^{\prime} \left( t \right)
 & =
 - \gamma S_R  \left( t \right) + g \, r \left( t \right) \left(  D_R \left( t \right) + d \right)
 + g \, S_I  \left( t \right)^2 / r \left( t \right)
 \\
S_I ^{\prime} \left( t \right)
 & =
 - \gamma S_I  \left( t \right) - g \, S_I  \left( t \right) S_R  \left( t \right) / r \left( t \right)
  \\
D_R ^{\prime} \left( t \right)
 & =
 -  2 \gamma D_R  \left( t \right) - 4g \, r \left( t \right) S_R  \left( t \right) 
 \end{aligned}
 \right. .
\end{equation}
and
\begin{equation}
\label{eq:L8}
\phi ^{\prime} \left( t \right)
 =
g \, S_I  \left( t \right)   / r \left( t \right) .
\end{equation}

Since $C_\mathfrak{b} > 1$,
\eqref{eq:L4} has the fix point:
$
r = r_0
$,
$
S_R =   \kappa \, r_0 / g
$,
$
S_I = 0
$
and
$
D_R = d \left( 1 / C_\mathfrak{b} - 1 \right)
$,
where
$
r _0 = \frac{\gamma}{\sqrt{2} \left\vert g \right\vert} \sqrt{ C_\mathfrak{b} - 1}
$.
The Jacobian matrix of the function describing the right-hand side of \eqref{eq:L4} evaluated at this fix point is equal to
\[
J 
=
\begin{pmatrix} 
- \kappa                  & g               & 0                             & 0
\\
\kappa \gamma /g  &  - \gamma & 0                             & g \, r_0
\\
0                             & 0               & -\gamma - \kappa   & 0
\\
- 4 \kappa \, r_0      & -4 g \, r_0  & 0                             & -2 \gamma 
\end{pmatrix} .
\]
As
$
\det \left( J - \lambda \, I \right)
=
\left( \lambda + \gamma + \kappa \right)
\left( 
\lambda^3 + \left(  3 \gamma + \kappa \right) \lambda^2
+ 
\left( 2 \gamma^2 C_\mathfrak{b} + 2 \gamma \kappa \right) \lambda
+
8 g^2 \kappa r_0^2
\right) 
$,
the eigenvalues of J are $- \gamma  - \kappa $, which is less than $0$,
and the zeros of the polynomial 
\begin{equation}
\label{eq:L5}
\lambda^3 + \left(  3 \gamma + \kappa \right) \lambda^2
+ 
\left( 2 \gamma^2 C_\mathfrak{b} + 2 \gamma \kappa \right) \lambda
+
8 g^2 \kappa r_0^2 .
\end{equation}
Since the coefficients of \eqref{eq:L5} are positive,
the real roots of \eqref{eq:L5} are negative.
Substituting $\lambda = u +  \mathrm{i} v$, with $u, v \in \mathbb{R}$,
into \eqref{eq:L5}
we deduce that $u +  \mathrm{i} v$, with  $v \neq 0$,
is a root of \eqref{eq:L5} iff
$
v^2 
=
3 u^2 + 2 \left( 3 \gamma + \kappa  \right) u + 2   \gamma^2 C_\mathfrak{b} + 2 \kappa \gamma
$
and
\begin{equation}
\label{eq:L6}
\begin{aligned}
 0
& =
8 u^3 
+
\left( 24 \gamma + 8 \kappa \right) u^2
+
\left( 16 \gamma \kappa + 4 C_\mathfrak{b} \gamma^2 + 18 \gamma^2 + 2 \kappa^2 \right) u
\\
& \quad
+
2 \gamma^2 C_\mathfrak{b} \left( 3 \gamma - \kappa \right) + 2 \gamma \kappa^2 + 10 \gamma^2 \kappa .
\end{aligned}
\end{equation}
Applying Descartes' rule of signs we obtain that \eqref{eq:L6} has one positive root whenever 
\[
2 \gamma^2 C_\mathfrak{b} \left( 3 \gamma - \kappa \right) + 2 \gamma \kappa^2 + 10 \gamma^2 \kappa < 0 .
\]
Therefore,
all the roots of  \eqref{eq:L6} are strictly negative iff
$2 \gamma^2 C_\mathfrak{b} \left( 3 \gamma - \kappa \right) + 2 \gamma \kappa^2 + 10 \gamma^2 \kappa > 0$,
which is equivalent to 
 $\kappa > 3 \gamma $ and $ \kappa^2 +  5 \kappa \gamma >  \gamma \left( \kappa - 3 \gamma \right) C_\mathfrak{b}$.
In this case,
the real parts of all eigenvalues of J are less than $0$,
and so the nonzero equilibrium point of \eqref{eq:L4} is locally exponentially stable  
(see, e.g., Section 23.4  of \cite{Arnold1992}).
Therefore, 
there exist constants $\epsilon, \lambda, K > 0$ such that for all $t \geq 0$
\[
\max \left\{ 
  \left\vert r \left( t \right)  - r_0 \right\vert ,  \left\vert  S_R \left( t \right)  -  \kappa \, r_0 / g \right\vert,
 \left\vert  S_I  \left( t \right)  \right\vert, \left\vert  D_R  \left( t \right)  - d \left( 1 / C_\mathfrak{b} - 1 \right) \right\vert 
\right\}
\leq K \exp \left(  - \lambda t \right)
\]
whenever
$
\max \left\{ 
  \left\vert r \left( 0  \right)  - r_0 \right\vert ,  \left\vert  S_R \left(  0 \right)  -  \kappa \, r_0 / g \right\vert,
 \left\vert  S_I  \left( 0 \right)  \right\vert, \left\vert  D_R  \left( 0\right)  - d \left( 1 / C_\mathfrak{b} - 1 \right) \right\vert 
\right\}
< \epsilon 
$.
This implies
$
 \left\vert g \, S_I  \left( t \right)   / r \left( t \right)  \right\vert  \leq K \exp \left(  - \lambda t \right)
$.
Using \eqref{eq:L8} gives 
\[
\left\vert
\hbox{\rm e}^{  \mathrm{i}  \, \phi \left( t \right) }
-
\hbox{\rm e}^{ 
 \mathrm{i}  \left( \phi  \left( 0 \right) + g \int_0^{+ \infty}  \frac{S_I  \left( s \right) }{  r \left( s \right) } ds \right) }
 \right\vert  
\leq
 \left\vert \phi  \left( t \right) -  \phi  \left( 0 \right) -  g \int_0^{+ \infty}  \frac{S_I  \left( s \right) }{  r \left( s \right) } ds \right\vert 
\leq K \hbox{\rm e}^{  - \lambda t } .
\]
This leads to the assertion of the lemma.
\end{proof}

\begin{proof}[Proof of Theorem \ref{th:LimitCycle}]
Consider the the unitary transformation:
\begin{equation*}
\widetilde{ \rho }_t : =
\exp \left(  \mathrm{i}  \omega  \left( N +  \sigma^3 / 2 \right) t \right)
\rho_t
\, \exp \left( - \mathrm{i}  \omega  \left( N +  \sigma^3 / 2 \right) t \right)
\hspace{1cm}
\forall t \geq 0 .
\end{equation*}
From a careful computation we obtain that 
$\rho_t $ is a $N$-weak  solution  to  (\ref{eq:Laser1}) 
iff
$\widetilde{ \rho }_t $ is a $N$-weak solution to the non-linear equation \eqref{eq:Laser6},
i.e., equation (\ref{eq:Laser1}) with $\omega = 0$,
as  in the proof of Theorem \ref{th:FreeS-LaserE}.

Now, 
we shift the analysis from  (\ref{eq:Laser6}) to the linear  quantum master equation (\ref{eq:AuxiliarGKSL}),
in a similar way to that in the proof of Theorem \ref{th:LongTime}.
Set 
$ A \left( t \right) =  \Tr\left(  \widetilde{ \rho }_{t} \, a \right)$, 
$ S \left( t  \right)  = \Tr\left( \widetilde{ \rho }_{t} \, \sigma^{-}  \right)$ 
and 
$ D \left( t  \right) = \Tr\left(   \widetilde{ \rho }_{t}  \, \sigma^{3}\right)$
for all $t \geq 0$.
Then,
$ A \left( t \right) $, $ S \left( t  \right) $ and $ D \left( t  \right) $
satisfy (\ref{eq:Lorenz}) with  $\omega = 0$ (see, e.g., \cite{FagMora2019}),
and so 
$\left( \ \widetilde{ \rho }_{t} \right)_{t \geq 0}$ coincides with the unique $N$-weak solution to (\ref{eq:AuxiliarGKSL})
with  $\omega = 0$,
$\alpha \left( t \right) = g \,  S\left( t  \right)$,
$\beta \left( t \right) = g \,  A\left( t  \right)$, 
and initial datum $\rho_0$.
This leads us to study the long-time behavior of $\widetilde{ \rho }_{t} $ 
by means of Corollary \ref{cor:Ineq-Conv}.

Using Corollary \ref{cor:Ineq-Conv} with 
$\alpha \left( t \right) = g \,  S\left( t  \right)$,
$\beta \left( t \right) = g \,  A\left( t  \right)$, 
$
\alpha_0 
= 
\frac{1}{ \sqrt{2} \left| g \right| }
\left( z_{\infty} \kappa \gamma \sqrt{ C_\mathfrak{b} -1 } \right) 
$
and 
$
\beta_0
= 
\frac{1}{\sqrt{2} \left| g \right|}
\left( z_{\infty} g \gamma \sqrt{ C_\mathfrak{b} -1 } \right) 
$
we deduce that for all  $t  \geq s \geq 0$,
 \begin{align*}
 \Tr \left( \left\vert   \widetilde{ \rho }_t  -  \widetilde{ \varrho }_{\infty} \right\vert  \right)
& \leq 
\Tr \left( \left\vert   \rho^R_{t-s} \left(  \widetilde{ \rho }_s  \right)   -  \widetilde{ \varrho }_{\infty} \right\vert  \right)
+
4 \left\vert g \right\vert  \int_s^t \left\vert 
S \left( u \right) -  \frac{\kappa}{g}
\frac{\gamma \sqrt{ C_\mathfrak{b} - 1} }{\sqrt{2} \left\vert g \right\vert} z_{\infty}
\right\vert 
\sqrt{ \Tr \left(  \widetilde{ \rho }_u \, N  \right) +1 } \, du
\\
& \quad
+ 
2 \left\vert g \right\vert  \left( \left\Vert  \sigma^{-} \right\Vert +   \left\Vert  \sigma^{+} \right\Vert \right)
 \int_s^t \left\vert A \left( u \right) - \frac{\gamma \sqrt{ C_\mathfrak{b} - 1} }{\sqrt{2} \left\vert g \right\vert}  z_{\infty} \right\vert du ,
 \end{align*}
where 
$\rho^R_t \left( \cdot \right)$ is the one-parameter semigroup of contractions 
described by the $N$-weak solutions of  (\ref{eq:3.2g})
and 
\[
\widetilde{ \varrho }_{\infty}
=
\ketbra{ \mathcal{E} \left( \frac{ z_{\infty} \gamma \sqrt{ C_\mathfrak{b} -1 } }{ \sqrt{2} \left| g \right| }  \right)} 
{ \mathcal{E} \left( \frac{ z_{\infty} \gamma \sqrt{ C_\mathfrak{b} -1 }}{ \sqrt{2} \left| g \right| } \right)}
\otimes 
\begin{pmatrix}
 \frac{1}{2} \left( 1 + \frac{ d }{ C_\mathfrak{b} }  \right)
 &
 \frac{z_{\infty} \kappa \gamma }{ \sqrt{2} g \left| g \right|} \sqrt{ C_\mathfrak{b} -1 }
 \\
 \frac{ \overline{z_{\infty}} \kappa \gamma }{ \sqrt{2} g \left| g \right|} \sqrt{ C_\mathfrak{b} -1 }
 &
 \frac{1}{2} \left( 1 - \frac{d }{ C_\mathfrak{b} }  \right) 
\end{pmatrix} 
\]
with
$
z_{\infty}
=
\frac{ A  \left( 0 \right) }{  \left\vert  A  \left( 0 \right)  \right\vert } 
\exp \left(  \mathrm{i}  g \int_0^{+ \infty} \Im \left( \frac{S  \left( s \right) }{  A \left( s \right) } \right) ds \right)  
$.
Applying Lemma \ref{lem:LorenzBifurcation},
together with the upper bound of the term 
$ \Tr \left( \left\vert   \rho^R_{t-s} \left(  \widetilde{ \rho }_s  \right)   -  \widetilde{ \varrho }_{\infty} \right\vert  \right) $
provided by Theorem \ref{th:ConvAuxiliarQME},
we obtain that there exist constants $\epsilon, \lambda, K > 0$ such that for all $t \geq s \geq 0$:
\begin{equation}
\label{eq:LC1}
\begin{aligned}
 \Tr \left( \left\vert   \widetilde{ \rho }_t  -  \widetilde{ \varrho }_{\infty} \right\vert  \right)
& \leq 
12 \, \hbox{\rm e}^{- \gamma  \left( t - s \right) }    \left( 1 +  \left\vert d \right\vert \right)
+
\hbox{\rm e}^{- \kappa  \left( t - s \right)  }  \left( 
\frac{ \sqrt{2} \gamma \sqrt{ C_\mathfrak{b} - 1} }{  \left\vert  g  \right\vert }
+
4  \sqrt{ \Tr \left(  \widetilde{ \rho }_s  N  \right) } \right)
\\
& \quad +
K \left\vert g \right\vert  \int_s^t  \hbox{\rm e}^{- \lambda u } 
\sqrt{ \Tr \left(  \widetilde{ \rho }_u \, N  \right) +1 } \, du
+ 
K \left\vert g \right\vert  \left( \left\Vert  \sigma^{-} \right\Vert +   \left\Vert  \sigma^{+} \right\Vert \right)
 \int_s^t  \hbox{\rm e}^{- \lambda u } du 
 \end{aligned}
\end{equation}
in case 
\begin{equation}
\label{eq:LC2}
\max \left\{ 
  \left\vert    \left\vert A \left( 0 \right)  \right\vert  -  \frac{\gamma \sqrt{ C_\mathfrak{b} - 1} }{\sqrt{2} \left\vert g \right\vert}  \right\vert  ,  
 \left\vert  S \left(  0 \right)   -   \frac{ \kappa }{ g } \frac{\gamma \sqrt{ C_\mathfrak{b} - 1} }{\sqrt{2} \left\vert g \right\vert}  \frac{A \left( 0 \right) }{ \left\vert A \left( 0 \right)  \right\vert }\right\vert,
  \left\vert D \left( 0\right) -  \frac{d}{C_\mathfrak{b}} \right\vert 
\right\}
< \epsilon .
\end{equation}

Now, we examine the long-time behavior of $\Tr\left( \widetilde{ \rho }_{t} \, N  \right) $.
According to Lemma \ref{lem:EvolNumero} we have
\begin{align*}
 \Tr\left( \widetilde{ \rho }_{t} \, N  \right) 
& =
\hbox{\rm e}^{- 2  \kappa \, t} \Tr\left( \rho_{0} \,  N  \right) 
+
2 g \, \int_0^t \hbox{\rm e}^{- 2  \kappa \left( t - s \right)}  \Re \left(  \overline{ S \left( s \right)}  \, \Tr \left( \widetilde{ \rho }_s \, a  \right) \right) ds
\\
& =
\hbox{\rm e}^{- 2  \kappa \, t} \Tr\left( \rho_{0} \,  N  \right) 
+
2 g   \int_0^t \hbox{\rm e}^{- 2  \kappa \left( t - s \right)}  \Re \left(  \overline{ S \left( s \right)}  \,  A \left( s \right)  \right) ds .
\end{align*}
Since
\begin{align*}
 \left\vert
 \overline{ S \left( s \right)}  \,  A \left( s \right) 
 -
 \frac{\kappa}{g}
\frac{\gamma^2 \left(  C_{\mathfrak{b}} - 1 \right) } {2  \, g^2 }  
\right\vert
& \leq
\left\Vert \sigma^{-} \right\Vert \left\vert 
A \left( s \right) 
- 
\frac{\gamma \sqrt{ C_\mathfrak{b} - 1} }{\sqrt{2} \left\vert g \right\vert}  z_{\infty}
\right\vert
\\
& \quad 
+ 
\frac{\gamma \sqrt{ C_\mathfrak{b} - 1} }{\sqrt{2} \left\vert g \right\vert} 
\left\vert S \left( s \right)   
- 
\frac{\kappa}{g} \frac{\gamma \sqrt{ C_\mathfrak{b} - 1} }{\sqrt{2} \left\vert g \right\vert}  z_{\infty}
\right\vert ,
\end{align*}
from Lemma \ref{lem:LorenzBifurcation} we get 
\[
\left\vert 
\Tr\left( \widetilde{ \rho }_{t} \, N  \right) 
-
\frac{\gamma^2 \left(  C_{\mathfrak{b}} - 1 \right) } {2  \, g^2 } 
\right\vert
\leq
\hbox{\rm e}^{- 2  \kappa \, t} \Tr\left( \rho_{0} \,  N  \right)
+
2 g K  \int_0^t \hbox{\rm e}^{- 2  \kappa \left( t - s \right)} \hbox{\rm e}^{-  \lambda s }  ds 
\]
whenever \eqref{eq:LC2} holds.
As 
$
\Tr\left( \widetilde{ \varrho }_{\infty} \,  N  \right)
=
\gamma^2 \left(  C_{\mathfrak{b}} - 1 \right) / \left( 2  \, g^2 \right) 
$
we have 
\begin{equation}
\label{eq:LC3}
 \left\vert 
\Tr\left( \widetilde{ \rho }_{t} \, N  \right) 
-
\Tr\left( \widetilde{ \varrho }_{\infty} \,  N  \right)
\right\vert
\leq
 K\left( \Tr\left( \rho_{0} \, N  \right)  \right) \, \hbox{\rm e}^{- 2  \kappa \, t}
\hspace{2cm}
\forall t \geq 0 .
\end{equation}

We are in position to show the exponential convergence of $ \rho_t -  \varrho_c \left( \omega t - \theta_{\infty} \right)$ to $0$. 
According to \eqref{eq:LC3} we have
$
\Tr\left( \widetilde{ \rho }_{t} \, N  \right) 
\leq 
 K\left( \Tr\left( \rho_{0} \, N  \right)  \right) \, \hbox{\rm e}^{- 2  \kappa \, t}
 +
 \Tr\left( \widetilde{ \varrho }_{\infty} \,  N  \right)
$,
and so 
taking $s = t/ 2 $ in  \eqref{eq:LC1} we get 
\begin{equation}
\label{eq:LC7}
 \Tr \left( \left\vert   \widetilde{ \rho }_t  -  \widetilde{ \varrho }_{\infty} \right\vert  \right)
 \leq
 K\left( \Tr\left( \rho_{0} \, N  \right)  \right)  \exp \left(  - \lambda t \right) 
 \hspace{1cm}
\forall t \geq 0 ,
\end{equation}
where, by abuse of notation, we recall that $\lambda$ is a strictly positive constant
and $ K\left( \cdot \right) $ is a non-decreasing non-negative function.
Therefore, for all $t \geq 0$:
\begin{align*}
&
 \Tr \left( \left\vert  
 \rho_t 
 - 
 \exp \left( - \mathrm{i}  \omega  \left( N +  \sigma^3 / 2 \right) t \right) 
 \widetilde{ \varrho }_{\infty}
 \exp \left(  \mathrm{i}  \omega  \left( N +  \sigma^3 / 2 \right) t \right)
 \right\vert  \right)
 \\
 &
\leq 
 \left\Vert \exp \left( - \mathrm{i}  \omega  \left( N +  \sigma^3 / 2 \right) t \right) \right\Vert 
\left\Vert \exp \left( \mathrm{i}  \omega  \left( N +  \sigma^3 / 2 \right) t \right) \right\Vert 
  \Tr \left( \left\vert   \widetilde{ \rho }_t  -  \widetilde{ \varrho }_{\infty} \right\vert  \right)
%   \\
% &
 \leq
 K\left(  \Tr\left( \rho_{0} \, N  \right)  \right)\hbox{\rm e}^{  - \lambda t }.
\end{align*}
Similar to the final part of the proof of Theorem \ref{th:FreeS-LaserE}
we get 
\[
\exp \left( - \mathrm{i}  \omega  \left( N +  \sigma^3 / 2 \right) t \right) 
 \widetilde{ \varrho }_{\infty}
 \exp \left(  \mathrm{i}  \omega  \left( N +  \sigma^3 / 2 \right) t \right)
 =
 \varrho_c \left( \omega t - \theta_{\infty} \right) ,
\]
where $\varrho_c $ is defined by (\ref{eq:Def_orbita}).
This gives  \eqref{eq:LC4}.

Let $P_n$ be  the orthogonal projection of $\ell^2 \left(\mathbb{Z}_+ \right) \otimes \mathbb{C}^2 $ onto the linear span of 
$e_0 \otimes e_{\pm}, \ldots, e_n \otimes e_{\pm} $, i.e.,
$
P_n = \sum_{j=0}^n \ketbra{ e_j }{ e_j }
$.
By the triangular inequality,
\begin{align*}
  \left\vert  
  \Tr \left(   \rho_t \, A \right)
  -
   \Tr \left( \varrho_c \left( \omega t - \theta_{\infty} \right) \, A \right)
\right\vert  
& \leq
 \left\vert  
  \Tr \left(  \rho_t \, A \, P_n  \right)
  -
   \Tr \left(\varrho_c \left( \omega t - \theta_{\infty} \right) \, A \, P_n \right)
\right\vert
\\
& \quad
+
\left\vert  
  \Tr \left(  \rho_t \, P_n \, A  \left( I - P_n \right)  \right)
  -
   \Tr \left( \varrho_c \left( \omega t - \theta_{\infty} \right) \, P_n \, A   \left( I - P_n \right)   \right)
\right\vert
\\
&  \quad
\hspace{-2cm}
+
\left\vert  
  \Tr \left( \rho_t \left( I - P_n \right)  A  \left( I - P_n \right)   \right)
-
\Tr \left( \varrho_c \left( \omega t - \theta_{\infty} \right) \left( I - P_n \right)  A  \left( I - P_n \right)   \right)
\right\vert .
\end{align*}
Using \eqref{eq:LC6} and $ \left\Vert N P_n \right\Vert  =  n  $ gives 
$
\left\Vert P_n A \left( I - P_n \right) \right\Vert 
=
\left\Vert  \left( I - P_n \right) A^* P_n  \right\Vert 
\leq 
 \left( n +1 \right)  \widetilde{K}
$
and
$ \left\Vert A P_n \right\Vert \leq \left( n +1 \right)  \widetilde{K}$.
Hence,
\begin{align*}
 \left\vert  
  \Tr \left( \rho_t \,  A    \right)
  -
   \Tr \left( \varrho_c \left( \omega t - \theta_{\infty} \right) \, A  \right)
\right\vert  
& \leq
2 \left( n +1 \right) \widetilde{K} \,
 \Tr \left( \left\vert  
 \rho_t 
 - 
\varrho_c \left( \omega t - \theta_{\infty} \right)
 \right\vert  \right)
 \\
 &  \quad
\hspace{-2.3cm}
 +
\left\vert  
  \Tr \left(  \rho_t \left( I - P_n \right)  A  \left( I - P_n \right)   \right)
\right\vert
+
\left\vert  
  \Tr \left(  \varrho_c \left( \omega t - \theta_{\infty} \right) \left( I - P_n \right)  A  \left( I - P_n \right)   \right)
\right\vert ,
\end{align*}
and so \eqref{eq:LC4} yields  
\begin{align}
 \label{eq:LC9}
 \left\vert  
  \Tr \left( \rho_t \,  A    \right)
  -
   \Tr \left( \varrho_c \left( \omega t - \theta_{\infty} \right) \, A  \right)
\right\vert  
& \leq
2  \left( n +1 \right) \widetilde{K} \,
K \left( \Tr\left( \rho_{0} \, N  \right)  \right) \exp \left(  - \lambda t \right)
\\
\nonumber
&  \quad
\hspace{-3cm}
 +
\left\vert  
  \Tr \left(  \rho_t \left( I - P_n \right)  A  \left( I - P_n \right)   \right)
\right\vert
+
\left\vert  
  \Tr \left(  \varrho_c \left( \omega t - \theta_{\infty} \right) \left( I - P_n \right)  A  \left( I - P_n \right)   \right)
\right\vert .
\end{align}

Next,
we estimate the last two terms of the right-hand  side of (\ref{eq:LC9}).
Using \eqref{eq:LC6} we get 
\[
\left\vert
 \langle   \left( I - P_n \right)  x  ,   \left( I - P_n \right)  A  \left( I - P_n \right)  x  \rangle 
 \right\vert
\leq
4 \widetilde{K}   \left(  \left\Vert  \left( I - P_n \right)  x   \right\Vert^2 +  \langle  \left( I - P_n \right)  x  ,  N  \left( I - P_n \right)  x  \rangle \right)
\]
for all $x \in  \mathcal{D}\left( N \right) $
(see, e.g.,  the last part of the proof of Theorem \ref{th:LongTime}).
Hence
\begin{align*}
 \left\vert  
  \Tr \left(  \rho_t \left( I - P_n \right)  A  \left( I - P_n \right)   \right)
\right\vert
& \leq
4 \widetilde{K}   \left(
 \Tr \left(  \rho_t \left( I - P_n \right)   \right) +  \Tr \left(  \rho_t \, N \left( I - P_n \right)   \right)
 \right)
 \\
 & =
 4 \widetilde{K}   \left(
 \Tr \left(  \widetilde{ \rho }_t \left( I - P_n \right)   \right) +  \Tr \left(  \widetilde{ \rho }_t  \, N \left( I - P_n \right)   \right)
 \right)
\end{align*}
and
\begin{align*}
 \left\vert  
  \Tr \left(  \rho_t ^{\theta_{\infty}}  \left( I - P_n \right)  A  \left( I - P_n \right)   \right)
\right\vert
&
\leq
4 \widetilde{K}   \left(
 \Tr \left(  \rho_t ^{\theta_{\infty}} \left( I - P_n \right)   \right) +  \Tr \left(  \rho_t ^{\theta_{\infty}} \, N \left( I - P_n \right)   \right)
 \right)
 \\
 & =
 4 \widetilde{K}   \left(
 \Tr \left(  \widetilde{ \varrho }_{\infty} \left( I - P_n \right)   \right) +  \Tr \left(  \widetilde{ \varrho }_{\infty} \, N \left( I - P_n \right)   \right)
 \right)
\end{align*}
since  
$ \left( I - P_n \right) N    =  N \left( I - P_n \right)$.
Combining \eqref{eq:LC3} with
\[
 \left\vert  \Tr \left(  \widetilde{ \rho }_t N \left( I - P_n \right)   \right) \right\vert
\leq
\left\vert  \Tr \left(  \widetilde{ \rho }_t N  \right) -  \Tr \left(   \widetilde{ \varrho }_{\infty} N  \right) \right\vert
+
\left\vert   \Tr \left(   \widetilde{ \varrho }_{\infty} N   \left( I - P_n \right)  \right)  \right\vert
+
\left\vert   \Tr \left(  \left(  \widetilde{ \varrho }_{\infty} -  \widetilde{ \rho }_t  \right) N  P_n   \right)  \right\vert 
\]
we obtain
\[
\left\vert  \Tr \left(  \widetilde{ \rho }_t N \left( I - P_n \right)   \right) \right\vert
\leq
K\left( \Tr\left( \rho_{0} \, N  \right)  \right) \, \hbox{\rm e}^{- 2  \kappa \, t}
+
\left\vert   \Tr \left(   \widetilde{ \varrho }_{\infty} N   \left( I - P_n \right)  \right)  \right\vert
+
n  \,  \Tr \left(    \left\vert\widetilde{ \varrho }_{\infty} -  \widetilde{ \rho }_t    \right\vert \right)   .
\]
Moreover,
$
\left\vert  \Tr \left(  \widetilde{ \rho }_t \left( I - P_n \right)   \right) \right\vert
\leq
\left\vert  \Tr \left(  \widetilde{ \varrho }_{\infty}  \left( I - P_n \right)   \right) \right\vert
+
 \Tr \left(   \left\vert\widetilde{ \varrho }_{\infty} -  \widetilde{ \rho }_t   \right\vert \right)   
$
since
\[
\left\vert  \Tr \left(  \widetilde{ \rho }_t \left( I - P_n \right)   \right) \right\vert
=
\left\vert  1 -  \Tr \left(  \widetilde{ \varrho }_t \, P_n \right) \right\vert
\leq
\left\vert  \Tr \left(  \widetilde{ \varrho }_{\infty}  \left( I - P_n \right)   \right) \right\vert
+
\left\vert   \Tr \left(  \left(  \widetilde{ \varrho }_{\infty} -  \widetilde{ \rho }_t  \right)  P_n   \right)  \right\vert .
\]
Finally, applying \eqref{eq:LC7} we get
\begin{align}
\label{eq:LC8}
& 
\left\vert  
  \Tr \left(  \rho_t \left( I - P_n \right)  A  \left( I - P_n \right)   \right)
\right\vert
+
\left\vert  
  \Tr \left(  \varrho_c \left( \omega t - \theta_{\infty} \right) \left( I - P_n \right)  A  \left( I - P_n \right)   \right)
\right\vert
%\\
%& \leq
%4 \widetilde{K}   \left(
% \Tr \left(  \widetilde{ \rho }_t \left( I - P_n \right)   \right) +  \Tr \left(  \widetilde{ \rho }_t  \, N \left( I - P_n \right)   \right)
% \right)
% +
% 4 \widetilde{K}   \left(
% \Tr \left(  \widetilde{ \rho }_{\infty} \left( I - P_n \right)   \right) +  \Tr \left(  \widetilde{ \rho }_{\infty} \, N \left( I - P_n \right)   \right)
% \right)
 \\
 \nonumber
& \leq
 8 \widetilde{K} \,  \Tr \left(  \widetilde{ \varrho }_{\infty}  \left( I - P_n \right)   \right) 
+
8 \widetilde{K}  \,  \Tr \left(   \widetilde{ \varrho }_{\infty} N   \left( I - P_n \right)  \right) 
+
 \left( n +1 \right)  \widetilde{K} \, K\left( \Tr\left( \rho_{0} \, N  \right)  \right)   \exp \left(  - \lambda t \right) ,
\end{align}
where, for simplicity of notation, $\lambda$ is a strictly positive constant
and $ K\left( \cdot \right) $ is a non-decreasing non-negative function that does not depend on $A$.

From the definition of $\widetilde{ \varrho }_{\infty}$ we deduce that
\begin{align*}
  \Tr \left(   \widetilde{ \varrho }_{\infty} N   \left( I - P_n \right)  \right) 
 & =
 \frac{\gamma^2  \left( C_\mathfrak{b} -1 \right)}{ 2 g^2} 
 \exp \left(  - \frac{\gamma^2  \left( C_\mathfrak{b} -1 \right)}{ 2 g^2}  \right)
 \sum_{k=n}^{+ \infty} \frac{\left(  \frac{\gamma^2  \left( C_\mathfrak{b} -1 \right)}{ 2 g^2}  \right) ^k }{k !}
 \\
 & <
 \frac{\gamma^2  \left( C_\mathfrak{b} -1 \right)}{ 2 g^2} 
 \frac{\left(  \frac{\gamma^2  \left( C_\mathfrak{b} -1 \right)}{ 2 g^2}  \right) ^n }{n !}
\end{align*}
and
\[
  \Tr \left(   \widetilde{ \varrho }_{\infty}  \left( I - P_n \right)  \right) 
  =
 \exp \left(  - \frac{\gamma^2  \left( C_\mathfrak{b} -1 \right)}{ 2 g^2}  \right)
 \sum_{k=n+1}^{+ \infty} \frac{\left(  \frac{\gamma^2  \left( C_\mathfrak{b} -1 \right)}{ z2 g^2}  \right) ^k }{k !}
 <
 \frac{\left(  \frac{\gamma^2  \left( C_\mathfrak{b} -1 \right)}{ 2 g^2}  \right) ^{\left( n+1 \right)} }{ \left( n+1 \right) !} .
\]
Fix a natural number $\widetilde{n} \in \mathbb{N}$ satisfying  
$ \widetilde{n} >  \exp \left( \lambda \right)  \gamma^2  \left( C_\mathfrak{b} -1 \right)  / \left(  2 g^2 \right) $.
For any $n > \widetilde{n}$,
\begin{align*}
\Tr \left(   \widetilde{ \varrho }_{\infty} N   \left( I - P_n \right)  \right) 
+
\Tr \left(   \widetilde{ \varrho }_{\infty}  \left( I - P_n \right)  \right) 
& <
 \frac{\gamma^2  \left( C_\mathfrak{b} -1 \right)}{ g^2} 
 \frac{\left(  \frac{\gamma^2  \left( C_\mathfrak{b} -1 \right)}{ 2 g^2}  \right) ^n }{n !} 
 \\
& 
<
\frac{\gamma^2  \left( C_\mathfrak{b} -1 \right)}{ g^2} 
\frac{ \left(  \frac{\gamma^2  \left( C_\mathfrak{b} -1 \right)}{ 2 g^2}  \right)^{ \widetilde{n}} }{\widetilde{n} ! } 
\left(  \frac{\gamma^2  \left( C_\mathfrak{b} -1 \right)}{ 2 \, g^2 \, \widetilde{n} } \right)^{ \left( n - \widetilde{n}  \right) } .
\end{align*}
Using \eqref{eq:LC8} gives
\begin{align*}
&
\left\vert  
  \Tr \left(  \rho_t \left( I - P_n \right)  A  \left( I - P_n \right)   \right)
\right\vert
+
\left\vert  
  \Tr \left(  \varrho_c \left( \omega t - \theta_{\infty} \right) \left( I - P_n \right)  A  \left( I - P_n \right)   \right)
\right\vert
\\
&
<
\widetilde{K} \, K  \exp \left( - \left( n - \widetilde{n}  \right) \ln \left(  \frac{ 2 g^2 \widetilde{n} }{  \gamma^2  \left( C_\mathfrak{b} -1 \right) }  \right) \right)
+
\left( n +1 \right)  \widetilde{K} \, K \left( \Tr\left( \rho_{0} \, N  \right)  \right)   \exp \left(  - \lambda t \right) 
\end{align*}
for all  $n > \widetilde{n}$.
Now,
taking $n = t + \widetilde{n} $ in \eqref{eq:LC9} we get
\[
 \left\vert
 \Tr \left( \rho_t \,  A    \right)
  -
   \Tr \left( \varrho_c \left( \omega t - \theta_{\infty} \right) \, A  \right)
\right\vert  
 \leq
\left( t + \widetilde{n} +1 \right) 
\widetilde{K}  \, K \left( \Tr\left( \rho_{0} \, N  \right)  \right)  \hbox{\rm e}^{  - \lambda t }
+
\widetilde{K}  \, K \, \hbox{\rm e}^{  - t \ln \left(  \frac{ 2 \, g^2 \, \widetilde{n} }{  \gamma^2  \left( C_\mathfrak{b} -1 \right) }  \right) } .
\]
Since $ \widetilde{n} >  \exp \left( \lambda \right)  \gamma^2  \left( C_\mathfrak{b} -1 \right)  / \left(  2 g^2 \right) $,
$
 \left\vert
 \Tr \left( \rho_t \,  A    \right)
  -
   \Tr \left( \varrho_c \left( \omega t - \theta_{\infty} \right) \, A  \right)
\right\vert  
 \leq
\widetilde{K}  \, K \left( \Tr\left( \rho_{0} \, N  \right)  \right)  \hbox{\rm e}^{  -\frac{ \lambda}{2} t }
$ for all $ t \geq 0 $. 
\end{proof}

\subsection{Proof of Lemma \ref{lem:VecindadCL}}
\label{sec:Proof:VecindadCL}

\begin{proof}
According to (\ref{eq:DistLC}) we have that there exists $\vartheta \in \left[ 0, 2 \pi \right[$ such that
$ \mathfrak{d}_{N} \left( \varrho , \varrho_c \left( \vartheta \right) \right) < \varepsilon$.
Since $ a \, \mathcal{E} \left( \zeta \right)  =  \zeta \, \mathcal{E} \left( \zeta \right)$
for any  $\zeta \in \mathbb{C}$,
a direct computation gives
$
\Tr\left(  \varrho_c \left( \vartheta \right) \, a \right)   
=
\frac{ \gamma \sqrt{ C_\mathfrak{b} -1 } }{ \sqrt{2} \left| g \right|} 
\hbox{\rm e}^{- \mathrm{i} \vartheta} 
$,
$
\Tr\left(   \varrho_c \left( \vartheta \right) \, \sigma^{3} \right)
=
\frac{ d }{ C_\mathfrak{b} }
$,
and
$
\Tr\left(   \varrho_c \left( \vartheta \right) \sigma^{-}   \right)
=
\hbox{\rm e}^{- \mathrm{ i } \vartheta}  \frac{ \kappa \gamma }{ \sqrt{2} g \left| g \right|} \sqrt{ C_\mathfrak{b} -1 } 
$.
From Definition \ref{def_distancia} we now deduce that 
$
\left\vert  \Tr\left(   \varrho \, \sigma^{3} \right) -   \frac{ d }{ C_\mathfrak{b} }  \right\vert 
< \sqrt{2} \, \varepsilon
$,
$
\left\vert  \Tr\left(  \varrho \, \sigma^{-}  \right) 
- 
\frac{\kappa }{ g }  \frac{  \gamma \sqrt{ C_\mathfrak{b} -1 } }{ \sqrt{2}  \left| g \right|}  \hbox{\rm e}^{- \mathrm{ i } \vartheta}    \right\vert 
< \sqrt{2} \, \varepsilon 
$,
and
\begin{equation}
\label{eq:Lemma2.7.1}
\left\vert  \Tr\left(   \varrho \, a \right) -   
\frac{ \gamma \sqrt{ C_\mathfrak{b} -1 } }{ \sqrt{2} \left| g \right|} 
\hbox{\rm e}^{- \mathrm{i} \vartheta}    \right\vert 
< \sqrt{2} \, \varepsilon ,
\end{equation}
because 
$ \left\Vert \sigma^{-}  \right\Vert =  \left\Vert \sigma^{3}  \right\Vert = 1$
and 
$ \left\Vert a \, x \right\Vert^2 = \langle x , N \, x \rangle \leq \left\Vert N \, x \right\Vert^2 $ ,
$ \left\Vert a^\dagger \, x \right\Vert^2 = \langle x , (N+1) \, x \rangle  \leq  \left\Vert x \right\Vert^2 + \left\Vert N \, x \right\Vert^2 $
 for all $x \in \mathcal{D}\left( N \right) $.
Using (\ref{eq:Lemma2.7.1}) gives 
\[
\left\vert  \left\vert \Tr\left(   \varrho \, a \right)  \right\vert  -   
\frac{ \gamma \sqrt{ C_\mathfrak{b} -1 } }{ \sqrt{2} \left| g \right|} 
 \right\vert 
\leq 
\left\vert  \Tr\left(    \varrho \, a \right) -   
\frac{ \gamma \sqrt{ C_\mathfrak{b} -1 } }{ \sqrt{2} \left| g \right|} 
\hbox{\rm e}^{- \mathrm{i} \vartheta}    \right\vert 
< \sqrt{2} \, \varepsilon .
\]
As $ \Tr\left( \varrho \, a \right) \neq 0$,
$
\left\vert   \Tr\left(   \varrho \, a \right)   -   
\frac{ \gamma \sqrt{ C_\mathfrak{b} -1 } }{ \sqrt{2} \left| g \right|} \frac{ \Tr\left(   \varrho \, a \right) }{\left\vert \Tr\left(   \varrho \, a \right)  \right\vert}
 \right\vert 
 =
\left\vert  \left\vert \Tr\left(  \varrho \, a \right)  \right\vert  -   
\frac{ \gamma \sqrt{ C_\mathfrak{b} -1 } }{ \sqrt{2} \left| g \right|} 
 \right\vert 
< 
\sqrt{2} \, \varepsilon 
$,
and so  (\ref{eq:Lemma2.7.1}) yields 
$
\left\vert 
\frac{ \gamma \sqrt{ C_\mathfrak{b} -1 } }{ \sqrt{2} \left| g \right|} 
\hbox{\rm e}^{- \mathrm{i} \vartheta}  
-
 \frac{ \gamma \sqrt{ C_\mathfrak{b} -1 } }{ \sqrt{2} \left| g \right|} \frac{ \Tr\left(   \varrho \, a \right) }
 {\left\vert \Tr\left(  \varrho \, a \right)  \right\vert}
 \right\vert 
< 2 \, \sqrt{2} \, \varepsilon
$.
Applying the triangle inequality we get
\[
\left\vert  \Tr\left(   \varrho \, \sigma^{-} \right)    -   \frac{ \kappa }{ g } \frac{\gamma \sqrt{ C_\mathfrak{b} - 1} }{\sqrt{2} \left\vert g \right\vert}  
 \frac{  \Tr\left(  \varrho \, a \right) }{ \left\vert  \Tr\left(  \varrho  \, a \right)  \right\vert }\right\vert
 <
 \sqrt{2} \varepsilon  \left( 1 + 2 \, \kappa / g  \right) .
\]
On the other hand, 
applying the triangular inequality we obtain the second assertion of the lemma.
\end{proof}

\section*{Acknowledgements}
 The authors thank the two anonymous referees
 for suggestions and comments that improved the presentation.

%\bibliographystyle{siam.bst}
%
%\bibliography{StochasticSchrodinger}

\providecommand{\noopsort}[1]{}\providecommand{\singleletter}[1]{#1}%

\end{document}